\documentclass[sigconf,nonacm]{acmart}

\usepackage{graphicx}
\usepackage{balance}  
 
\usepackage{xspace,transparent}
\usepackage[multiple]{footmisc}
\usepackage{listings}
\usepackage{url}
\usepackage{float}
\usepackage{tabularx,colortbl}
\usepackage{ragged2e}
\usepackage[ruled,vlined]{algorithm2e}

\usepackage{wrapfig}
\usepackage{lipsum}
\usepackage{float}
\usepackage{graphicx,graphics}
\usepackage{subfig}
\usepackage[noend]{algpseudocode}
\usepackage{thmtools, thm-restate}
\usepackage{mathrsfs}
\usepackage{tikz}
\usepackage{comment}
\usepackage{makecell}
\usepackage{multirow}

\usetikzlibrary{positioning,calc}
\usetikzlibrary{decorations.text}
\usetikzlibrary{decorations.pathmorphing,decorations.pathreplacing}
\usetikzlibrary{arrows,petri, topaths,fit,arrows.meta}
\usepackage[inline]{enumitem}

\allowdisplaybreaks

\usepackage{arydshln}
\usepackage{subfig}

\usepackage{listings}
\newcommand{\system}{\textsf{LinCQA}}
\newcommand{\examplequery}{q^{\mathsf{ex}}}
\newcommand{\examplequerynb}{q^{\mathsf{nex}}}

\newcommand{\cforest}{\mathcal{C}_{\mathsf{forest}}}

\newcounter{reviewcounter}

\newcommand{\inlineitem}[1][]{%
\ifnum\enit@type=\tw@
    {\descriptionlabel{#1}}
  \hspace{\labelsep}%
\else
  \ifnum\enit@type=\z@
       \refstepcounter{\@listctr}\fi
    \quad\@itemlabel\hspace{\labelsep}%
\fi}
\makeatother

\newtheorem{definition}{Definition}
\newtheorem{rules}{Rule}

\newcommand{\cut}[1]{}

\newenvironment{packed_item}{
\begin{itemize}
   \setlength{\itemsep}{1pt}
   \setlength{\parskip}{0pt}
   \setlength{\parsep}{0pt}
}
{\end{itemize}}

\newenvironment{packed_enum}{
\begin{enumerate}
   \setlength{\itemsep}{1pt}
  \setlength{\parskip}{0pt}
   \setlength{\parsep}{0pt}
}
{\end{enumerate}}

\newcommand{\introparagraph}[1]{\vspace{1mm} \noindent \textbf{#1 }}
\newcommand{\primarykey}[1]{\textbf{#1}}
\newcommand{\relation}[1]{\textsf{#1}}
\newcommand{\database}[1]{\textbf{\textsf{#1}}}
\newcommand{\obtainedfrom}{\text{ :- } }

\newcommand{\PreserveBackslash}[1]{\let\temp=\\#1\let\\=\temp}
\newcolumntype{C}[1]{>{\PreserveBackslash\centering}p{#1}}
\newcolumntype{R}[1]{>{\PreserveBackslash\raggedleft}p{#1}}
\newcolumntype{L}[1]{>{\PreserveBackslash\raggedright}p{#1}}

\newcommand{\FO}{$\mathbf{FO}$}

\newcommand{\LSPACE}{$\mathbf{L}$}

\newcommand{\coNP}{$\mathbf{coNP}$}
\newcommand{\NP}{$\mathbf{NP}$}

\newcommand{\card}[1]{|{#1}|}

\newcommand{\cqa}[1]{\mathsf{CERTAINTY}({#1})}

\mathchardef\mhyphen="2D

\newcommand{\db}{\mathbf{db}}
\newcommand{\block}{\mathbf{b}}

\newcommand{\var}[1]{\mathsf{vars}({#1})}
\newcommand{\queryvars}[1]{{\mathsf{vars}}({#1})}
\newcommand{\atomvars}[1]{{\mathbf{vars}}({#1})}

\newcommand{\fd}[2]{{#1}\rightarrow{#2}}
\newcommand{\attacks}[1]{\stackrel{#1}{\leadsto}}
\newcommand{\nattacks}[1]{\stackrel{#1}{\not\leadsto}}
\newcommand{\keyclosure}[2]{{#1}^{+,{#2}}}

\newcommand{\goodkey}[3]{{#1}_{\mathsf{gkey}}({#2}, {#3})}

\newcommand{\fk}[1]{{#1}_{\mathsf{fkey}}}
\newcommand{\jk}[1]{{#1}_{\mathsf{join}}}

\newcommand{\key}[1]{{\mathsf{key}}({#1})}

\newcommand{\revm}[1]{{\color{black} #1}}
\newcommand{\reva}[1]{{\color{black}  #1}}
\newcommand{\revb}[1]{{\color{black}   #1}}
\newcommand{\revc}[1]{{\color{black}    #1}}
\newcommand{\revjef}[1]{{\color{black}    #1}}

\definecolor{purple}{rgb}{0.65,0.05,0.3}
\definecolor{red}{rgb}{1,0,0}
\definecolor{RED}{rgb}{1,0,0}

\AtBeginDocument{%
  \providecommand\BibTeX{{%
    \normalfont B\kern-0.5em{\scshape i\kern-0.25em b}\kern-0.8em\TeX}}}

\setcopyright{acmcopyright}
\copyrightyear{2018}
\acmYear{2018}
\acmDOI{XXXXXXX.XXXXXXX}

\acmConference[Conference acronym 'XX]{Make sure to enter the correct
  conference title from your rights confirmation emai}{June 03--05,
  2018}{Woodstock, NY}
\acmPrice{15.00}
\acmISBN{978-1-4503-XXXX-X/18/06}

\begin{document}
\title[LinCQA: Faster Consistent Query Answering with Linear Time Guarantees]{LinCQA: Faster Consistent Query Answering \\ with Linear Time Guarantees}

\author{Zhiwei Fan}
  \affiliation{%
    \institution{University of Wisconsin--Madison}
    \city{Madison, WI}
    \country{USA}
    }
  \email{zhiwei@cs.wisc.edu}

\author{Paraschos Koutris}
  \affiliation{%
    \institution{University of Wisconsin--Madison}
    \city{Madison, WI}
    \country{USA}
    }
  \email{paris@cs.wisc.edu}

  \author{Xiating Ouyang}
  \affiliation{%
    \institution{University of Wisconsin--Madison}
    \city{Madison, WI}
    \country{USA}
    }
  \email{xouyang@cs.wisc.edu}

  \author{Jef Wijsen}
  \affiliation{%
    \institution{University of Mons}
    \city{Mons}
    \country{Belgium}
    }
  \email{jef.wijsen@umons.ac.be}

\renewcommand{\shortauthors}{Zhiwei Fan, Paraschos Koutris, Xiating Ouyang and Jef Wijsen}

\begin{abstract}

Most data analytical pipelines often encounter the problem of querying inconsistent data that violate pre-determined integrity constraints. 
Data cleaning is an extensively studied paradigm that singles out a consistent repair of the inconsistent data. Consistent query answering (CQA) is an alternative approach to data cleaning that asks for all tuples guaranteed to be returned by a given query on all (in most cases, exponentially many) repairs of the inconsistent data. This paper identifies a class of acyclic select-project-join (SPJ) queries for which CQA can be solved via SQL rewriting with a linear time guarantee. Our rewriting method can be viewed as a generalization of Yannakakis's algorithm for acyclic joins to the inconsistent setting. We present \system, a system that can output rewritings in both SQL and non-recursive Datalog rules for every query in this class. We show that \system\ often outperforms the existing CQA systems on both synthetic and real-world workloads, and in some cases, by orders of magnitude.

\end{abstract}

\maketitle

\section{Introduction}
\label{sec:intro}

A database is {\em inconsistent} if it violates one or more integrity constraints that are supposed to be satisfied. 
Database inconsistency can naturally occur when the dataset results from an integration of heterogeneous sources, or because of noise during data collection.

{\em Data cleaning} \cite{rahm2000data} is the most widely used approach to manage inconsistent data in practice. It first {\em repairs} the inconsistent database by removing or modifying the inconsistent records so as to obey the integrity constraints. Then, users can run queries on a {\em clean} database.
There has been a long line of research on data cleaning. Several frameworks have been proposed \cite{DBLP:conf/icde/GeGM0W21,DBLP:journals/pvldb/RezigOAEMS21,DBLP:journals/pvldb/GeertsMPS13,DBLP:conf/sigmod/ArasuK09,GrecoGZ03}, using techniques such as knowledge bases and machine learning~\cite{DBLP:journals/pvldb/RekatsinasCIR17,DBLP:conf/sigmod/ChuIKW16,DBLP:journals/mst/BertossiKL13,DBLP:journals/pvldb/HeCGZNC18,DBLP:journals/pvldb/EbaidEIOQ0Y13,DBLP:conf/icde/LiRBZCZ21,DBLP:conf/sigmod/BergmanMNT15,DBLP:conf/icde/ChuIP13,DBLP:conf/icde/TongCZLC14,DBLP:journals/pvldb/KrishnanWWFG16}. Data cleaning has also been studied under different contexts \cite{kohler2021possibilistic,DBLP:journals/pvldb/ChengCX08,bertossi2013data,DBLP:conf/sigmod/KhayyatIJMOPQ0Y15,DBLP:conf/icde/BohannonFGJK07,DBLP:journals/pvldb/ProkoshynaSCMS15}. 
However, the process of data cleaning is often ad hoc and arbitrary choices are frequently made regarding which data to keep in order to restore database consistency. 
This comes at the price of losing important information since the number of cleaned versions of the database can be exponential in the database size. Moreover, data cleaning is commonly seen as a laborious and time-intensive process in data analysis.
There have been efforts to accelerate the data cleaning process~\cite{DBLP:journals/pvldb/RekatsinasCIR17,DBLP:conf/icde/ChuIP13,DBLP:conf/sigmod/ChuMIOP0Y15,DBLP:journals/pvldb/RezigOAEMS21},
but in most cases, users need to wait until the data is clean before being able to query the database. 



\begin{table*}
\caption{A summary of systems for consistent query answering}
\vspace{-2ex}
\label{tbl:all-systems}
\begin{tabular}{c | c c c c}
System & Target class of queries & Intermediate output & Backend \\
\hline
EQUIP \cite{10.14778/2536336.2536341} & all SPJ Queries &  Big Integer Program (BIP)  & DBMS \& BIP solver \\
CAvSAT \cite{DBLP:journals/corr/abs-2103-03314,DBLP:conf/sat/DixitK19} & all SPJ Queries &  SAT formula & DBMS \& SAT solver \\
\hdashline
Conquer \cite{fuxman2005conquer} & $\cforest$ &  SQL rewriting & DBMS \\
Conquesto \cite{Conquesto} & self-join-free SPJ Queries in \FO & Datalog rewriting & Datalog engine \\
\system\ (this paper) & PPJT &  SQL rewriting / Datalog rewriting & DBMS or Datalog engine
\end{tabular}
\vspace{-1ex}
\end{table*}

{\em Consistent query answering} (CQA) is an alternative approach to data cleaning for managing inconsistent data~\cite{10.1145/303976.303983} that has recently received more attention~\cite{10.1145/3377391.3377393,Bertossi19}.  Instead of singling out the ``best'' repair, CQA considers \emph{all} possible repairs of the inconsistent database, returning the intersection of the query answers over all repairs, called the \emph{consistent answers}. 
CQA serves as a viable complementary procedure to data cleaning for multiple reasons. First, it deals with inconsistent data at query time without needing an expensive offline cleaning process during which the users cannot query the database. Thus, users can quickly perform preliminary data analysis to obtain the consistent answers while waiting for the cleaned version of the database.
Second, consistent answers can also be returned alongside the answers obtained after data cleaning, by marking which answers are certainly/reliably correct and which are not. This information may provide further guidance in critical decision-making data analysis tasks. Third, CQA can be used to design more efficient data cleaning algorithms \cite{DBLP:journals/pvldb/KarlasLWGC0020}.


In this paper, we will focus on CQA for the most common kind of integrity constraint: {\em primary keys}. A primary key constraint enforces that no two distinct tuples in the same table agree on all primary key attributes. CQA under primary key constraints has been extensively studied over the last two decades. 

\revb{From a theoretical perspective, CQA for select-project-join (SPJ) queries is computationally hard as it potentially requires inspecting an exponential number of repairs.} However, for some SPJ queries the consistent answers can be computed in polynomial time, and for some other SPJ queries CQA is \emph{first-order rewritable} (\FO-rewritable): we can construct another query such that executing it directly on the inconsistent database will return the consistent answers of the original query. 
After a long line of research \cite{KOLAITIS201277,KoutrisS14,KoutrisW15,DBLP:conf/icdt/KoutrisW19,KoutrisWTOCS20}, it was proven that given any self-join-free SPJ query, the problem is either \FO-rewritable, polynomial-time solvable but not \FO-rewritable, or \coNP-complete \cite{KoutrisW17}. 

From a systems standpoint, most CQA systems fall into two categories (summarized in Table~\ref{tbl:all-systems}):
(1)~systems that can compute the consistent answers of join queries with arbitrary denial constraints but require solvers for computationally hard problems (e.g., EQUIP~\cite{10.14778/2536336.2536341} relies on Integer Programming solvers, and CAvSAT~\cite{DBLP:journals/corr/abs-2103-03314,DBLP:conf/sat/DixitK19} requires SAT solvers), and  (2)~systems that output the \FO-rewriting of the input query, but only target a specific class of queries that occurs frequently in practice. Fuxman and Miller~\cite{fuxman2005conquer} identified a  class of \FO-rewritable  queries called $\cforest$ and implemented their rewriting in ConQuer, which outputs a single SQL query. Conquesto~\cite{Conquesto} is the most recent system targeting \FO-rewritable join queries by producing the rewriting in Datalog. 

We identify several drawbacks with all systems above. Both EQUIP and CAvSAT rely on solvers for \NP-complete problems, which does not guarantee efficient termination, even if the input query is \FO-rewritable. Even though $\cforest$ captures many join queries seen in practice, it excludes queries that involve \emph{(i)}~joining with only part of a composite primary key, often appearing in snowflake schemas, and \emph{(ii)}~joining two tables on both primary-key and non-primary-key attributes, 
which commonly occur in settings such as entity matching and cross-comparison scenarios. 
Conquesto, on the other hand, implements the generic \FO-rewriting algorithm without strong performance guarantees. Moreover, neither ConQuer nor Conquesto has theoretical guarantees on the running time of their produced rewritings. 

\introparagraph{Contributions.} To address the above observed issues, we make the following contributions: 

\smallskip
\noindent {\em Theory \& Algorithms.} We identify a subclass of acyclic Boolean join queries that captures a wide range of queries commonly seen in practice for which we can produce \FO-rewritings with a {\em linear running time guarantee} (Section~\ref{sec:rewriting}). This class subsumes all acyclic Boolean queries in $\cforest$. For consistent databases, Yannakakis's algorithm~\cite{DBLP:journals/jacm/BeeriFMY83} evaluates acyclic Boolean join queries in linear time in the size of the database. Our algorithm shows that even when inconsistency is introduced w.r.t. primary key constraints, the consistent answers of many acyclic Boolean join queries can still be computed in linear time, exhibiting no overhead to Yannakakis's algorithm. 
Our technical treatment follows Yannakakis's algorithm by considering a rooted join tree with an additional annotation of the \FO-rewritability property, called a \emph{pair-pruning join tree (PPJT)}. Our algorithm follows the pair-pruning join tree to compute the consistent answers and degenerates to Yannakakis's algorithm if the database has no inconsistencies. 

\smallskip
\noindent {\em Implementation.} We implement our algorithm in \system\ (\underline{Lin}ear \underline{C}onsistent \underline{Q}uery \underline{A}nswering) \footnote{\url{https://github.com/xiatingouyang/LinCQA/}}, a system prototype that produces an efficient and optimized rewriting in both SQL and non-recursive Datalog rules with negation (Section~\ref{sec:implementation}).

\smallskip
\noindent {\em Evaluation.} We perform an extensive experimental evaluation comparing \system\ to the other state-of-the-art CQA systems. Our findings show that \emph{(i)}~a properly implemented rewriting can significantly outperform a generic CQA system (e.g., CAvSAT); \emph{(ii)}~\system\ achieves the best overall performance throughout all our experiments under different inconsistency scenarios; and \emph{(iii)}~the strong theoretical guarantees of \system\ translate to a significant performance gap for worst-case database instances. \system\ often outperforms other CQA systems, in several cases by orders of magnitude on both synthetic and real-world workloads. We also demonstrate that CQA can be an effective approach even for real-world datasets of very large scale ($\sim$$400$GB), which, to the best of our knowledge, have not been tested before.

\section{Related Work}
\label{sec:related-work}

Inconsistency in databases has been studied in different contexts~\cite{10.5555/1709465.1709573,kahale2020meta,katsis2010inconsistency,DBLP:journals/jcss/BarceloF17,DBLP:conf/icdt/BarceloF15,DBLP:conf/icdt/LopatenkoB07,DBLP:journals/is/RodriguezBM13,DBLP:conf/pods/CalauttiCP21}.  The notion of  Consistent Query Answering (CQA) was introduced in the seminal work by Arenas, Bertossi, and Chomicki \cite{10.1145/303976.303983}. After twenty years, their contribution was acknowledged in a \emph{Gems of PODS session}~\cite{Bertossi19}. An overview of complexity classification results in CQA appeared recently in the \emph{Database Principles} column of SIGMOD Record \cite{10.1145/3377391.3377393}.

The term $\cqa{q}$ was coined in~\cite{wijsen2010first} to refer to CQA for Boolean queries $q$ on databases that violate primary keys, one per relation, which are fixed by $q$'s schema. The complexity classification of $\cqa{q}$ for the class of self-join-free
Boolean conjunctive queries started with the work by Fuxman and Miller~\cite{FUXMAN2007610}, and was further pursued in~\cite{KOLAITIS201277,KoutrisS14,KoutrisW15,KoutrisW17,DBLP:conf/icdt/KoutrisW19,KoutrisWTOCS20},
which eventually revealed that the complexity of $\cqa{q}$ for self-join-free conjunctive queries displays a trichotomy between \FO,  \LSPACE-complete, and \coNP-complete. A recent result also extends the complexity classification of $\cqa{q}$ to path queries that may contain self-joins \cite{DBLP:conf/pods/KoutrisOW21}. 
The complexity of $\cqa{q}$ for self-join-free Boolean conjunctive queries with negated atoms was studied in~\cite{KoutrisW18}. 
For self-join-free Boolean conjunctive queries w.r.t.\ multiple keys, it remains decidable whether or not $\cqa{q}$ is in \FO~\cite{KoutrisW20}.

Several systems for CQA that are used for comparison in our study have already been described in the introduction: ConQuer~\cite{fuxman2005conquer}, Conquesto~\cite{Conquesto}, CAvSAT~\cite{DBLP:journals/corr/abs-2103-03314,DBLP:conf/sat/DixitK19}, and EQUIP~\cite{10.14778/2536336.2536341}. Most early systems for CQA used efficient solvers for Disjunctive Logic Programming and Answer Set Programming (ASP) \cite{chomicki2004hippo,GrecoGZ03,manna_ricca_terracina_2015,DBLP:journals/tplp/ArenasBC03,DBLP:conf/sccc/MarileoB05,DBLP:conf/icdt/LopatenkoB07}. 
  
Similar notions to CQA are also emerging in machine learning with the goal of computing the consistent classification result of certain machine learning models over inconsistent training data~\cite{DBLP:journals/pvldb/KarlasLWGC0020}.
\setlength{\abovedisplayskip}{3pt}
\setlength{\belowdisplayskip}{3pt}

\section{Background}
\label{sec:prelim}

In this section, we define some notations used in our paper. We use the example \textbf{Company} database shown in Figure~\ref{fig:incon_db_example} to illustrate our constructs, where the primary key attribute of each table is highlighted in bold.

\definecolor{Gray}{gray}{0.9}

\begin{figure*}
\begin{tabular}{c c c}
\begin{tabular}{|l l l}
 & \textsf{Employee} & \\
\primarykey{\textsf{employee\_id}} & \textsf{office\_city} & \textsf{wfh\_city} \\
\hline
0011 & Boston &  Boston \\
\rowcolor{Gray}
0011 & Chicago & New York \\
0011 & Chicago & Chicago \\
\hdashline
\rowcolor{Gray}
 0022 & New York & New York \\
 0022 & Chicago &   Chicago  \\
\hdashline
\rowcolor{Gray}
 0034 & Boston & New York
\end{tabular}
&
\begin{tabular}{|l c c}
 & \textsf{Manager} & \\
\primarykey{\textsf{office\_city}} & \textsf{manager\_id} & \textsf{start\_year} \\
\hline
 Boston & 0011 & 2020 \\
\rowcolor{Gray} Boston & 0011 & 2021 \\
\hdashline
\rowcolor{Gray}  Chicago & 0022 & 2020 \\
\hdashline
\rowcolor{Gray} LA & 0034 & 2020 \\
 LA & 0037 & 2020 \\
\hdashline
\rowcolor{Gray}New York & 0022 & 2020 
\end{tabular}
&
\begin{tabular}{|l c}
\textsf{Contact} & \\
\primarykey{\textsf{office\_city}} & \textsf{contact\_id} \\
\hline
\rowcolor{Gray}Boston & 0011\\
 Boston & 0022 \\
\hdashline
\rowcolor{Gray} Chicago &  0022 \\
\hdashline
\rowcolor{Gray} LA & 0034  \\
 LA & 0037 \\
\hdashline     
\rowcolor{Gray}New York & 0022 
\end{tabular}
\end{tabular}
\caption{An inconsistent database (\database{Company}). \revb{A repair of this database is highlighted with the gray color.}}
\label{fig:incon_db_example}
\end{figure*}

\introparagraph{Database instances, blocks, and repairs.}
A \emph{database schema} is a finite set of table names.
Each table name is associated with a finite sequence of \emph{attributes}, \reva{and the length of that sequence is called the \emph{arity} of that table}. 
Some of these attributes are declared as \emph{primary-key attributes}, forming together the \emph{primary key}.
A \emph{database instance} $\db$ associates to each table name a finite set of \emph{tuples} \reva{that agree on the arity of the table}, called a \emph{relation}.
A relation is \emph{consistent} if it does not contain two distinct tuples that agree on all primary-key attributes. 
A \emph{block} of a relation is a maximal set of tuples that agree on all primary-key attributes.
Thus, a relation is consistent if and only if it has no block with two or more tuples.
A \emph{repair} of a (possibly inconsistent) relation is obtained by selecting exactly one tuple from each block.
Clearly, a relation with $n$ blocks of size~$2$ each has $2^{n}$ repairs, an exponential number. 
A database instance $\db$ is consistent if all relations in it are consistent.
A repair of a (possibly inconsistent) database instance is obtained by selecting one repair for each relation. 
In the technical treatment, it will be convenient to view a database instance $\db$ as a set of facts: if the relation associated with table name $R$ contains tuple $\vec{t}$, then we say that $R(\vec{t})$ is a fact of $\db$. 

\begin{example}
The \textbf{Company} database in Figure~\ref{fig:incon_db_example} is inconsistent with respect to primary key constraints. For example, in the \relation{Employee} table there are 3 distinct tuples sharing the same primary key \relation{employee\_id} $0011$. The blocks in the \textbf{Company} database are highlighted using dashed lines. An example repair of the $\textbf{Company}$ database can be obtained by {\em choosing exactly one tuple from each block}, and there are in total $96=3\times 2^{5}$ distinct repairs.  \qed
\end{example}

\introparagraph{Atoms and key-equal facts.}
Let $\vec{x}$ be a sequence of variables and constants. We write $\var{\vec{x}}$ for the set of variables that appear in $\vec{x}$. 
An \emph{atom} $F$ with relation name~$R$ takes the form $R(\underline{\vec{x}},\vec{y})$, where the primary key is underlined; we denote $\key{F}=\var{\vec{x}}$.
Whenever a database instance $\db$ is understood, we write $R(\underline{\vec{c}},*)$ for the block containing all tuples with primary-key value~$\vec{c}$ in relation $R$.
\begin{example}
For the \textbf{Company} database, we can have atoms $\relation{Employee}(\underline{x}, y, y)$, $\relation{Manager}(\underline{u}, v, \texttt{2020})$, and $\relation{Contact}(\underline{\texttt{LA}}, \texttt{2020})$. The block $\relation{Manager}
(\underline{\relation{Boston}}, *)$ contains two facts: \relation{Manager}(\underline{\relation{Boston}}, \relation{0011}, \relation{2020}) and \relation{Manager}(\underline{\relation{Boston}}, \relation{0011}, \relation{2021}). \qed
\end{example}

\introparagraph{Conjunctive Queries.}
For select-project-join (SPJ) queries, we will also use the term \emph{conjunctive queries (CQ)}.
Each CQ $q$ can be represented as a succinct rule of the following form:
%
\begin{align}\label{eq:rule}
q(\vec{u}) \obtainedfrom R_{1}(\underline{\vec{{x}_{1}}},\vec{{y}_{1}}), \dots, R_{n}(\underline{\vec{{x}_{n}}},\vec{{y}_{n}})
\end{align}
where each $R_{i}(\underline{\vec{{x}_{i}}},\vec{{y}_{i}})$ is an atom for $1 \leq i \leq n$. We denote by $\queryvars{q}$ the set of variables that occur in $q$ and $\vec{u}$ is said to be the \emph{free variables} of $q$. The atom $q(\vec{u})$ is the \emph{head} of the rule, and the remaining atoms are called the \emph{body} of the rule, $\mathsf{body}(q)$.

A CQ $q$ is Boolean (BCQ) if it has no free variables, and it is {\em full} if all its variables are free.
We say that $q$ has a {\em self-join\/} if some relation name occurs more than once in $q$.
A CQ without self-joins is called {\em self-join-free\/}. 
If a self-join-free query $q$ is understood, an atom $R(\underline{\vec{x}},\vec{y})$ in $q$ can be denoted by~$R$.
If the body of a CQ of the form~\eqref{eq:rule} can be partitioned into two nonempty parts that have no variable in common, then we say that the query is \emph{disconnected}; otherwise it is \emph{connected}.

For a CQ $q$, let $\vec{x} = \langle x_1, \dots , x_{\ell} \rangle$ be a sequence of distinct variables that occur in $q$ and
$\vec{a} = \langle a_1, \dots , a_{\ell}\rangle$ be a sequence of constants, then $q_{[\vec{x} \rightarrow \vec{a}]}$ denotes the query obtained from $q$ by replacing all occurrences of $x_i$ with $a_i$ for all $1 \leq i \leq \ell$.

\vspace{-1ex}
\begin{example}
\label{ex:dbexample}
Consider the query over the \textbf{Company} database that {\em returns the id's of all employees who work in some office city with a manager who started in year 2020.} It can be expressed by the following SQL query:
\begin{lstlisting}
SELECT E.employee_id FROM Employee E, Manager M
WHERE E.office_city=M.office_city AND M.start_year=2020
\end{lstlisting} 
\vspace{-1ex}
and the following CQ: 
$$q(x) \obtainedfrom \relation{Employee}(\underline{x}, y, z), \relation{Manager}(\underline{y}, w, \texttt{2020}).$$

The following CQ $q'$ is a BCQ, since it merely asks whether $0011$ is such an employee\_id satisfying the conditions in $q$:
$$q'() \obtainedfrom \relation{Employee}(\underline{0011}, y, z), \relation{Manager}(\underline{y}, w, \texttt{2020}).$$ It is easy to see that $q'$ is equivalent to $q_{[x \rightarrow 0011]}$. \qed
\end{example}



\introparagraph{Datalog.} A Datalog program $P$ is a finite set of rules of the form~\eqref{eq:rule},
with the extension that negated atoms can be used in rule bodies. A rule can be interpreted as a logical implication: if the body is true, then so is the head of the rule. We assume that rules are always safe, meaning that every variable occurring in the rule must also occur in a non-negated atom of the rule body.
\revb{A relation belongs to the intensional database (IDB) if it is defined by rules, i.e., if it appears as the head of some rule; otherwise it belongs to the extensional database (EBD), i.e., it is a stored table.}
Our rewriting uses non-recursive Datalog with negation \cite{DBLP:journals/jcss/AjtaiG94}. This means that the rules of a Datalog program $P$ can be partitioned into $(P_{1},P_{2},\dots,P_{n})$ such that the rule body of a rule in $P_{i}$ uses only IDB predicates defined by rules in some $P_{j}$ with $j<i$.
Here, it is understood that all rules with the same head predicate belong in the same partition.  


\introparagraph{Consistent query answering.}
For every CQ $q$, given an input database instance $\db$, \revb{the problem $\cqa{q}$ asks for the intersection of query outputs over all repairs of $\db$}. 
If $q$ is Boolean, the problem $\cqa{q}$ then asks whether $q$ is satisfied by every repair of the input database instance~$\db$. \revb{In this work, we study the {\em data complexity} of $\cqa{q}$, i.e., the size of the query $q$ is assumed to be a fixed constant.}

The problem $\cqa{q}$ has a first-order rewriting (\FO-rewriting) if there is another first-order query~$q'$ (which, in most cases, uses the difference operator and hence is not a SPJ query) 
such that evaluating $q'$ on the input database $\db$ would return the answers of $\cqa{q}$. In other words, executing $q'$ directly on the inconsistent database \emph{simulates} computing the original query $q$ over all possible repairs.
\begin{example}
\label{ex:cqa-prelim}
Recall that in Example~\ref{ex:dbexample}, the query $q$ returns $\{0011, 0022, 0034\}$ on the inconsistent database \textbf{Company}. For $\cqa{q}$ however, the only output is $0022$: for any repair that contains the tuples \relation{Employee}(\underline{0011}, \relation{Boston}, \relation{Boston}) and \relation{Manager}(\underline{\relation{Boston}}, 0011, 2021), neither $0011$ nor $0034$ would be returned by $q$; and in any repair, $0022$ is returned by $q$ with the following crucial observation:
\emph{
Regardless of which tuple in $\relation{Employee}(\underline{0022}, *)$ the repair contains, both offices are present in the \relation{Manager} table and both managers in Chicago and New York offices started in 2020.
}

Based on the observation, it is sufficient to solve $\cqa{q}$ 
by running the following single SQL query, called an \FO-rewriting of $\cqa{q}$. 
\lstset{aboveskip=1.5pt,belowskip=1.5pt}
 \begin{lstlisting}
		SELECT E.employee_id  FROM Employee E EXCEPT
		SELECT E.employee_id  FROM Employee E 
		WHERE E.office_city NOT IN (
				SELECT M.office_city FROM Manager EXCEPT
				SELECT M.office_city FROM Manager 
				WHERE M.start_year <> 2020)
\end{lstlisting}
\vspace{-3ex}
\qed 
\end{example}

\revm{

\introparagraph{Acyclic queries and join trees}
Let $q$ be a CQ. A {\em join tree} of $q$ is an undirected tree whose nodes are the atoms of~$q$ such that for every two distinct atoms $R$ and $S$, their common variables occur in all atoms on the unique path from $R$ to $S$ in the tree. 
A CQ $q$ is {\em acyclic}\footnote{Throughout this paper, whenever we say that a CQ is acyclic, we mean acyclicity as defined in~\cite{DBLP:journals/jacm/BeeriFMY83}, a notion that today is also known as $\alpha$-acyclicity, to distinguish it from other notions of acyclicity.} if it has a join tree. If $\tau$ is a subtree of a join tree of a query~$q$, we will denote by $q_{\tau}$ the query whose atoms are the nodes of $\tau$.
Whenever $R$ is a node in an undirected tree $\tau$, then $(\tau,R)$ denotes the rooted tree obtained by choosing $R$ as the root of the tree.

\begin{example}
The join tree of the query $q$ in Example~\ref{ex:dbexample} has a single edge between $\relation{Employee}(\underline{x}, y, z)$ and $\relation{Manager}(\underline{y}, w, \texttt{2020})$.
\end{example}

\introparagraph{Attack graphs.}
Let $q$ be an acyclic, self-join-free BCQ with join tree~$\tau$. 
%
For every atom $F \in q$, we define $\keyclosure{F}{q}$ as the set of all variables in $q$ that are functionally determined by $\key{F}$ with respect to all functional dependencies of the form $\fd{\key{G}}{\var{G}}$ with $G\in q\setminus\{F\}$. 
Following~\cite{10.1145/2188349.2188351}, the \emph{attack graph} of $q$ is a directed graph whose vertices are the atoms of~$q$. There is a directed edge, called \emph{attack}, from $F$ to $G$ ($F\neq G$), if on the unique path between $F$ and $G$ in~$\tau$, every two adjacent atoms share a variable not in $\keyclosure{F}{q}$.
An atom without incoming edges in the attack graph is called \emph{unattacked}.
The attack graph of $q$ is used to determine the data complexity of $\cqa{q}$: the attack graph of $q$ is acyclic if and only if $\cqa{q}$ is in \FO~\cite{KoutrisW17}.
\begin{example}
For the query $q$ in Example~\ref{ex:dbexample}, 
$\keyclosure{\relation{Employee}}{q} = \{x\}$ and $\keyclosure{\relation{Manager}}{q} = \{y\}$. It follows that $\relation{Employee}$ attacks $\relation{Manager}$ because the variable $y$ is shared between atoms $\relation{Employee}$ and $\relation{Manager}$ and $y \notin \keyclosure{\relation{Employee}}{q}$. However, $\relation{Manager}$ does not attack $\relation{Employee}$ since the only shared variable $y$ is in $\keyclosure{\relation{Manager}}{q}$.

The attack graph of $q$ is acyclic since it only contains one attack from $\relation{Employee}$ to $\relation{Manager}$. It follows that $\cqa{q}$ is in \FO, as witnessed by the \FO-rewriting in Example~\ref{ex:cqa-prelim}.
\end{example}
}

\section{A Linear-Time Rewriting}\label{sec:rewriting}

Before presenting our linear-time rewriting for $\cqa{q}$, we first provide a motivating example.
Consider the following query on the \database{Company} database shown in Figure~\ref{fig:incon_db_example}:

\renewcommand{\quote}{\list{}{\rightmargin=\leftmargin\topsep=0pt}\item\relax}
\vspace{0.5ex}
\begin{quote}
\emph{
\revm{
Is there an office whose contact person works for the office and, moreover, manages the office since 2020?
}
}
\end{quote}
\vspace{0.5ex}
%

\noindent
This query can be expressed by the following CQ:
%
\revm{
\begin{align*}
\examplequery() \obtainedfrom \relation{Employee}(\underline{x}, y, z), \relation{Manager}(\underline{y}, x, \texttt{2020}), \relation{Contact}(\underline{y}, x).
\end{align*}
}
To the best of our knowledge, the most efficient running time for $\cqa{\examplequery}$ guaranteed by existing systems  is quadratic in the input database size, denoted~$N$.
The problem $\cqa{\examplequery}$ admits an \FO-rewriting by the classification theorem in \cite{DBLP:conf/pods/KoutrisOW21}.
However, the non-recursive Datalog rewriting of $\cqa{\examplequery}$ produced by Conquesto contains cartesian products between two tables, which means that it runs in $\Omega(N^2)$ time in the worst case. 
Also, since $\examplequery$ is not in $\mathcal{C}_{\mathsf{forest}}$, ConQuer cannot produce an \FO-rewriting. 
Both EQUIP and CAvSAT solve the problem through Integer Programming or SAT solvers, which can take exponential time.
One key observation is that $\examplequery$ requires a primary-key to primary-key join and a non-key to non-key join at the same time. As will become apparent in our technical treatment in Section~\ref{sec:datalog}, this property allows us to solve $\cqa{\examplequery}$ in $O(N)$ time, while existing CQA systems will run in more than linear time.

The remainder of this section is organized as follows. In Section~\ref{sec:ppjt}, we introduce the pair-pruning join tree (PPJT). In Section~\ref{sec:datalog}, we consider every Boolean query $q$ having a PPJT and present a novel linear-time non-recursive Datalog program for $\cqa{q}$ (Theorem~\ref{thm:main:boolean}). Finally, we extend our result to all acyclic self-join-free CQs in Section~\ref{sec:nonboolean} (Theorem~\ref{thm:main:full}) .

\subsection{Pair-pruning Join Tree}
\label{sec:ppjt}

Here we introduce the notion of a {\em pair-pruning join tree} (PPJT). We first assume that the query $q$ is connected, and then discuss how to handle disconnected queries at the end of the section.

Recall that an atom in a self-join-free query can be uniquely denoted by its relation name.
For example, we may use $\relation{Employee}$ as a shorthand for the atom $\relation{Employee}(\underline{x},y,z)$ in $\examplequery$.


\begin{definition}[PPJT]
\label{defn:ppjt}
Let $q$ be an acyclic self-join-free BCQ. Let $\tau$ be a join tree of $q$ and $R$ a node in $\tau$. The rooted tree $(\tau, R)$ is a \textit{pair-pruning join tree (PPJT)} of $q$ if for any rooted subtree $(\tau', R')$ of $(\tau, R)$, the atom~$R'$ is unattacked in $q_{\tau'}$.
\end{definition}


\begin{example}
For the join tree $\tau$ in Figure~\ref{fig:join-tree},
the rooted tree $(\tau, \relation{Employee})$ is a PPJT for $\examplequery$. The atom $\relation{Employee}(\underline{x}, y, z)$ is unattacked in $q$. For the child subtree $(\tau_M, \relation{Manager})$ of $(\tau, \relation{Employee})$, the atom $\relation{Manager}(\underline{y}, x, \texttt{2020})$ is also unattacked in the following subquery
\begin{align*}
\examplequery_{\tau_M}() \obtainedfrom \relation{Manager}(\underline{y}, x, \texttt{2020}), \relation{Contact}(\underline{y}, x).
\end{align*}
Finally, for the subtree $(\tau_C, \relation{Contact})$, the atom $\relation{Contact}(\underline{y}, x)$ is also unattacked in the corresponding subquery $\examplequery_{\tau_C}() \obtainedfrom \relation{Contact}(\underline{y}, x).
$
%
%
Hence $(\tau, \relation{Employee})$ is a PPJT of $\examplequery$. \qed
\end{example}

\begin{figure}[t]
\begin{tikzpicture}[level distance=10em,
every node/.style = {line width=1pt, shape=rectangle, text centered, anchor=center, rounded corners, draw, align=center}, grow=right]
\tikzset{blue dotted/.style={draw=blue!50!white, line width=1pt,
                               dash pattern=on 1pt off 4pt on 6pt off 4pt,
                                inner sep=5mm, rectangle, rounded corners}};

\node (E) {$\textsf{Employee}(\underline{x}, y, z)$};
\node (M) [right of=E, node distance=3cm] {$\textsf{Manager}(\underline{y}, x, \texttt{2020})$};
\node (C) [right of=M, node distance=2.8cm] {$\textsf{Contact}(\underline{y}, x)$};

\path[->] (E) edge [line width=1pt] (M);
\path[->] (M) edge [line width=1pt] (C);

\node (box1) [draw=blue, line width=1pt,
                               dash pattern=on 1pt off 4pt on 6pt off 4pt,
                                inner xsep=3mm, inner ysep=4mm, rectangle, rounded corners,  fit = (E) (C)] {};
\node [above left = 0mm and -2mm of box1, draw=none,above] {{\color{blue} $\tau$}};                                
 \node (box2) [draw=red, line width=1pt,
                               dash pattern=on 1pt off 4pt on 6pt off 4pt,
                                inner xsep=2.2mm, inner ysep=2.6mm, rectangle, rounded corners,  fit = (M) (C)] {};    
\node [above left = 0.9mm and -2mm of box2, draw=none,above] {{\color{red} $\tau_M$}};                                   
\node (box3) [draw=brown, line width=1pt,
                               dash pattern=on 1pt off 4pt on 6pt off 4pt,
                                inner xsep=1.3mm, inner ysep=1.2mm, rectangle, rounded corners,  fit = (C) (C)] {};                                      
\node [above left = 2.2mm and -2mm of box3, draw=none,above] {{\color{brown} $\tau_C$}}; 

\end{tikzpicture}
\caption{A pair-pruning join tree (PPJT) of the query $\examplequery$.
}
\label{fig:join-tree}
\end{figure}

\introparagraph{Which queries admit a PPJT?}
As we show next, having a PPJT is a sufficient condition for the existence of an \FO-rewriting. 

\begin{proposition}\label{prop:ppjt}
Let $q$ be an acyclic self-join-free BCQ. If $q$ has a PPJT, then $\cqa{q}$ admits an \FO-rewriting.
\end{proposition}

\reva{
Proposition~\ref{prop:ppjt} is proved in Appendix~\ref{proof:prop:ppjt}, in which we show that the if $q$ has a PPJT, then the attack graph of $q$ must be acyclic.}
We note that not all acyclic self-join-free BCQs with an acyclic attack graph have a PPJT, as demonstrated in the next example.

\begin{example}
\label{ex:noppjt}
Let $q() \obtainedfrom R(\underline{x, w}, y), S(\underline{y, w}, z), T(\underline{w}, z)$. The attack graph of $q$ is acyclic. The only join tree $\tau$ of $q$ is the path $R - S - T$. However, neither $(\tau, R)$ nor $(\tau, S)$ is a PPJT for $q$ since $R$ and $S$ are attacked in $q$; and $(\tau, T)$ is not a PPJT since in its subtree $(\tau', S)$, $S$ is attacked in the subquery that contains $R$ and $S$. \qed
\end{example}

\revb{Fuxman and Miller~\cite{FUXMAN2007610} identified a large class of self-join-free CQs, called $\cforest$, that includes most queries with primary-key-foreign-key joins, path queries, and queries on a star schema, such as found in SSB and TPC-H \cite{DBLP:conf/tpctc/ONeilOCR09, poess2000new}. This class covers most of the SPJ queries seen in practical settings. 
In view of this, the following proposition is of practical significance.

\begin{proposition}
\label{prop:cforest}
Every acyclic BCQ in $\cforest$ has a PPJT.
\end{proposition}

Furthermore, it is easy to verify that, unlike $\cforest$, PPJT captures \emph{all} \FO-rewritable self-join-free SPJ queries on two tables, a.k.a. binary joins.
For example, the binary join $q_5$ in Section~\ref{sec:synthetic} admits a PPJT but is not in $\cforest$.

Proposition~\ref{prop:cforest} is proved in Appendix~\ref{proof:prop:cforest}.
}


\introparagraph{How to find a PPJT.}
For any acyclic self-join-free BCQ $q$, we can check whether $q$ admits a PPJT via a brute-force search over all possible join trees and roots. If $q$ involves $n$ relations, then there are at most $n^{n-1}$ candidate rooted join trees for PPJT ($n^{n-2}$ join trees and for each join tree, $n$ choices for the root). \revb{For the data complexity of $\cqa{q}$, this exhaustive search runs in constant time since we assume $n$ is a constant. In practice, the search cost is acceptable for most join queries that do not involve too many tables.}

%
Appendix~\ref{sec:no-key-containment} shows that the foregoing brute-force search for $q$ can be optimized to run in polynomial time when~$q$ has an acyclic attack graph and, \revb{when expressed as a rule}, does not contain two distinct body atoms $R(\underline{\vec{x}},\vec{y})$ and $S(\underline{\vec{u}},\vec{w})$ such that every variable occurring in $\vec{x}$ also occurs in~$\vec{u}$. Most queries we observe and used in our experiments fall under this category.

\introparagraph{Main Result.} 
We previously showed that the existence of a PPJT implies an \FO-rewriting that computes the consistent answers. Our main result shows that it also leads to an efficient algorithm that runs in linear time.

\begin{theorem}
\label{thm:main:boolean}
Let $q$ be an acyclic self-join-free BCQ that admits a PPJT, and $\db$ be a database instance of size $N$. Then, there exists an algorithm for $\cqa{q}$ that runs in time $O(N)$.
\end{theorem}

It is worth contrasting our result with Yannakakis' algorithm, which computes the result of any acyclic BCQ also in linear time $O(N)$~\cite{DBLP:conf/vldb/Yannakakis81}. Hence, the existence of a PPJT implies that computing $\cqa{q}$ will have the same asymptotic complexity. 

\introparagraph{Disconnected CQs.}
Every disconnected BCQ $q$ can be written as $q = q_1, q_2, \dots, q_n$ where $\var{q_i} \cap \var{q_j} = \emptyset$ for $1 \leq i < j \leq n$ and each $q_i$ is connected. If each $q_i$ has a PPJT, then $\cqa{q}$ can be solved by checking whether the input database is a ``yes''-instance for each $\cqa{q_i}$, by Lemma~B.1 of \cite{DBLP:conf/pods/KoutrisOW21}.

\subsection{The Rewriting Rules}
\label{sec:datalog}

We now show how to produce an efficient rewriting in Datalog \reva{and prove Theorem~\ref{thm:main:boolean}}. In Section~\ref{sec:implementation}, we will discuss how to translate the Datalog program to SQL. Let $q$ be an acyclic self-join-free BCQ with a PPJT $(\tau, R)$ and $\db$ an instance for the problem $\cqa{q}$: does the query $q$ evaluate to true on every repair of $\db$?

\revb{
Let us first revisit Yannakakis' algorithm for evaluating~$q$ on a database~$\db$ in linear time. 
Given a rooted join tree $(\tau, R)$ of $q$, Yannakakis' algorithm visits all nodes in a bottom-up fashion. For every internal node $S$ of $(\tau, R)$, it keeps the tuples in table $S$ that join with every child of $S$ in $(\tau, R)$, where each such child has been visited recursively. In the end, the algorithm returns whether the root table $R$ is empty or not.
Equivalently, Yannakakis' algorithm evaluates~$q$ on~$\db$ by removing tuples from each table that cannot contribute to an answer in~$\db$ at each recursive step.

Our algorithm for CQA proceeds like Yannakakis' algorithm in a bottom-up fashion. At each step, we remove tuples from each table that cannot contribute to an answer to $q$ in at least one repair of $\db$. 
Informally, if a tuple cannot contribute to an answer in at least one repair of $\db$ containing it, then it cannot contribute to a consistent answer to $q$ on $\db$.
Specifically, given a PPJT $(\tau, R)$ of $q$, to compute all tuples of each internal node~$S$ of $(\tau, R)$ that may contribute to a consistent answer, we need to ``prune'' the blocks of $S$ in which there is some tuple that violates either the local selection condition on table $S$, or the joining condition with some child table of $S$ in $(\tau, R)$. 
The term ``pair-pruning'' is motivated by the latter process, where we consider only one pair of tables at a time. This idea is formalized in Algorithm~\ref{fig:ppjt}, where the procedures \textsc{Self-Pruning} and \textsc{Pair-Pruning} prune, respectively, the blocks that violate the local selection condition and the joining condition.

}


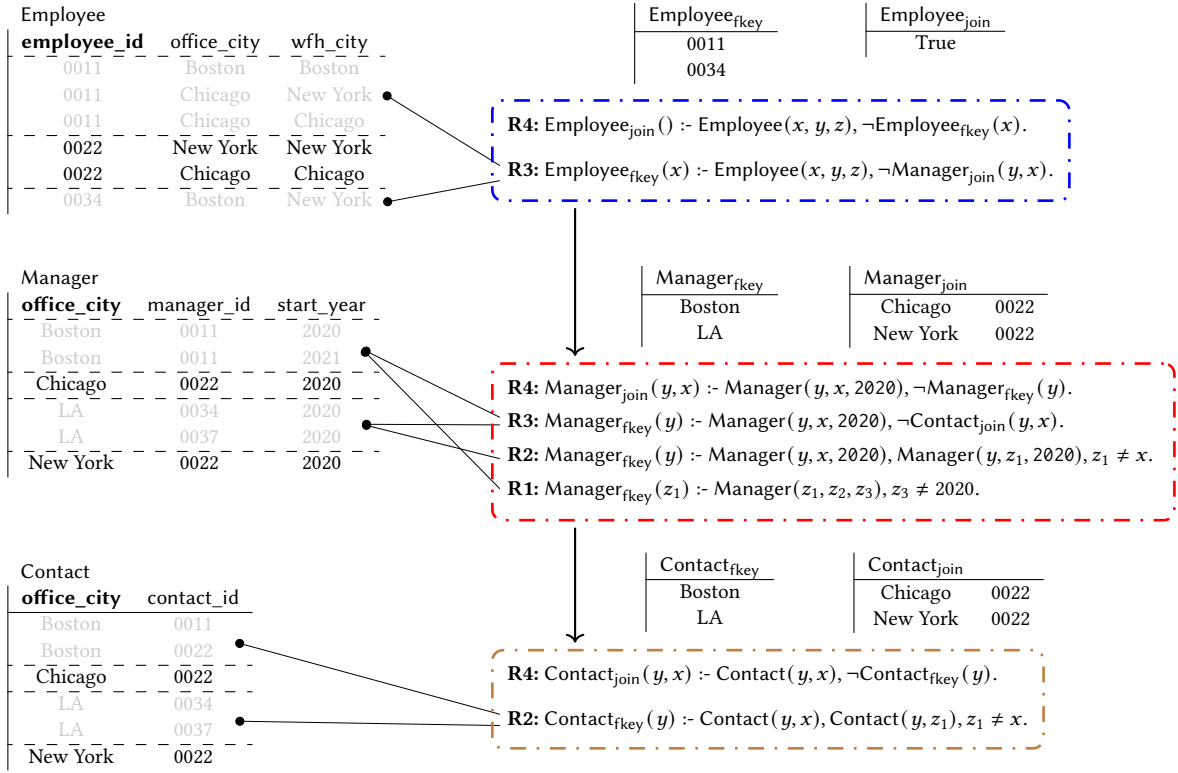
\begin{figure*}[t]
\small

     \begin{tikzpicture}[level distance=10em,
every node/.style = {line width=0.5pt, shape=rectangle, text centered, anchor=center, rounded corners, draw, align=center}, grow=right,
brace/.style = {black,decorate, decoration = {brace, raise=10pt,amplitude=3pt}}]
\tikzset{blue dotted/.style={draw=blue!50!white, line width=1pt,
                               dash pattern=on 1pt off 4pt on 6pt off 4pt,
                                inner sep=5mm, rectangle, rounded corners}};

\node (E-exit) [draw=none] {
     \textbf{R4: }$\relation{Employee}_{\mathsf{join}}() \obtainedfrom \relation{Employee}(x, y, z), \lnot \fk{\relation{Employee}}(x).$};
\node (E-self) [below = 0.6cm of E-exit.west, anchor=west, draw=none] {    
     \textbf{R3: }$\fk{\relation{Employee}}(x)  \obtainedfrom \relation{Employee}(x, y, z), \lnot \relation{Manager}_{\mathsf{join}}(y, x).$};

\node (E-table) [above left=1.5cm and 4cm of E-exit.north west, anchor=north, draw=none] {
\begin{tabular}{|c c c}
\multicolumn{3}{l}{\textsf{Employee}} \\
\primarykey{\textsf{employee\_id}} & \textsf{office\_city} & \textsf{wfh\_city}  \\
\cdashline{1-3}
{\transparent{0.2}0011} & {\transparent{0.2}Boston} &{\transparent{0.2}Boston}  \\
{\transparent{0.2}0011} & {\transparent{0.2}Chicago} & \transparent{0.2}New York  \\
{\transparent{0.2}0011} & \transparent{0.2}Chicago & \transparent{0.2}Chicago  \\
\cdashline{1-3}
0022 & New York & New York  \\
0022 & Chicago & Chicago  \\
\cdashline{1-3}
{\transparent{0.2}0034} & \transparent{0.2}Boston & \transparent{0.2}New York 
\end{tabular}
};

\path[-Circle] (E-self.west)+(0,0.1) edge [line width=0.3pt] node (o3) [draw=none,anchor=west] {} ($(-0.2,0.2)+(E-table.east)$) ;
\path[-Circle] (E-self.west)+(0,-0.1) edge [line width=0.3pt] node (o4) [draw=none,anchor=west] {} ($(-0.2,-1.25)+(E-table.east)$) ;

\node (E-jkey)  [above right=1.5cm and 6cm of E-exit.north west, anchor=north, draw=none] {
\begin{tabular}{|c c}
$\textsf{Employee}_{\mathsf{join}}$ & \\
\cline{1-1}
{True} &
\end{tabular}
};

\node (E-fkey)  [above right=1.5cm and 3cm of E-exit.north west, anchor=north, draw=none] {
\begin{tabular}{|c c}
$\fk{\textsf{Employee}}$ & \\
\cline{1-1}
0011 & \multirow{2}{0.1cm}{}\\
0034 &  \\
\end{tabular}
};

\node (M-exit) [below = 2.9cm of E-self.west, anchor=west, draw=none]{     
      \textbf{R4: }$\relation{Manager}_{\mathsf{join}}(y, x)  \obtainedfrom \relation{Manager}(y, x, \texttt{2020}), \lnot \fk{\relation{Manager}}(y).$}; 
     
\node (M-self) [below =0.9cm of M-exit.west, anchor=west, draw=none]{ 
     $\begin{aligned}
     \textbf{R3: }& \fk{\relation{Manager}}(y)  \obtainedfrom \relation{Manager}(y, x, \texttt{2020}), \lnot \relation{Contact}_{\mathsf{join}}(y, x). \\
     \textbf{R2: }&\fk{\relation{Manager}}(y)  \obtainedfrom \relation{Manager}(y, x, \texttt{2020}), \relation{Manager}(y, z_1, \texttt{2020}), z_1 \neq x. \\
     \textbf{R1: }&\fk{\relation{Manager}}(z_1)  \obtainedfrom \relation{Manager}(z_1, z_2, z_3), z_3 \neq 2020.
     \end{aligned}$
     };

\node (M-table) [below =3.5cm of E-table.west, anchor=west, draw=none] {
\begin{tabular}{|c c c}
\multicolumn{3}{l}{\textsf{Manager}} \\
\primarykey{\textsf{office\_city}} & \textsf{manager\_id} & \textsf{start\_year}  \\
\cdashline{1-3}
{\transparent{0.2}Boston} & {\transparent{0.2}0011} & {\transparent{0.2}2020}  \\
{\transparent{0.2}Boston} & {\transparent{0.2}0011} & {\transparent{0.2}2021} \\
\cdashline{1-3}
 Chicago &  0022 & 2020  \\
\cdashline{1-3}
{\transparent{0.2}LA} & {\transparent{0.2}0034} & {\transparent{0.2}2020} \\
{\transparent{0.2}LA} & {\transparent{0.2}0037} & {\transparent{0.2}2020}  \\
\cdashline{1-3}
{New York} & {0022} & 2020
\end{tabular}
};

\path[-Circle] (M-self.west)+(0,-0.4) edge [line width=0.3pt] node (o1) [draw=none,anchor=west] {} ($(-0.4,0.3)+(M-table.east)$) ;
\path[-Circle] (M-self.west)+(0,0) edge [line width=0.3pt] node (o1) [draw=none,anchor=west] {} ($(-0.4,-0.7)+(M-table.east)$) ;
\path[-Circle] (M-self.west)+(0,0.45) edge [line width=0.3pt] node (o1) [draw=none,anchor=west] {} ($(-0.4,-0.7)+(M-table.east)$) ;
\path[-Circle] (M-self.west)+(0,0.55) edge [line width=0.3pt] node (o1) [draw=none,anchor=west] {} ($(-0.4,0.3)+(M-table.east)$) ;

\node (M-jkey) [above right=1.5cm and 6cm of M-exit.north west, anchor=north, draw=none] {
\begin{tabular}{|c c}
$\textsf{Manager}_{\mathsf{join}}$ & \\
\hline
Chicago &  0022 \\
New York & 0022 
\end{tabular}
};

\node (M-fkey)   [above right=1.5cm and 3cm of M-exit.north west, anchor=north, draw=none] {
\begin{tabular}{|c c}
$\fk{\textsf{Manager}}$ & \\
\cline{1-1}
{Boston} & \\
{LA} &  \\
\end{tabular}
};

\node (C-exit) [below =3.8 cm of M-exit.west, anchor=west, draw=none] {
     \textbf{R4: }$\relation{Contact}_{\mathsf{join}}(y, x)  \obtainedfrom \relation{Contact}(y, x), \lnot \fk{\relation{Contact}}(y).$};

\node (C-self) [below = 0.6cm of C-exit.west, anchor = west, draw=none] {
     \textbf{R2: }$\fk{\relation{Contact}}(y)  \obtainedfrom \relation{Contact}(y, x), \relation{Contact}(y, z_1), z_1 \neq x.$};

\node (C-table) [below =3.9cm of M-table.west, anchor=west, draw=none] {
\begin{tabular}{|c c }
\multicolumn{2}{l}{\textsf{Contact}} \\
\primarykey{\textsf{office\_city}} & \textsf{contact\_id}  \\
\cline{1-2}
\transparent{0.2}Boston & {\transparent{0.2}0011}  \\
{\transparent{0.2}Boston} & {\transparent{0.2}0022}  \\
\cdashline{1-2}
{Chicago} &  {0022}  \\
\cdashline{1-2}
\transparent{0.2}LA & \transparent{0.2}0034 \\
{\transparent{0.2}LA} & \transparent{0.2}0037  \\
\cdashline{1-2}    
{New York} & {0022}  
\end{tabular}
};

\path[-Circle] (C-self.west)+(0,0.1) edge [line width=0.3pt] node (o1) [draw=none,anchor=west] {} ($(-0.4,0.3)+(C-table.east)$) ;
\path[-Circle] (C-self.west)+(0,-0.05) edge [line width=0.3pt] node (o2) [draw=none,anchor=west] {} ($(-0.4,-0.75)+(C-table.east)$) ;

\node (C-jkey)  [above right=1.5cm and 6cm of C-exit.north west, anchor=north, draw=none] {
\begin{tabular}{|c c}
$\textsf{Contact}_{\mathsf{join}}$ & \\
\hline
Chicago &  0022 \\
New York & 0022 
\end{tabular}
};

\node (C-fkey) [above right=1.5cm and 3cm of C-exit.north west, anchor=north, draw=none] {
\begin{tabular}{|c c}
$\fk{\relation{Contact}}$ \\
\cline{1-1}
{Boston}& \multirow{2}{0cm}{} \\
{LA}
\end{tabular}
};

\path[->] (E-self.south west)+(1,-0.2) edge [line width=0.8pt] node (PP) [draw=none,anchor=west] 
{} ($(1,0.2)+(M-exit.north west)$) ;

\path[->] (M-self.south west)+(1,-0.2) edge [line width=0.8pt] node (PP2) [draw=none,anchor=west] 
{} ($(C-exit.north west)+(1,0.2)$) ;

\node (box1) [draw=blue, line width=1pt,
                               dash pattern=on 1pt off 4pt on 6pt off 4pt,
                                inner sep=1mm, rectangle, rounded corners,  fit = (E-exit) (E-self)] {};

 \node (box2) [draw=red, line width=1pt,
                               dash pattern=on 1pt off 4pt on 6pt off 4pt,
                                inner sep=1mm, rectangle, rounded corners,  fit = (M-exit) (M-self)] {};    

\node (box3) [draw=brown, line width=1pt,
                               dash pattern=on 1pt off 4pt on 6pt off 4pt,
                                inner sep=1mm, rectangle, rounded corners,  fit = (C-exit) (C-self)] {};                                      
                                                        
\end{tikzpicture}
\caption{\revm{The non-recursive Datalog program for evaluating $\cqa{\examplequery}$ together with an example execution on the \textbf{Company} database in Figure~\ref{fig:incon_db_example}. The faded-out rows denote blocks that are removed since they do not contribute to any consistent answer. The arrows denote which rules remove which blocks (some blocks can be removed by multiple rules).}}
\label{fig:example-rewriting}
\end{figure*}

\revb{
%
To ease the exposition of the rewriting, we now present both procedures in Datalog syntax. We will use two predicates for every atom $S$ in the tree (let $T$ be the unique parent of $S$ in $\tau$):
}
\begin{itemize}
\item the predicate $\fk{S}$ has arity equal to $|\key{S}|$ and collects the primary-key values of the $S$-table that cannot contribute to a consistent answer for $q$~\footnote{The \textsf{f} in $\textsf{fkey}$ is for ``\underline{f}alse key''. 
}; and
\item the predicate $\jk{S}$ has arity equal to $|\var{S} \cap \var{T}|$ and collects the values for these variables in the $S$-table that may contribute to a consistent answer.
\end{itemize}

%
\SetKwInput{KwInput}{Input}                
\SetKwInput{KwOutput}{Output}              
\begin{algorithm}[!ht]
\DontPrintSemicolon
  
  \KwInput{PPJT $(\tau, R)$ of $q$}
  \KwOutput{a Datalog program $P$ deciding $\cqa{q}$}

  $P:=\emptyset$ \; 
  $P:=P\cup\textsc{Self-Pruning}(R)$ \;
  \ForEach{child node $S$ of $R$ in $\tau$}{
  $P:=P\cup\textsc{PPJT-Rewriting}(\tau, S)$ \;
  $P:=P\cup\textsc{Pair-Pruning}(R, S)$ \;
  }
  $P:=P\cup\textsc{Exit-Rule}(R)$ \;
  \Return $P$ \;
\caption{$\textsc{PPJT-Rewriting}(\tau, R)$}
    \label{fig:ppjt}
\end{algorithm}

Figure~\ref{fig:example-rewriting} depicts how each step generates the rewriting rules for  $\examplequery$.
We now describe how each step is implemented in detail. 

\smallskip
\noindent \underline{\textsc{Self-Pruning}($R$)}: Let $R(\underline{x_1, \dots x_k}, x_{k+1}, \dots, x_n)$, where $x_i$ can be a variable or a constant. \revb{The first rule finds the primary-key values of the $R$-table that can be pruned because some tuple with that primary-key violates the local selection conditions imposed on $R$.}


\begin{rules}
\label{rule:self-pruning}
If $x_i = c$ for some constant $c$, we add the rule 
\begin{align*}
\fk{R}(z_1, \dots, z_k) \obtainedfrom R(z_1, \dots, z_n), z_i \neq c.
\end{align*}
If for some variable $x_i$ there exists $j<i$ with $x_i = x_j$, we add the rule
\begin{align*}
\fk{R}(z_1, \dots, z_k) \obtainedfrom R(z_1, \dots, z_n), z_i \neq z_j.
\end{align*}
Here, $z_1, \dots, z_n$ are fresh distinct variables.
\end{rules}

The second rule finds the primary-key values of the $R$-table that can be pruned because $R$ joins with its parent $T$ in the tree. The underlying intuition is that if some $R$-block of the input database contains two tuples that disagree on a non-key position that is used in an equality-join with~$T$, then for every given $T$-tuple~$t$, we can pick an $R$-tuple in that block that does not join with $t$. \revb{Therefore, that $R$-block cannot contribute to a consistent answer.}


\begin{rules}
\label{rule:good-join}
 For each variable $x_i$ with $i>k$ (so in a non-key position) such that $x_i \in \var{T}$, we produce a rule
\begin{align*}
     \fk{R}(x_1, \dots, x_k)  \obtainedfrom & R(x_1, \dots, x_k, x_{k+1}, \dots, x_n), \\
     & R(x_1, \dots, x_k, z_{k+1}, \dots,  z_k), z_i \neq x_i.
\end{align*}
where $z_{k+1}, \dots, z_n$ are fresh variables.
\end{rules}

\begin{example}
The self-pruning phase on$(\tau_M, \relation{Manager})$ produces one rule using Rule~\ref{rule:self-pruning}. When executed on the \database{Company} database, the key Boston is added to $\fk{\relation{Manager}}$, \revb{since the tuple (Boston, 0011, 2021) has $\textsf{start\_year}  \neq 2020$.}

Finally, the self-pruning phase on the PPJT $(\tau_C, \relation{Contact})$ produces one rule using Rule~\ref{rule:good-join} (here $x$ is the non-key join variable). Hence, the keys Boston and LA will be added to $\fk{\relation{Contact}}$. \qed
\end{example}

\smallskip
\noindent \underline{\textsc{Pair-Pruning}($R,S$)}: Suppose that $q$ contains the atoms $R(\underline{\vec{x}}, \vec{y})$ and $S(\underline{\vec{u}}, \vec{v})$, where the $S$-atom is a child of the $R$-atom in the PPJT. Let $\vec{w}$ be a sequence of distinct variables containing all (and only) variables in $\var{R} \cap \var{S}$.
\revb{The third rule prunes all $R$-blocks containing some tuple that cannot join with any $S$-tuple to contribute to a consistent answer.}




\begin{rules}
\label{rule:bad-key-simplified}
Add the rule
\begin{align*}
     \fk{R}(\vec{x}) & \obtainedfrom R(\vec{x}, \vec{y}), \lnot S_{\mathsf{join}}(\vec{w}),
\end{align*}
where the rules for $S_{\mathsf{join}}$ will be defined in Rule~\ref{rule:return-simplified}.
\end{rules}
The rule is safe because every variable in $\vec{w}$ occurs in $R(\vec{x}, \vec{y})$.


\begin{example}
Figure~\ref{fig:example-rewriting} shows the two pair-pruning rules generated (in general, there will be one pair-pruning rule for each parent-child edge in the PPJT. In both cases, the join variables are $\{y,x\}$. \revb{For the table  \relation{Employee}, the rule prunes the two blocks with keys $0011, 0034$ and adds them to $\fk{\relation{Employee}}$.  } \qed
\end{example}

\smallskip
\noindent \underline{\textsc{Exit-Rule}($R$)}: Suppose that $q$ contains $R(\underline{\vec{x}}, \vec{y})$. 
If $R$ is an internal node, let $\vec{w}$ be a sequence of distinct variables containing all (and only) the join variables of $R$ and its parent node in $\tau$. If $R$ is the root node, let~$\vec{w}$ be the empty vector.
\revb{The exit rule removes the pruned blocks of~$R$ and projects on the variables in $\vec{w}$. If~$R$ is an internal node, the resulting tuples in the projection could contribute to a consistent answer, and will be later used for pair pruning; if $R$ is the root, the projection returns the final result.}

\begin{rules}
\label{rule:return-simplified}
 If $\fk{R}$ exists in the head of a rule, we produce the rule
\begin{align*}
     {R}_{\mathsf{join}}(\vec{w}) & \obtainedfrom R(\vec{x}, \vec{y}), \lnot \fk{R}(\vec{x}).
\end{align*}
Otherwise, we produce the rule
\begin{align*}
     {R}_{\mathsf{join}}(\vec{w}) & \obtainedfrom R(\vec{x}, \vec{y}).
\end{align*}
\end{rules}

%

\begin{example}
Figure~\ref{fig:example-rewriting} shows the three exit rules for $\examplequery$---one rule for each node in the PPJT. 
The boolean predicate $\jk{\relation{Employee}}$ determines whether \texttt{True} is the consistent answer to the query. \qed
\end{example}

\introparagraph{Runtime Analysis} It is easy to see that \textbf{Rule 1, 3, and 4} can be evaluated in linear time. We now argue how to evaluate \textbf{Rule 2} in linear time as well. Indeed, instead of performing the self-join on the key, it suffices to create a hash table using the primary key as the hash key (which can be constructed in linear time). Then, for every value of the key, we can easily check whether all tuples in the block have the same value at the $i$-th attribute.



\revb{
\introparagraph{Sketch of Correctness}
Let $q$ be an acyclic self-join-free BCQ with a PPJT $(\tau, R)$ and $\db$ an instance for $\cqa{q}$. The easier property to show is the \emph{soundness} of our rewriting \textbf{Rules~1, 2, 3, 4}: if the predicate $R_{\mathsf{join}}$ is nonempty when our rewriting is executed on $\db$, then every repair of $\db$ must necessarily satisfy~$q$. 
The argumentation uses a straightforward bottom-up induction on the PPJT: for every rooted subtree $(\tau',S)$ of $(\tau,R)$, the tuples in $S_{\mathsf{join}}$ are consistent answers to the corresponding subquery $q_{\tau'}$ projected on the join variables with the parent of~$S$ (i.e., on the variables~$\vec{w}$ in \textbf{Rule 4}).

The more difficult property to show is the \emph{completeness} of our rewriting rules: if every repair of $\db$ satisfies $q$, then the predicate $R_{\mathsf{join}}$ must be nonempty after executing the rules on $\db$.
The crux here is a known result (see, for example, Lemma~4.4 in~\cite{KoutrisW17}) which states that for every unattacked atom $R$ in a self-join-free BCQ $q$, the following holds true: 
\begin{quote}
\emph{if every repair of $\db$ satisfies $q$, then there is a nonempty block $\block$ of $R$ such that in each repair of $\db$, the query~$q$ can be made true by using the (unique) tuple of $\block$ in that repair.}
\end{quote}
Our recursive construction of a PPJT $(\tau, R)$ ensures that for each rooted subtree $(\tau', S)$ of $(\tau, R)$, $S$ is unattacked in $q_{\tau'}$. 
Therefore, it suffices to compute the blocks in $S$ that could contribute to a consistent answer to~$q_{\tau'}$ at each recursive step in a bottom-up fashion, eventually returning the consistent answer to $q$ in $\db$. 

The soundness and completeness arguments taken together imply that our rewriting rules return only and all consistent answers.
The formal correctness proof is in Appendix~\ref{sec:correctness-boolean}.
}

\subsection{Extension to Non-Boolean Queries}
\label{sec:nonboolean}

Let $q(\vec{u})$ be an acyclic self-join-free CQ with free variables $\vec{u}$, and $\db$ be a database instance. If $\vec{c}$ is a sequence of constants of the same length as $\vec{u}$, we say that $\vec{c}$ is a {\em consistent answer} to $q$ on $\db$ if $\vec{c} \in q(I)$ in every repair $I$ of $\db$.
Furthermore, we say that $\vec{c}$ is a {\em possible answer} to $q$ on $\db$ if $\vec{c} \in q(\db)$. It can be easily seen that for CQs every consistent answer is a possible answer.

 Lemma~\ref{lemma:parameterization} reduces computing the consistent answers of non-Boolean queries to that of Boolean queries.

\begin{lemma}
\label{lemma:parameterization}
Let $q$ be a CQ with free variables $\vec{u}$, and let $\vec{c}$ be a sequence of constants of the same length as $\vec{u}$. Let $\db$ be an database instance. Then $\vec{c}$ is a consistent answer to $q$ on $\db$ 
if and only if $\db$ is a ``yes''-instance for $\cqa{q_{[\vec{u} \rightarrow \vec{c}]}}$. 
\end{lemma}

 If $q$ has free variables $\vec{u} = (u_1, u_2, \dots, u_n)$, we say that $q$ admits a PPJT if the Boolean query $q_{[\vec{u} \rightarrow \vec{c}]}$ admits a PPJT, where $\vec{c} = (c_1, c_2, \dots, c_n)$ is a sequence of distinct constants.
 We can now state our main result for non-Boolean CQs.

\begin{theorem}\label{thm:main:full}
Let $q$ be an acyclic self-join-free Conjunctive Query that admits a PPJT, and $\db$ be a database instance of size $N$.
Let $\mathsf{OUT}_p$ be the set of possible answers to $q$ on $\db$,
and $\mathsf{OUT}_c$ the set of consistent answers to $q$ on $\db$. Then:
\begin{packed_enum}
\item the set of consistent answers can be computed in time $O(N \cdot |\mathsf{OUT}_p|)$; and
\item moreover, if $q$ is full, the set of consistent answers can be computed in time $O(N + |\mathsf{OUT}_c|)$.
\end{packed_enum}
\end{theorem}

To contrast this with Yannakakis result, for acyclic full CQs we have a running time of $O(N + |\mathsf{OUT}|)$, and a running time of $O(N \cdot |\mathsf{OUT}|)$ for general CQs. 

%

\begin{proof}[Proof Sketch]

\reva{Our algorithm first evaluates $q$ on $\db$ to yield a set $S$ of size $|\mathsf{OUT}_{p}|$ in time $O(N \cdot |\mathsf{OUT}_p|)$.  
We then return all answers $\vec{c} \in S$ such that $\db$ is a ``yes''-instance for $\cqa{q_{[\vec{u} \rightarrow \vec{c}]}}$, which runs in $O(N)$ by Theorem~\ref{thm:main:boolean}. This approach gives an algorithm with running time $O(N \cdot |\mathsf{OUT}_p|)$.

If $q$ is full, we proceed by (i) removing all blocks with at least two tuples from $\db$ to yield $\db^c$ and (ii) evaluating $q$ on $\db^c$. 
In our algorithm, step (i) runs in $O(N)$ and since $q$ is full, step (ii) runs in time $O(N + |\mathsf{OUT}_c|)$. The correctness proof of these algorithms are in Appendix~\ref{proof:thm:main:full}.
}
\end{proof}


\introparagraph{Rewriting for non-Boolean Queries}
Let $\vec{c} = (c_1, c_2, \dots, c_n)$ be a sequence of fresh, distinct constants.
If $q_{[\vec{u} \rightarrow \vec{c}]}$ has a PPJT, the Datalog rewriting for $\cqa{q}$ can be obtained as follows:
\begin{packed_enum}
\item Produce the program $P$ for $\cqa{q_{[\vec{u} \rightarrow \vec{c}]}}$ using the rewriting algorithm for Boolean queries (Subsection~\ref{sec:datalog}).
\item Replace each occurrence of the constant $c_i$ in $P$ with the free variable $u_i$.
\item Add the rule: $\mathsf{ground}(\vec{u}) \obtainedfrom \mathsf{body}(q)$.
\item For a relation $T$, let $\vec{u}_T$ be a sequence of all free variables that occur in the subtree rooted at $T$. Then, append $\vec{u}_T$ to every occurrence of $\jk{T}$ and $\fk{T}$.
\item For any rule of $P$ that has a free variable $u_i$ that is unsafe, add the atom $\mathsf{ground}(\vec{u})$ to the rule.
\end{packed_enum}

\begin{example}
Consider the non-Boolean query 
$$\examplequerynb(w) \obtainedfrom \mathsf{Employee}(\underline{x}, y, z), \mathsf{Manager}(\underline{y}, x, w), \mathsf{Contact}(\underline{y}, x).$$ 

Note that the constant \texttt{2020} in $\examplequery$ is replaced by the free variable $w$ in $\examplequerynb$. Hence, the program $P$ for $\cqa{\examplequerynb_{[w \rightarrow c]}}$ is the same as Figure~\ref{fig:example-rewriting}, with the only difference that \texttt{2020} is replaced by the constant $c$.
The ground rule produced is:
$$\mathsf{ground}(w) \obtainedfrom \mathsf{Employee}({x}, y, z), \mathsf{Manager}({y}, x, w), \mathsf{Contact}({y}, x),$$
and Figure~\ref{fig:nonboolean-yannakakis} shows how Yannakakis' algorithm evaluates $\examplequerynb$.  

To see how the rule of $P$ would change for the non-Boolean case, consider the self-pruning rule for $\mathsf{Contact}$. This rule would remain as is, because it contains no free variable and the predicate $\fk{\relation{Contact}}$ remains unchanged. In contrast, consider the first self-pruning rule for $\mathsf{Manager}$, which in $P$ would be:
$$ \fk{\relation{Manager}}(y_1)  \obtainedfrom \relation{Manager}(y_1, y_2, y_3), y_3 \neq w. $$

Here, $w$ is unsafe, so we need to add the atom $\mathsf{ground}(w)$. Additionally, $w$ is now a free variable in the subtree rooted at $\relation{Manager}$, so the predicate $\fk{\relation{Manager}}(y_1)$ becomes $\fk{\relation{Manager}}(y_1,w) $. The transformed rule will be:
$$ \fk{\relation{Manager}}(y_1,w)  \obtainedfrom \relation{Manager}(y_1, y_2, y_3), y_3 \neq w, \mathsf{ground}(w).$$
The full rewriting for $\examplequerynb$ can be seen in Figure~\ref{fig:example-nonboolean}. \qed
\end{example}

The above rewriting process may introduce cartesian products in the rules. In the next section, we will see how we can tweak the rules in order to avoid this inefficiency.

\begin{figure*}[t]
\small
\subfloat[Yannakakis' Algorithm]{
     \begin{tikzpicture}[level distance=10em,
     every node/.style = {line width=0.5pt, shape=rectangle, text centered, anchor=center, rounded corners, draw, align=center}, grow=right,
     brace/.style = {black,decorate, decoration = {brace, raise=10pt,amplitude=3pt}}]
     \tikzset{blue dotted/.style={draw=blue!50!white, line width=1pt,
                                    dash pattern=on 1pt off 4pt on 6pt off 4pt,
                                     inner sep=5mm, rectangle, rounded corners}};

     \node (E-exit) [draw=none] {$\relation{Employee}_{\mathsf{join}}(w) \obtainedfrom \relation{Employee}(x, y, z), \relation{Manager}_{\mathsf{join}}(y, x, w).$};

     \node (M-exit) [below = 2.1 cm of E-exit.west, anchor=west, draw=none]{       
           $\relation{Manager}_{\mathsf{join}}(y, x,w)  \obtainedfrom \relation{Manager}(y, x, w), \relation{Contact}_{\mathsf{join}}(y, x).$}; 
     
     \node (C-exit) [below =2.15 cm of M-exit.west, anchor=west, draw=none] {
          $\relation{Contact}_{\mathsf{join}}(y, x)  \obtainedfrom \relation{Contact}(y, x).$};

     \path[->] (E-exit.south west)+(3,-0.2) edge [line width=0.8pt] node (PP) [draw=none,anchor=west] 
     {} ($(3,0.2)+(M-exit.north west)$) ;

     \path[->] (M-exit.south west)+(3,-0.2) edge [line width=0.8pt] node (PP2) [draw=none,anchor=west] 
     {} ($(C-exit.north west)+(3,0.2)$) ;

     \node (box1) [draw=blue, line width=1pt,
                                    dash pattern=on 1pt off 4pt on 6pt off 4pt,
                                     inner sep=1mm, rectangle, rounded corners,  fit = (E-exit)] {};
     
      \node (box2) [draw=red, line width=1pt,
                                    dash pattern=on 1pt off 4pt on 6pt off 4pt,
                                     inner sep=1mm, rectangle, rounded corners,  fit = (M-exit)] {};    

     \node (box3) [draw=brown, line width=1pt,
                                    dash pattern=on 1pt off 4pt on 6pt off 4pt,
                                     inner sep=1mm, rectangle, rounded corners,  fit = (C-exit)] {};
     \node [below=0.4cm of box3, white] {};
     \end{tikzpicture}     
     \label{fig:nonboolean-yannakakis}
}
\quad
\subfloat[PPJT extended to non-Boolean queries]{
     \begin{tikzpicture}[level distance=10em,
     every node/.style = {line width=0.5pt, shape=rectangle, text centered, anchor=center, rounded corners, draw, align=center}, grow=right,
     brace/.style = {black,decorate, decoration = {brace, raise=10pt,amplitude=3pt}}]
     \tikzset{blue dotted/.style={draw=blue!50!white, line width=1pt,
                                    dash pattern=on 1pt off 4pt on 6pt off 4pt,
                                     inner sep=5mm, rectangle, rounded corners}};

     \node (ground) [draw=none] {
          $\mathsf{ground}(w) \obtainedfrom \mathsf{Employee}({x}, y, z), \mathsf{Manager}({y}, x, w), \mathsf{Contact}({y}, x).$};

     \node (E-exit) [below = 1cm of ground.west, anchor=west, draw=none] {
          \textbf{R4: }$\relation{Employee}_{\mathsf{join}}(w) \obtainedfrom \relation{Employee}(x, y, z), \lnot \fk{\relation{Employee}}(x,w), \mathsf{ground}(w).$};

     \node (E-self) [below = 0.6cm of E-exit.west, anchor=west, draw=none] {   
          \textbf{R3: }$\fk{\relation{Employee}}(x,w)  \obtainedfrom \relation{Employee}(x, y, z), \lnot \relation{Manager}_{\mathsf{join}}(y, x,w), \mathsf{ground}(w).$};

     \node (M-exit) [below = 1.5 cm of E-self.west, anchor=west, draw=none]{       
           \textbf{R4: }$\relation{Manager}_{\mathsf{join}}(y, x,w)  \obtainedfrom \relation{Manager}(y, x, w), \lnot \fk{\relation{Manager}}(y, w), \mathsf{ground}(w).$}; 
           
     \node (M-self) [below =0.9cm of M-exit.west, anchor=west, draw=none]{ 
           $\begin{aligned}
     \textbf{R3: }& \fk{\relation{Manager}}(y, w)  \obtainedfrom \relation{Manager}(y, x, w), \lnot \relation{Contact}_{\mathsf{join}}(y, x), \mathsf{ground}(w). \\
     \textbf{R2: }&\fk{\relation{Manager}}(y, w)  \obtainedfrom \relation{Manager}(y, x, w), \relation{Manager}(y, z_1, w), z_1 \neq x. \\
     \textbf{R1: }&\fk{\relation{Manager}}(z_1, w)  \obtainedfrom \relation{Manager}(z_1, z_2, z_3), z_3 \neq w, \mathsf{ground}(w). 
     \end{aligned}$};

     \node (C-exit) [below =1.9 cm of M-self.west, anchor=west, draw=none] {
          \textbf{R4: }$\relation{Contact}_{\mathsf{join}}(y, x)  \obtainedfrom \relation{Contact}(y, x), \lnot \fk{\relation{Contact}}(y).$};

     \node (C-self) [below = 0.6cm of C-exit.west, anchor = west, draw=none] {
          \textbf{R2: }$\fk{\relation{Contact}}(y)  \obtainedfrom \relation{Contact}(y, x), \relation{Contact}(y, z_1), z_1 \neq x.$};

     \path[->] (E-self.south west)+(3,-0.2) edge [line width=0.8pt] node (PP) [draw=none,anchor=west] 
     {} ($(3,0.2)+(M-exit.north west)$) ;


     \path[->] (M-self.south west)+(3,-0.2) edge [line width=0.8pt] node (PP2) [draw=none,anchor=west] 
     {} ($(C-exit.north west)+(3,0.2)$) ;


     \node (box1) [draw=blue, line width=1pt,
                                    dash pattern=on 1pt off 4pt on 6pt off 4pt,
                                     inner sep=1mm, rectangle, rounded corners,  fit = (E-exit) (E-self)] {};
     
      \node (box2) [draw=red, line width=1pt,
                                    dash pattern=on 1pt off 4pt on 6pt off 4pt,
                                     inner sep=1mm, rectangle, rounded corners,  fit = (M-exit) (M-self)] {};    

     \node (box3) [draw=brown, line width=1pt,
                                    dash pattern=on 1pt off 4pt on 6pt off 4pt,
                                     inner sep=1mm, rectangle, rounded corners,  fit = (C-exit) (C-self)] {};                                      
                                                        
     \end{tikzpicture}    
     \label{fig:example-nonboolean}
}
\caption{The non-recursive Datalog program for $\examplequerynb$ and $\cqa{\examplequerynb}$.}

\end{figure*}

\section{Implementation}
\label{sec:implementation}

In this section, we first present \system \footnote{\url{https://github.com/xiatingouyang/LinCQA/}}, a system that produces the consistent \FO-rewriting of a query $q$ in both Datalog and SQL formats if $q$ has a PPJT. Having a rewriting in both formats allows us to use both Datalog and SQL engines as a backend. 
We then briefly discuss how we address the flaws of Conquer and Conquesto that impair their actual runtime performance.
\vspace{-2ex}
\subsection{\system: Rewriting in Datalog/SQL}

Our implementation takes as input a self-join-free CQ $q$ written in either Datalog or SQL. \system\ first checks whether the query $q$ admits a PPJT, and if so, it proceeds to produce the consistent \FO-rewriting of $q$ in either Datalog or SQL, \reva{or it terminates otherwise}. 
\subsubsection{Datalog rewriting}
\label{sec:ground-atom}

\system\ implements all rules introduced in Subsection~\ref{sec:datalog}, with one modification to the ground rule atom. 
Let the input query be 
$$q(\vec{u}) \obtainedfrom R_1(\underline{\vec{x_1}}, \vec{y_1}), R_2(\underline{\vec{x_2}}, \vec{y_2}), \dots, R_k(\underline{\vec{x_k}}, \vec{y_k}).$$ 
In Subsection~\ref{sec:nonboolean}, the head of the ground rule is $\mathsf{ground}(\vec{u})$. 
In the implementation, we replace that rule with
$$\mathsf{ground}^*(\vec{x}_1, \vec{x}_2, \dots, \vec{x}_k, \vec{u}) \obtainedfrom \mathsf{body}(q),$$
keeping the key variables of all atoms. For each unsafe rule with head $R_{i,\mathsf{label}}$ where $\mathsf{label} \in \{\mathsf{fkey}, \mathsf{join}\}$, let $\vec{v}$ be the key in the occurrence of $R_i$ in the body of the rule (if the unsafe rule is produced by \textbf{Rule~\ref{rule:good-join}}, both occurrences of $R_i$ share the same key). Then, we add to the rule body the atom 
$$\mathsf{ground}^*(\vec{z}_1, \dots, \vec{z}_{i-1}, \vec{v}, \vec{z}_{i+1}, \dots, \vec{z}_k, \vec{u})$$ 
where $\vec{z}_i$ is a sequence of fresh variables of the same length as $\vec{x}_i$. 

The rationale is that appending $\mathsf{ground}(\vec{u})$ to all unsafe rules could potentially introduce a Cartesian product between $\mathsf{ground}(\vec{u})$ and some existing atom $R(\vec{v}, \vec{w})$ in the rule. The Cartesian product has size $O(N \cdot |\mathsf{OUT}_p|)$ and would take $\Omega(N \cdot |\mathsf{OUT}_p|)$ time to compute, often resulting in inefficient evaluations or even out-of-memory errors. On the other hand, adding $\mathsf{ground}^*$  guarantees a join with an existing atom in the rule. Hence the revised rules would take $O(N + |\mathsf{ground}^*|)$ time to compute. Note that the size of $\mathsf{ground}^*$ can be as large as $N^k \cdot |\mathsf{OUT}_p|$ in the worst case; but as we observe in the experiments, the size of $\mathsf{ground}^*$ is small in~practice. 

\lstset{aboveskip=1.5pt,belowskip=1.5pt}

\subsubsection{SQL rewriting}
We now describe how to translate the Datalog rules in Subsection~\ref{sec:datalog} to SQL queries. Given a query $q$, we first denote the following:

\begin{enumerate}
\item $\mathbf{KeyAttri(\relation{R})}$: the primary key attributes of relation $\relation{R}$;
\item $\mathbf{JoinAttri(\relation{R}, \relation{T})}$: the attributes of $\relation{R}$ that join with $\relation{T}$;
\item $\mathbf{Comp(\relation{R})}$: the conjunction of comparison predicates imposed entirely on $\relation{R}$, excluding all join predicates (e.g., $\relation{R}.A = 42$ and $\relation{R}.A = \relation{R}.B$); and
\item $\mathbf{NegComp(\relation{R})}$: the negation of $\mathbf{Comp(\relation{R})}$ (e.g., $\relation{R}.A  \neq 42$ or $\relation{R}.A \neq \relation{R}.B$).
\end{enumerate}

\introparagraph{Translation of \textbf{Rule}~\ref{rule:self-pruning}.} 
We translate \textbf{Rule~\ref{rule:self-pruning}} of Subsection~\ref{sec:datalog} into the following SQL query computing the keys of $\relation{R}$. 
\begin{lstlisting}
    SELECT $\mathbf{KeyAttri(\relation{{R}})}$ FROM $\relation{{R}}$ WHERE $\mathbf{NegComp(\relation{{R}})}$
\end{lstlisting}

\introparagraph{Translation of \textbf{Rule}~\ref{rule:good-join}.} 
We first produce the projection on all key attributes and the joining attributes of $R$ with its parent $T$ (if it exists), and then compute all blocks containing at least two facts that disagree on the joining attributes. This can be effectively implemented in SQL with \lstinline!GROUP BY! and \lstinline!HAVING!.
\begin{lstlisting}
 SELECT $\mathbf{KeyAttri(\relation{{R}})}$
 FROM (SELECT DISTINCT $\mathbf{KeyAttri(\relation{{R}})}$, $\mathbf{JoinAttri(\relation{{R}}, \relation{{T}})}$ FROM $\relation{R}$) t
 GROUP BY $\mathbf{KeyAttri(\relation{{R}})}$ HAVING COUNT(*) > 1
\end{lstlisting}

\introparagraph{Translation of \textbf{Rule}~\ref{rule:bad-key-simplified}.}
For \textbf{Rule~\ref{rule:bad-key-simplified}} in the pair-pruning phase, we need to compute all blocks of $R$ containing some fact that does not join with some fact in $S_{\mathsf{join}}$ for some child node $S$ of $R$. This can be achieved through a \emph{left outer join} between $R$ and each of its child node $\relation{S}^1_{\mathsf{join}}$, $\relation{S}^2_{\mathsf{join}}$, $\dots$, $\relation{S}^k_{\mathsf{join}}$, which are readily computed in the recursive steps. 
For each $1 \leq i \leq k$, let the attributes of $\relation{S}^i$ be $B^i_1$, $B^i_2$, $\dots$, $B^i_{m_i}$, joining with attributes $A_{\alpha^i_1}$, $A_{\alpha^i_2}$, $\dots$, $A_{\alpha^i_{m_i}}$ in $\relation{R}$ respectively. We produce the following rule:
\begin{lstlisting}
    SELECT $\mathbf{KeyAttri(\relation{R})}$ FROM $\relation{R}$
    LEFT OUTER JOIN $\relation{S}^1_{\mathsf{join}}$ ON 
      $\relation{R}.A_{\alpha^1_1} = \relation{S}^1_{\mathsf{join}}.B^1_{1}$ AND ... AND $\relation{R}.A_{\alpha^1_{m_1}} = \relation{S}^1_{\mathsf{join}}.B^1_{m_1}$
    ...                
    LEFT OUTER JOIN $\relation{S}^k_{\mathsf{join}}$ ON
      $\relation{R}.A_{\alpha^k_1} = \relation{S}^k_{\mathsf{join}}.B^k_{1}$ AND ... AND $\relation{R}.A_{\alpha^k_{m_k}} = \relation{S}^k_{\mathsf{join}}.B^k_{m_k}$
    WHERE  $\relation{S}^1_{\mathsf{join}}.B^1_{1}$ IS NULL OR ... $\relation{S}^1_{\mathsf{join}}.B^1_{m_1}$ IS NULL OR
           $\relation{S}^2_{\mathsf{join}}.B^2_{1}$ IS NULL OR ... $\relation{S}^2_{\mathsf{join}}.B^2_{m_2}$ IS NULL OR
           ...
           $\relation{S}^k_{\mathsf{join}}.B^k_{1}$ IS NULL OR ... $\relation{S}^k_{\mathsf{join}}.B^k_{m_k}$ IS NULL
\end{lstlisting}

The \textit{inconsistent blocks} represented by the keys found by the above three queries are \textit{combined} using \lstinline!UNION ALL!  (e.g., $\fk{R}$ in \textbf{Rule~\ref{rule:self-pruning},~\ref{rule:good-join},~\ref{rule:bad-key-simplified}}). 


\introparagraph{Translation of \textbf{Rule}~\ref{rule:return-simplified}.} Finally, we translate \textbf{Rule~\ref{rule:return-simplified}} computing the values on join attributes between \textit{good blocks} in $\relation{R}$ and its unique parent $\relation{T}$ if it exists. Let $A_1, A_2, \dots, A_k$ be the key attributes of $R$.
\begin{lstlisting}
    SELECT $\mathbf{JoinAttri(\relation{{R}}, \relation{{T}})}$ FROM $\relation{R}$ WHERE NOT EXISTS ( 
      SELECT * FROM  $\fk{\relation{R}}$
      WHERE $\relation{R}.A_1 = \fk{\relation{R}}.A_1$ AND ... AND $\relation{R}.A_k = \fk{\relation{R}}.A_k$)
\end{lstlisting}
If $\relation{R}$ is the root relation of the PPJT, we replace $\mathbf{JoinAttri(\relation{R}, \relation{T})}$ with \lstinline!DISTINCT 1! (i.e. a Boolean query). Otherwise, the results returned from the above query are stored as $\relation{R}_{\mathsf{join}}$ and the recursive process continues as described in Algorithm~\ref{fig:ppjt}. 

\introparagraph{Extension to non-Boolean queries.} Let $q$ be a non-Boolean query. We use $\mathbf{ProjAttri}(q)$ to denote a sequence of attributes of~$q$ to be projected and let $\mathbf{CompPredicate}(q)$ be the comparison expression in the \lstinline!WHERE! clause of $q$. We first produce the SQL query that computes the facts of $\mathsf{ground^{*}}$. 
\begin{lstlisting}
 SELECT $\mathbf{KeyAttri}(\relation{R}_1)$, $\dots$, $\mathbf{KeyAttri}(\relation{R}_k)$, $\mathbf{ProjAttri}(q)$
 FROM $\relation{R}_1$, $\relation{R}_2$, $\dots$, $\relation{R}_k$ WHERE $\mathbf{CompPredicate}(q)$
\end{lstlisting}

We then modify each SQL statement as follows. Consider a SQL statement whose corresponding Datalog rule is unsafe and let $T(\vec{v}, \vec{w})$ be an atom in the rule body. Let $\vec{u}_T$ be a sequence of free variables in $q_{\tau_T}$ and let $\mathbf{FreeAttri}(T)$ be a sequence of attributes in $q_{\tau_T}$ to be projected (i.e., corresponding to the variables in $\vec{u}_T$). Recall that $T_{\mathsf{join}}(\vec{v})$ and $\fk{T}(\vec{v})$ would be replaced with $T_{\mathsf{join}}(\vec{v}, \vec{u}_T)$ and $\fk{T}(\vec{v}, \vec{u}_T)$ respectively, we thus first append $\mathbf{FreeAttri}(T)$ to the \lstinline!SELECT! clause and then add a \lstinline!JOIN! between table $T$ and $\mathsf{ground}$ on all attributes in $\mathbf{KeyAttri}(T)$. Finally, for a rule that has some negated IDB containing a free variable corresponding to some attribute in $\mathsf{ground}$ (i.e., $\mathsf{ground}.A$), 
\begin{packed_item}
\item if the rule is produced by \textbf{Rule~\ref{rule:bad-key-simplified},} in each \lstinline!LEFT OUTER JOIN! with $S^i_{\mathsf{join}}$ we add 
the expression \lstinline!ground.A = $\relation{S}^i_{\mathsf{join}}.B$! connected by the \lstinline!AND! operator, where $B$ is an attribute to be projected  in $\relation{S}^i_{\mathsf{join}}$. 
In the \lstinline!WHERE! clause we also add an expression \lstinline!ground.A IS NULL!, connected by the \lstinline!OR! operator.
\item if the rule is produced by \textbf{Rule~\ref{rule:return-simplified}}, in the \lstinline!WHERE! clause of the subquery we add an expression \lstinline!ground.A = $\fk{\relation{R}}$.A!.
\end{packed_item} 


\subsection{Improvements upon existing CQA systems}

\label{sec:existing-cqa-systems}

ConQuer~\cite{fuxman2005conquer} and Conquesto~\cite{Conquesto} are two other CQA systems targeting their own subclasses of \FO-rewritable queries, both with noticeable performance issues. For a fair comparison with \system, we implemented our own optimized version of both systems. Specifically, we complement Conquer presented in \cite{fuxman2005conquer} which was only able to handle tree queries (a subclass of $\cforest$), allowing us to handle all queries in  $\cforest$. Additionally,  we optimized Conquesto\cite{Conquesto} to get rid of the unnecessarily repeated computation and the undesired cartesian products produced due to its original formulation. The optimized system has significant performance gains over the original implementation and is named FastFO. 

\section{Experiments}
\captionsetup{skip=1pt}
\label{sec:experiments}
Our experimental evaluation addresses the following questions:
\begin{enumerate}
   \item How do first-order rewriting techniques perform compared to generic state-of-the-art CQA systems (e.g., CAvSAT)?
   \item How does \system\ perform compared to other existing CQA techniques?
   \item How do different CQA techniques behave on inconsistent databases with different properties (e.g., varying inconsistent block sizes, inconsistency)?
   \item Are there instances where we can observe the worst-case guarantee of \system\ that other CQA techniques lack?  
\end{enumerate}

To answer these questions, we perform experiments using synthetic benchmarks used in prior works and a large real-world dataset of 400GB. We compare \system\ against several state-of-the-art CQA systems with improvements. To the best of our knowledge, this is the most comprehensive performance evaluation of existing CQA techniques and we are the first ones to evaluate different CQA techniques on a real-world dataset of this large scale.


\vspace{-1.5ex}
\subsection{Experimental Setup}
\label{sec:setup}
We next briefly describe the setup of our experiments.

\introparagraph{System configuration.} 
All of our experiments are conducted on a bare-metal server in Cloudlab \cite{cloudlab}, a large cloud infrastructure. The server runs Ubuntu 18.04.1 LTS and has
two Intel Xeon E5-2660 v3 2.60 GHz (Haswell EP) processors. Each processor has 10 cores, and 20 hyper-threading
hardware threads. The server has a SATA SSD with 440GB space being available, 160GB memory and each NUMA node is directly attached to 80GB of memory.  We run Microsoft SQL Server 2019 Developer Edition (64-bit) on Linux as the relational backend for all CQA systems. For CAvSAT,  MaxHS v3.2.1 \cite{DBLP:conf/cp/DaviesB11} is used as the solver for the output WPMaxSAT instances. 

\introparagraph{Other CQA systems.}  
We compare the performance of \system\ with several state-of-the-art CQA methods. 
\begin{description}
\item[ConQuer:] a CQA system that outputs a SQL rewriting for queries that are in $\cforest$~\cite{fuxman2005conquer}. 
\item[FastFO:] our own implementation of the general method that can handle any query for which CQA is \FO-rewritable. 
\item[CAvSAT:] a recent SAT-based system. It reduces the complement of CQA with arbitrary denial constraints to a SAT problem, which is solved with an efficient SAT solver~\cite{DBLP:conf/sat/DixitK19}.
\end{description}
For \system, ConQuer and FastFO, we only report execution time of \FO-rewritings, since the rewritings can be produced within $1$ms for all our queries. 
We report the performance of each \FO-rewriting using the best query plan. 
The preprocessing time required by CAvSAT \emph{prior} to computing the consistent answers is not reported. 
For each rewriting and database shown in the experimental results, we run the evaluation five times (unless timed out), discard the first run and report the average time of the last four runs.

\begin{figure*}[t]
\includegraphics[scale=0.15]{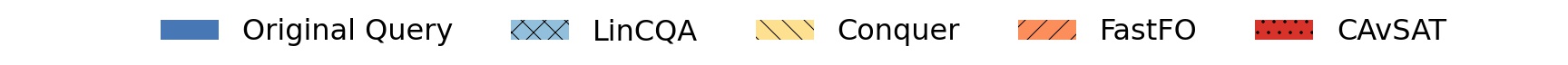}
\vspace{-6ex}	
\end{figure*}
\begin{figure*}[t]
\subfloat[$\mathsf{rSize} = 500\text{K}$]{
	\includegraphics[width=0.32\textwidth]{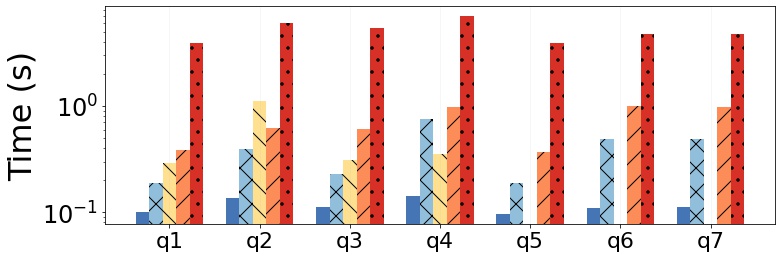}	
}
\subfloat[$\mathsf{rSize} = 1\text{M}$]{
	\includegraphics[width=0.32\textwidth]{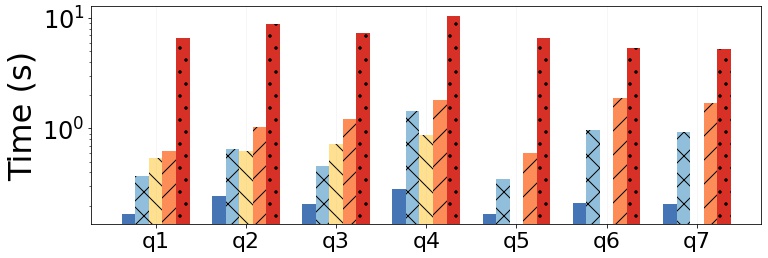}	
}
\subfloat[$\mathsf{rSize} = 5\text{M}$]{
	\includegraphics[width=0.32\textwidth]{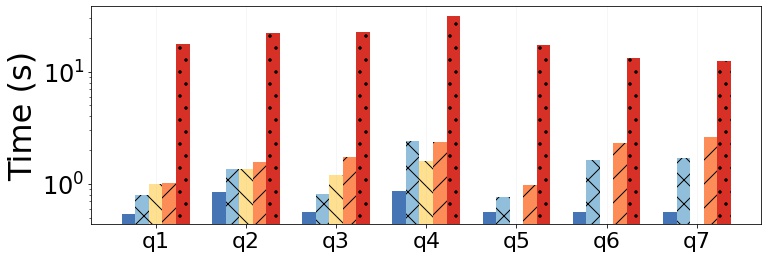}	
}
\caption{Performance comparison of different CQA systems on a synthetic workload with varying relation sizes.}
\label{fig:syn-size}
\end{figure*}

\begin{figure*}[t]
\vspace{-4ex}
\centering
\subfloat[$\mathsf{SF} = 1$]{
	\includegraphics[scale=0.24]{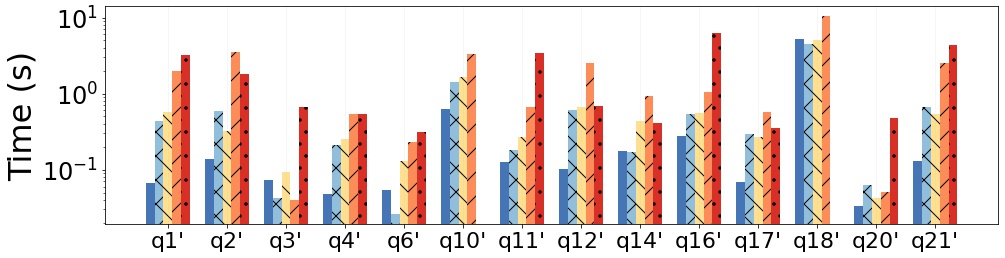}	
}
\subfloat[$\mathsf{SF} = 10$]{
	\includegraphics[scale=0.24]{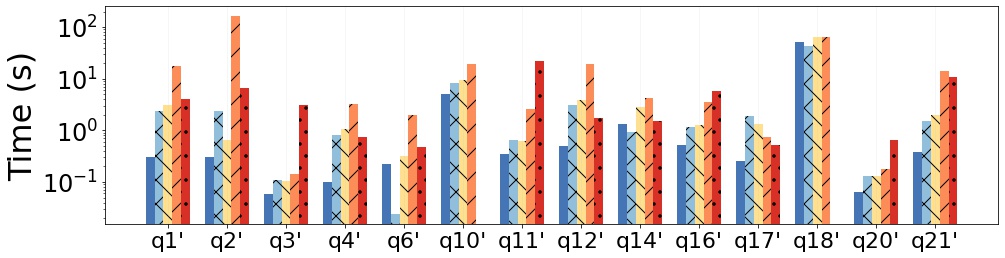}	
}
\caption{Performance comparison of different CQA systems on the TPC-H benchmark with varying scale factor ($\mathsf{SF}$).}
\label{fig:tpch}
\vspace{-2ex}
\end{figure*}

\vspace{-1.5ex}
\subsection{Databases and Queries}
\label{sec:synthetic}

\subsubsection{Synthetic workload}
We consider the synthetic workload used in previous works~\cite{DBLP:journals/pvldb/KolaitisPT13, DBLP:conf/sat/DixitK19, dixit2021answering}. Specifically, we take the seven queries that are consistent first-order rewritable in \cite{DBLP:conf/sat/DixitK19, DBLP:journals/pvldb/KolaitisPT13,dixit2021answering}. These queries feature joins between primary-key attributes and foreign-key attributes, as well as projections on non-key attributes:
\begin{align*}
q_{1}(z) \obtainedfrom  &\mathsf{R}_1(\underline{x},  y, z), \mathsf{R}_3(\underline{y}, v, w). \\
q_{2}(z, w) \obtainedfrom  &\mathsf{R}_1(\underline{x},  y, z), \mathsf{R}_2(\underline{y}, v, w). \\
q_{3}(z) \obtainedfrom  &\mathsf{R}_1(\underline{x},  y, z), \mathsf{R}_2(\underline{y}, v, w), \mathsf{R}_7(\underline{v}, u, d). \\
q_{4}(z, d) \obtainedfrom  &\mathsf{R}_1(\underline{x},  y, z), \mathsf{R}_2({\underline{y}, v, w}), \mathsf{R}_7(\underline{v}, u, d). \\
q_{5}(z) \obtainedfrom  &\mathsf{R}_1(\underline{x},  y, z), \mathsf{R}_8(\underline{y, v}, w). \\
q_{6}(z) \obtainedfrom  &\mathsf{R}_1(\underline{x},  y, z), \mathsf{R}_6(\underline{t}, y, w), \mathsf{R}_9(\underline{x}, y, d). \\
q_{7}(z) \obtainedfrom  &\mathsf{R}_3(\underline{x},  y, z), \mathsf{R}_4(\underline{y}, x, w), \mathsf{R}_{10}(\underline{x}, y, d).
\end{align*}
The synthetic instances are generated in two phases. In the first phase, we generate the consistent instance, while in the second phase we inject inconsistency. We use the following parameters for data generation:
$(i)$~\textsf{rSize}: the number of tuples per relation, $(ii)$~\textsf{inRatio}:  the ratio of the number of tuples that violate primary key constraints (i.e., number of tuples that are in inconsistent blocks) to the total number of tuples of the database, and $(iii)$~\textsf{bSize}: the number of inconsistent tuples in each inconsistent block.

\introparagraph{Consistent data generation.} Each relation in the consistent database has the same number of tuples, so that injecting inconsistency with specified $\mathsf{bSize}$ and $\mathsf{inRatio}$ makes the total number of tuples in the relation equal to $\mathsf{rSize}$. 
The data generation is \textit{query-specific}: for each of the seven queries, the data is generated in a way to ensure the output size of the original query 
on the consistent database is reasonably large. To achieve this purpose,  when generating the database instance for one of the seven queries, we ensure that for any two relations that join on some attributes, the number of matching tuples in each relation is approximately $25\%$;  for the third attribute in each ternary relation that does not participate in a join but is sometimes  \revjef{present in the final projection}, the values are chosen uniformly from the range $[1, \mathsf{rSize}/10]$. \vspace{0.5ex}

\introparagraph{Inconsistency injection.} In each relation, we first select a number of primary keys (or number of inconsistent blocks $\mathsf{inBlockNum}$) from the generated consistent instance.
Then, for each selected primary key, the inconsistency is injected by inserting the \textit{same number of additional tuples} ($\mathsf{bSize} - 1$) into each block. The parameter $\mathsf{inBlockNum}$
is calculated by the given $\mathsf{rSize}, \mathsf{inRatio}$ and $\mathsf{bSize}$:  $\mathsf{inBlockNum} = (\mathsf{inRatio} \cdot \mathsf{rSize}) / \mathsf{bSize}$.
\revb{We remark that there are alternative inconsistency injection methods available \cite{DBLP:journals/pvldb/ArocenaGMMPS15,DBLP:conf/icde/AntovaJKO08}.}

\subsubsection{TPC-H benchmark.}
\label{sec:tpch}
We also altered the 22 queries from the original TPC-H benchmark \cite{poess2000new} by removing aggregation, nested subqueries and selection predicates other than constant constraints, yielding 14 simplified conjunctive queries, namely queries $q_1'$, $q'_2$, $q'_3$, $q'_4$, $q'_6$, $q'_{10}$, $q'_{11}$, $q'_{12}$, $q'_{14}$, $q'_{16}$, $q'_{17}$, $q'_{18}$, $q'_{20}$, $q'_{21}$. All of the 14 queries are in $\cforest$ and hence each query has a PPJT, meaning that they can be handled by both ConQuer and \system. 

We generate the inconsistent instances by injecting inconsistency into the TPC-H databases of scale factor ($\mathsf{SF}$) 1 and 10 in the same way as described for the synthetic data. The only difference is that for a given consistent database instance, instead of fixing $\mathsf{rSize}$ for the inconsistent database, we determine the number of inconsistent tuples to be injected based on the size of the consistent database instance, the specified $\mathsf{inRatio}$ and $\mathsf{bSize}$.

\begin{table}
\footnotesize
\caption{A summary of the Stackoverflow Dataset}
\label{tbl:stackoverflow-schema}
\begin{tabular}{l c c c c}
Table & \# of rows ($\mathsf{rSize}$) & $\mathsf{inRatio}$ & max.\ $\mathsf{bSize}$ & \# of Attributes \\
\hline
Users & 14,839,627 & 0\% & 1 & 14 \\
Posts & 53,086,328 & 0\% & 1 & 20 \\
PostHistory & 141,277,451 & 0.001\% & 4 & 9 \\
Badges  & 40,338,942 & 0.58\% & 941 & 4 \\
Votes & 213,555,899 & 30.9\% & 1441 & 6
\end{tabular}
\end{table} 

\begin{table}
\caption{StackOverflow queries}
\label{tbl:stackoverflow-query}
\footnotesize	
\begin{tabular}{c p{7.5cm}}
$Q_1$ & \texttt{SELECT DISTINCT P.id, P.title FROM Posts P, Votes V WHERE P.Id = V.PostId AND P.OwnerUserId = V.UserId AND BountyAmount > 100} \\
$Q_2$ & \texttt{SELECT DISTINCT U.Id, U.DisplayName FROM Users U, Badges B WHERE U.Id = B.UserId AND B.name = "Illuminator"} \\
$Q_3$ & \texttt{SELECT DISTINCT U.DisplayName FROM Users U, Posts P WHERE U.Id = P.OwnerUserId AND P.Tags LIKE "<c++>"} \\
$Q_4$ & \texttt{SELECT DISTINCT U.Id, U.DisplayName FROM Users U, Posts P, Comments C WHERE C.UserId = U.Id AND C.PostId = P.Id AND P.Tags LIKE "\%SQL\%" AND C.Score > 5} \\
$Q_5$ & \texttt{SELECT DISTINCT P.Id, P.Title FROM Posts P, PostHistory PH, Votes V, Comments C WHERE P.id = V.PostId AND P.id = PH.PostId AND P.id = C.PostId AND P.Tags LIKE "\%SQL\%" AND V.BountyAmount > 100 AND PH.PostHistoryTypeId = 2 AND C.score = 0} 
\end{tabular}
\vspace{-4ex}
\end{table} 

\subsubsection{Stackoverflow Dataset.}
\label{sec:stackoverflow}
We obtained the \url{stackoverflow.com} metadata as of 02/2021 from the Stack Exchange Data Dump, with 551,271,294 rows taking up 400GB. \footnote{\url{https://archive.org/details/stackexchange}}\footnote{https://sedeschema.github.io/} The database tables used are summarized in Table~\ref{tbl:stackoverflow-schema}. 
We remark that the {\em Id} attributes in PostHistory, Comments, Badges, and Votes are surrogate keys and therefore not imposed as natural primary keys; instead, we properly choose composite keys as primary keys \reva{(or quasi-keys)}.
Table~\ref{tbl:stackoverflow-query} shows the five queries used in our CQA evaluation, where the number of tables joined together increases from~$2$ in $Q_{1}$ to~$4$ in $Q_{5}$. 

\begin{figure}[!h]
\vspace{-2ex}
\includegraphics[width=0.4\textwidth]{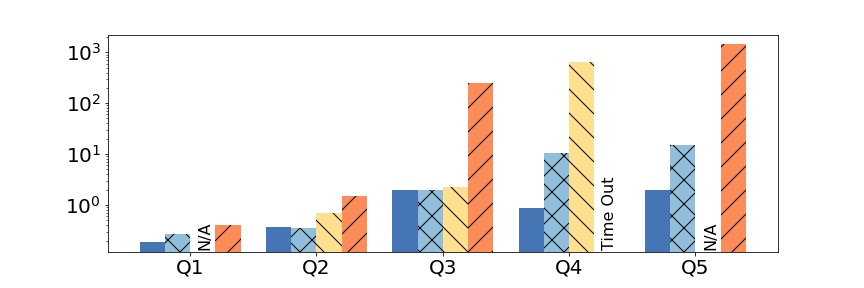}	
\caption{Runtime Comparison on StackOverflow}
\label{fig:stackoverflow_runtime}
\end{figure}

\subsection{Experimental Results}\label{sec:exp_results}
In this section, we report the evaluation of \system\ and the other CQA systems on synthetic workloads and the StackOverflow dataset. Table~\ref{tbl:every-dataset} summarizes the \revjef{number} of consistent and \revjef{possible} answers for each query in the selected datasets.

\introparagraph{Fixed inconsistency with varying relation sizes.}
To compare \system\ with other CQA systems, we evaluate all systems using both the synthetic workload and the altered TPC-H benchmark with fixed inconsistency ($\mathsf{inRatio} = 10\%$, $\mathsf{bSize} = 2$) as in previous works~\cite{DBLP:journals/pvldb/KolaitisPT13, DBLP:conf/sat/DixitK19, dixit2021answering}. We vary the size of each relation ($\mathsf{rSize} \in \{500K, 1M, 5M\}$) in the synthetic data (Figure~\ref{fig:syn-size}) and we evaluate on TPC-H database instances of scale factors 1 and 10 (Figure~\ref{fig:tpch}). Both figures include the time for running the original query on the inconsistent database (which returns the possible answers).

In the synthetic dataset, all three systems based on \FO-rewriting techniques outperform CAvSAT, often by an order of magnitude. This observation shows that if $\cqa{q}$ is \FO-rewritable, a properly implemented rewriting is more efficient than the generic algorithm in practice, refuting some observations in \cite{DBLP:conf/sat/DixitK19, DBLP:journals/pvldb/KolaitisPT13}.  Compared to ConQuer, \system\ performs better or comparably on $q_1$ through $q_4$. 
\revm{
\system\ is also more efficient than ConQuer for $q_1, q_2$ and $q_3$. As the database size increases, the relative performance gap between \system\ and ConQuer reduces for $q_4$. 
ConQuer cannot produce the SQL rewritings for queries $q_5, q_6$ and $q_7$ since they are not in $\cforest$. In summary, \system\ is more efficient and at worst competitive to ConQuer on relatively small databases with less than $5M$ tuples, and is applicable to a wider class of acyclic queries.
}

\begin{table*}
\tiny
\caption{The \revjef{number} of consistent and \revjef{possible} answers for each query in selected datasets.}
\label{tbl:every-dataset}
\small
\begin{tabular}{c | c c c c c c c | c c c c c c c}
& \multicolumn{7}{c|}{Synthetic ($\mathsf{rSize} = 5M$, $\mathsf{inRatio} = 10\%$, $\mathsf{bSize} = 2$)} & \multicolumn{5}{c}{StackOverflow} \\
 & $q_1$ & $q_2$ & $q_3$ & $q_4$ & $q_5$ & $q_6$ & $q_7$ & $Q_1$ & $Q_2$ & $Q_3$ & $Q_4$ & $Q_5$ \\
\hline
\revjef{\# cons.}  & 
311573 & 463459 & 290012 & 408230 & 311434 & 277287 & 277135 & 27578 & 145 & 38320 & 3925 & 1245 \\
\revjef{\# poss.} & 
571047 & 572244 & 534011 & 534953 & 574615 & 504907 & 474203 & 27578 & 145 & 38320 & 3925 & 1250 \\
\hline
& \multicolumn{14}{c}{TPC-H ($\mathsf{SF} = 10$)} \\
& $q_1'$ & $q_2'$ & $q_3'$ & $q_4'$ & $q_6'$ & $q_{10}'$ & \multicolumn{1}{c}{$q_{11}'$} & $q_{12}'$ & $q_{14}'$ & $q_{16}'$ & $q_{17}'$ & $q_{18}'$ & $q_{20}'$ & $q_{21}'$ \\
\hline
\revjef{\# cons.} & 4 & 28591 & 0 & 5 & 1 & 901514 & \multicolumn{1}{c}{289361} & 7 & 1 & 187489 & 1 & 13465732 & 3844 & 3776 \\
\revjef{\# poss.} & 4 & 35206 & 0 & 5 & 1 & 1089754 & \multicolumn{1}{c}{318015} & 7 & 1 & 187495 & 1 & 16617583 & 4054 & 4010
\end{tabular}
\end{table*}


In the TPC-H benchmark, the CQA systems are much closer in terms of performance. In this experiment, we observe that \system\ almost always produces the fastest rewriting, and even when it is not, its performance is comparable to the other baselines. It is also worth noting that for most queries in the TPC-H benchmark, the overhead over running the SQL query directly is much smaller when compared to the synthetic benchmark. 
Note that CAvSAT times out after 1 hour for queries $q'_{10}$ and $q'_{18}$ for both scale $1$ and $10$, while the systems based on \FO-rewriting techniques terminate. 
We also remark that for Boolean queries, CAvSAT will terminate \revjef{at an early stage} without processing the inconsistent part of the database using SAT solvers if the consistent part of the database already satisfies the query (e.g., $q_6'$, $q_{14}'$, $q_{17}'$ in TPC-H).
Overall, both \system\ and ConQuer perform better than FastFO, since they both are better at exploiting the structure of the join tree. 
\revm{We also note that ConQuer and LinCQA exhibit comparable performances on most queries in TPC-H.}
To compute the consistent answers for a certain query, we note that the actual runtime performance heavily depends on the query plan chosen by the query optimizer besides the SQL rewriting given, thus we focus on the overall performance of different CQA systems rather than a few cases in which the performance difference between different systems is relatively small.

\introparagraph{Fixed relation size with varying inconsistency.}
We perform experiments to observe how different CQA systems react when the inconsistency of the instance changes. Using synthetic data, we first fix $\mathsf{rSize} = 1$M, $\mathsf{bSize} = 2$ and run all CQA systems on databases instances of varying inconsistency ratio from $\mathsf{inRatio} = 10\%$ to $\mathsf{inRatio} = 100\%$. The results are depicted in Figure~\ref{fig:syn-incon}. We observe that the running time of CAvSAT increases when the inconsistency ratio of the database instance becomes larger. This happens because the SAT formula grows with larger inconsistency, and hence the SAT solver becomes slower. In contrast, the running time of all \FO-rewriting techniques is relatively stable across database instances of different inconsistency ratios. More interestingly, the running time of \system\ decreases when the inconsistency ratio becomes larger. This behavior occurs because of the early pruning on the relations at lower levels of the PPJT, which shrinks the size of candidate space being considered at higher levels of the PPJT and thus reduces the overall computation time. 
The overall performance trends of different systems are similar for all queries and thus we present only figures of $q_1, q_3, q_5, q_7$  here due to the space limit.  

In our next experiment, we fix the database instance size with $\mathsf{rSize} = 1$M and inconsistency ratio with $\mathsf{inRatio} = 10\%$, running  all CQA systems on databases of varying inconsistent block size $\mathsf{bSize}$ from $2$ to $10$. 
We observe that the performance of all CQA systems is not very sensitive to the change of inconsistent block sizes and thus we omit the results here due to the space limit. \reva{Figure~\ref{fig:syn-incon-rest} and~\ref{fig:syn-block} in the Appendix present the full results.}

\begin{figure*}[t]
\includegraphics[scale=0.15]{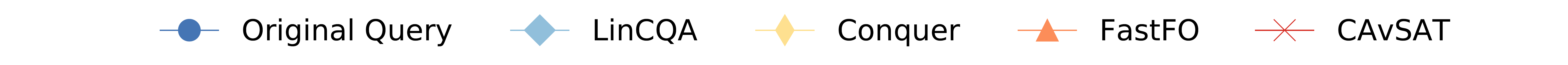}
\vspace{-6ex}	
\end{figure*}
\begin{figure*}[!h]
\subfloat[$q_1$]{
	\includegraphics[width=0.24\textwidth]{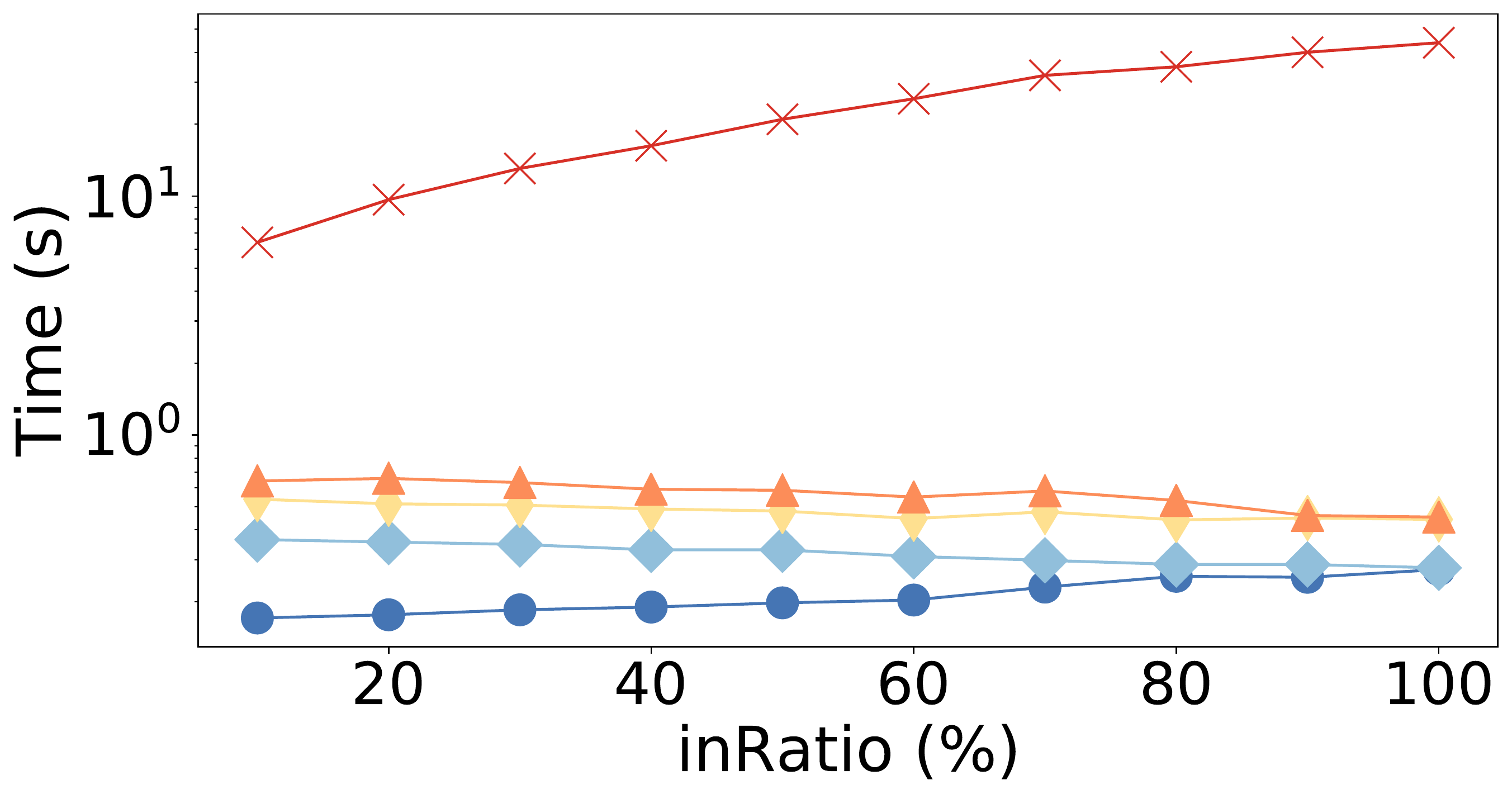}	
}
\subfloat[$q_3$]{
	\includegraphics[width=0.24\textwidth]{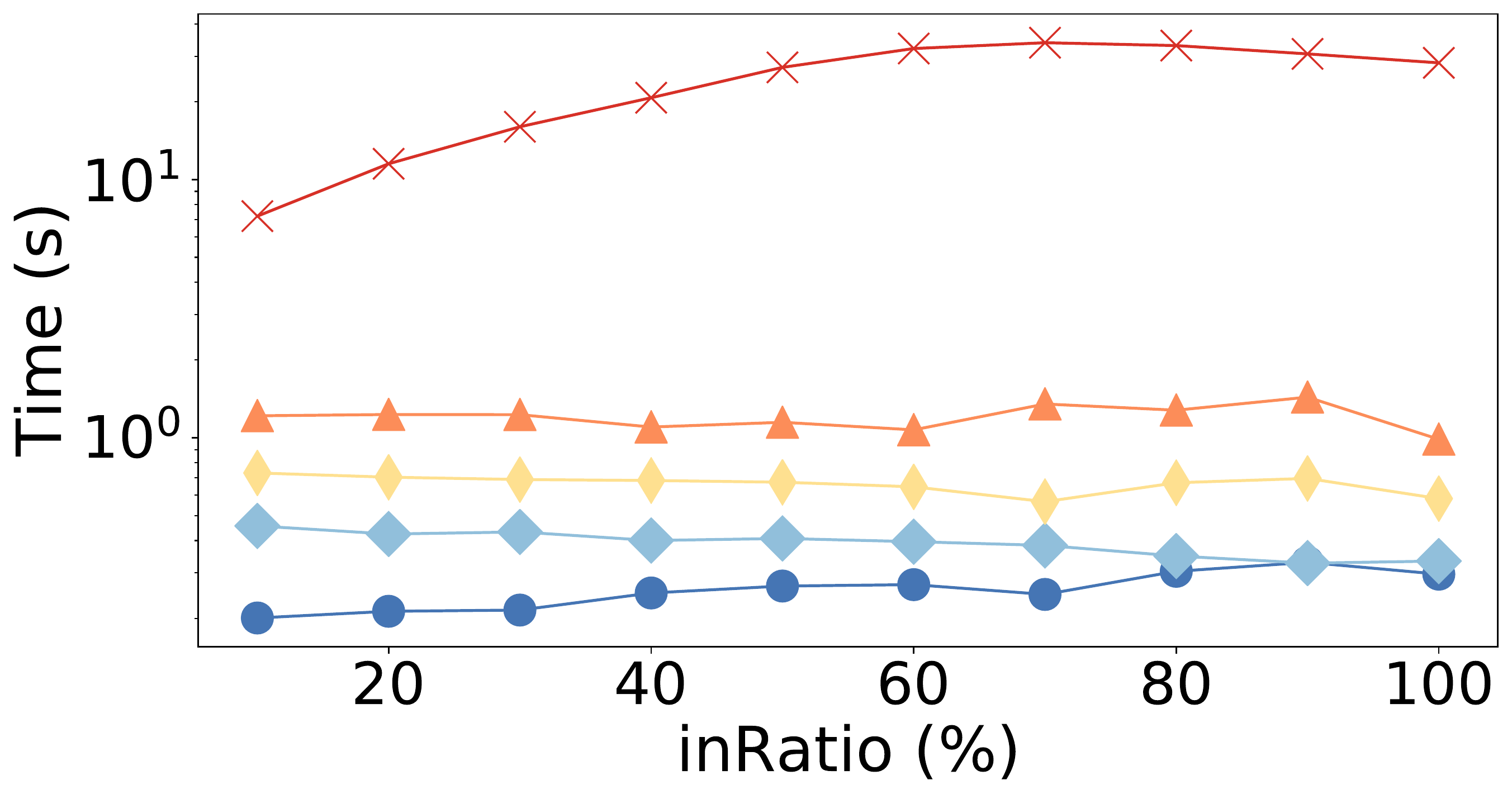}	
}
\subfloat[$q_5$]{
	\includegraphics[width=0.24\textwidth]{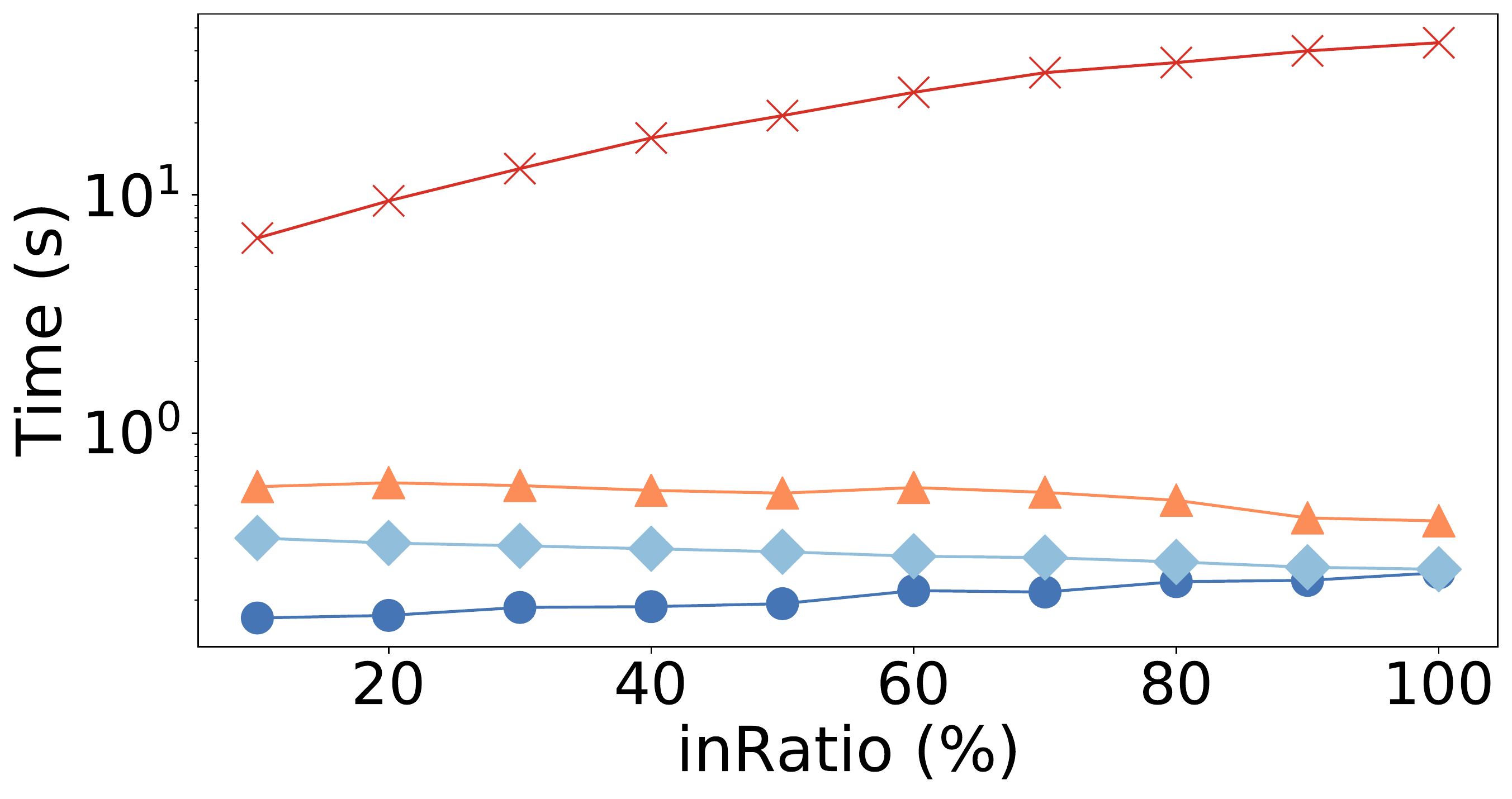}	
}
\subfloat[$q_7$]{
	\includegraphics[width=0.24\textwidth]{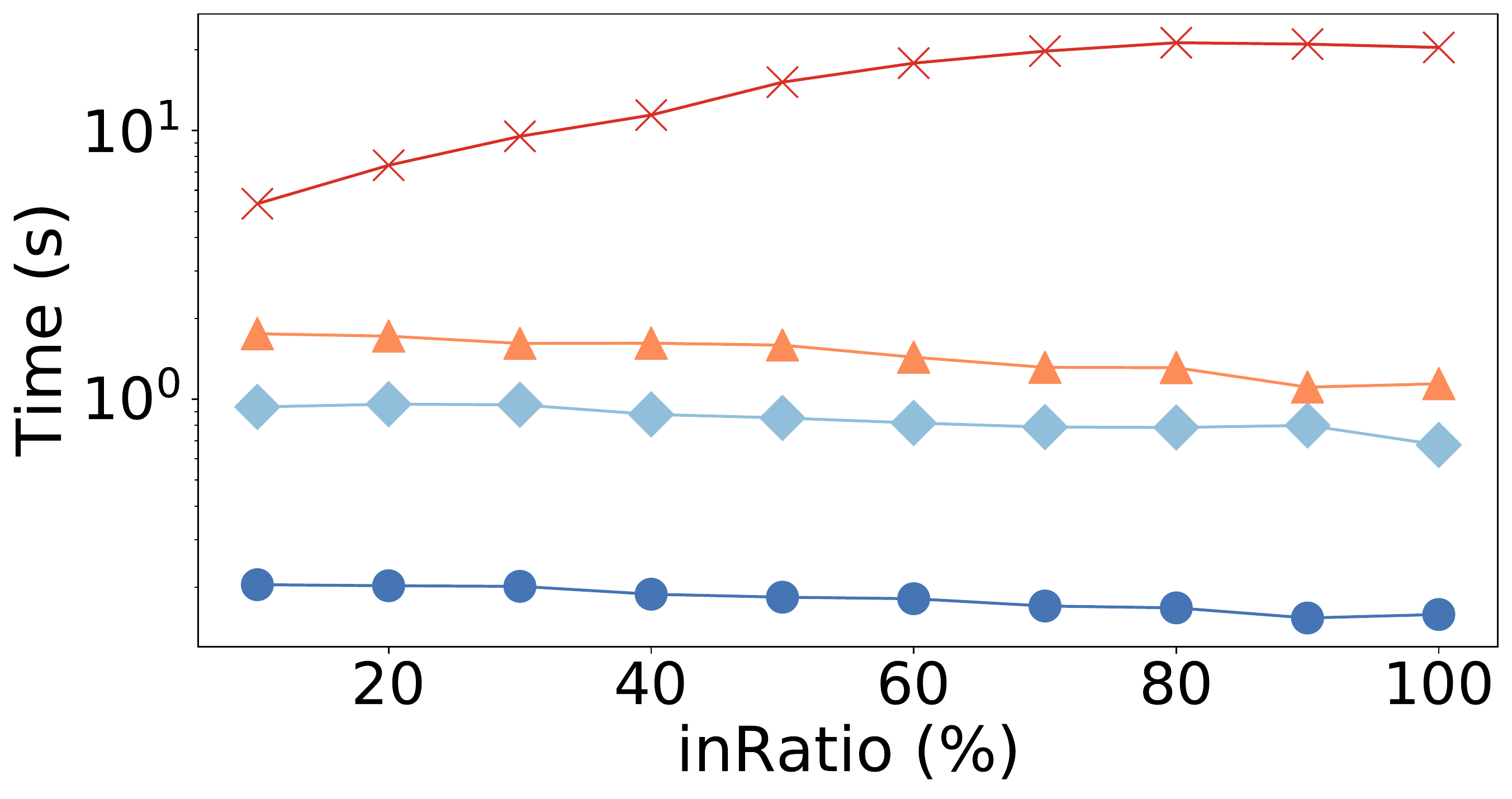}	
}
\caption{Performance of different systems on inconsistent databases with varying inconsistency ratio}
\label{fig:syn-incon}
\end{figure*}

\introparagraph{StackOverflow Dataset}
We use a 400GB StackOverflow dataset to evaluate the performance of different systems on large-scale real-world datasets. \revc{Another motivation to use such a large dataset is that \system\ and ConQuer exhibit comparable performance on the medium-sized synthetic and TPC-H datasets.} CAvSAT is excluded since it requires extra storage for preprocessing which is beyond the limit of the available disk space.
Since $Q_1$ and $Q_5$ are not in $\cforest$, ConQuer cannot handle them and their execution times are marked as ``N/A''. Query executions that do not finish within one hour are marked as ``Time Out''. We observe that on all five queries, \system\ significantly outperforms other competitors. 
\revm{In particular, when the database size is very large, \system\ is much more scalable than ConQuer due to its more efficient strategy. 
We intentionally select queries with small \revjef{possible} answer sizes for ease of experiments and presentation. Some queries with \revjef{possible} answer size up to $1\text{M}$ would require hours to be executed and it is prohibitive to measure the performances of our baseline systems.}
For queries that ConQuer ($Q_4$) and FastFO ($Q_3$, $Q_5$) take long to compute,  \system\ manages to finish execution quickly thanks to its efficient self-pruning and pair-pruning steps. 

\revc{ To see the performance change of different systems when executing in small available memory, we run the experiments on a SQL server with maximum allowed memory of $120$GB, $90$GB, $60$GB, $30$GB, and $10$GB  
respectively. Figure~\ref{fig:stackoverflow-memory} shows that, despite the memory reduction, \system\ is still the best performer on all five queries given different amounts of available memory. No obvious performance regression is observed on  $Q_1$ and $Q_2$ when reducing memory since both queries access only two tables.} 

\begin{figure*}[t]
\vspace{-5ex}
\subfloat[$Q_1$]{
     \includegraphics[width=0.19\textwidth]{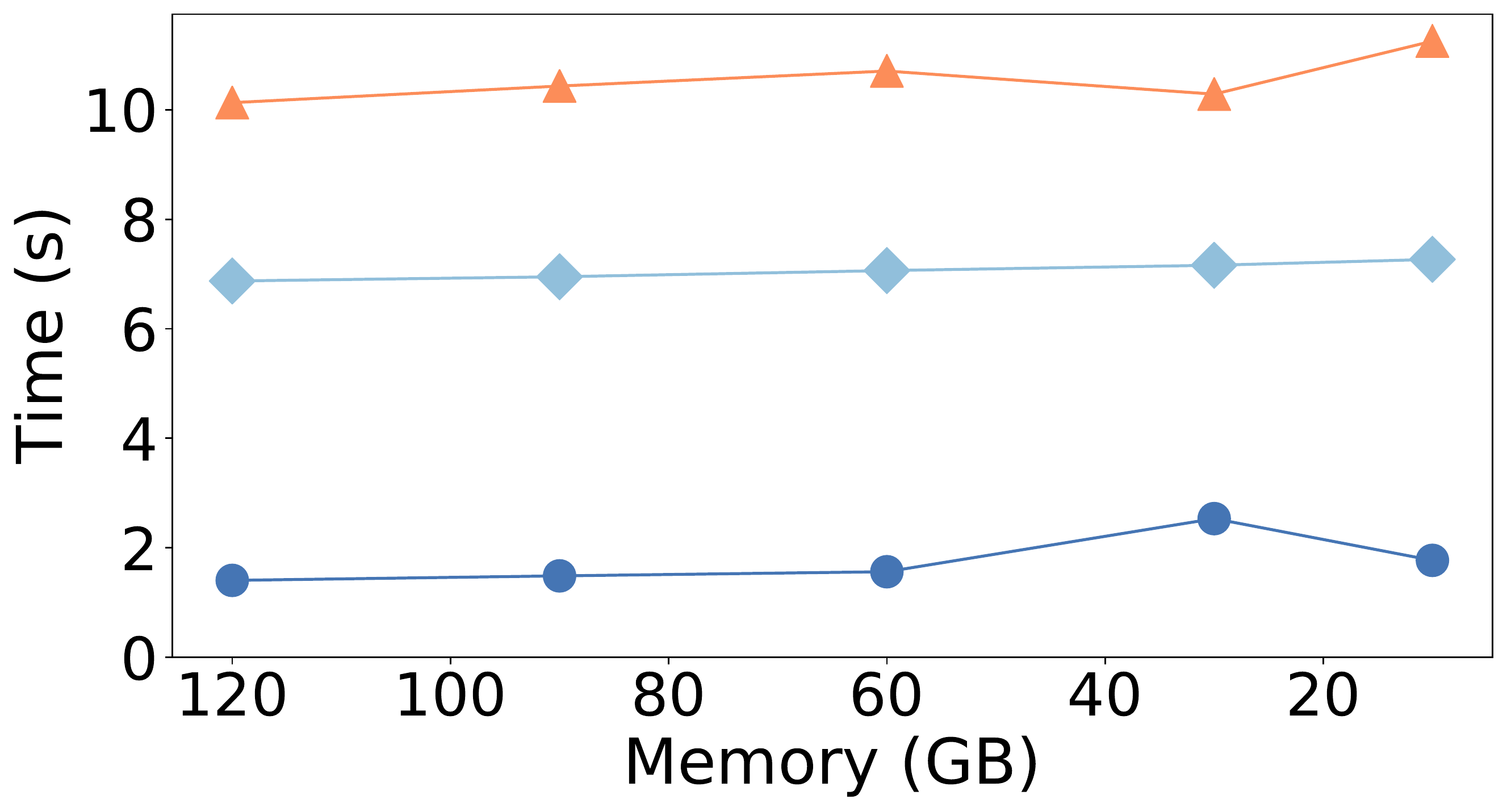} 
}
\subfloat[$Q_2$]{
     \includegraphics[width=0.19\textwidth]{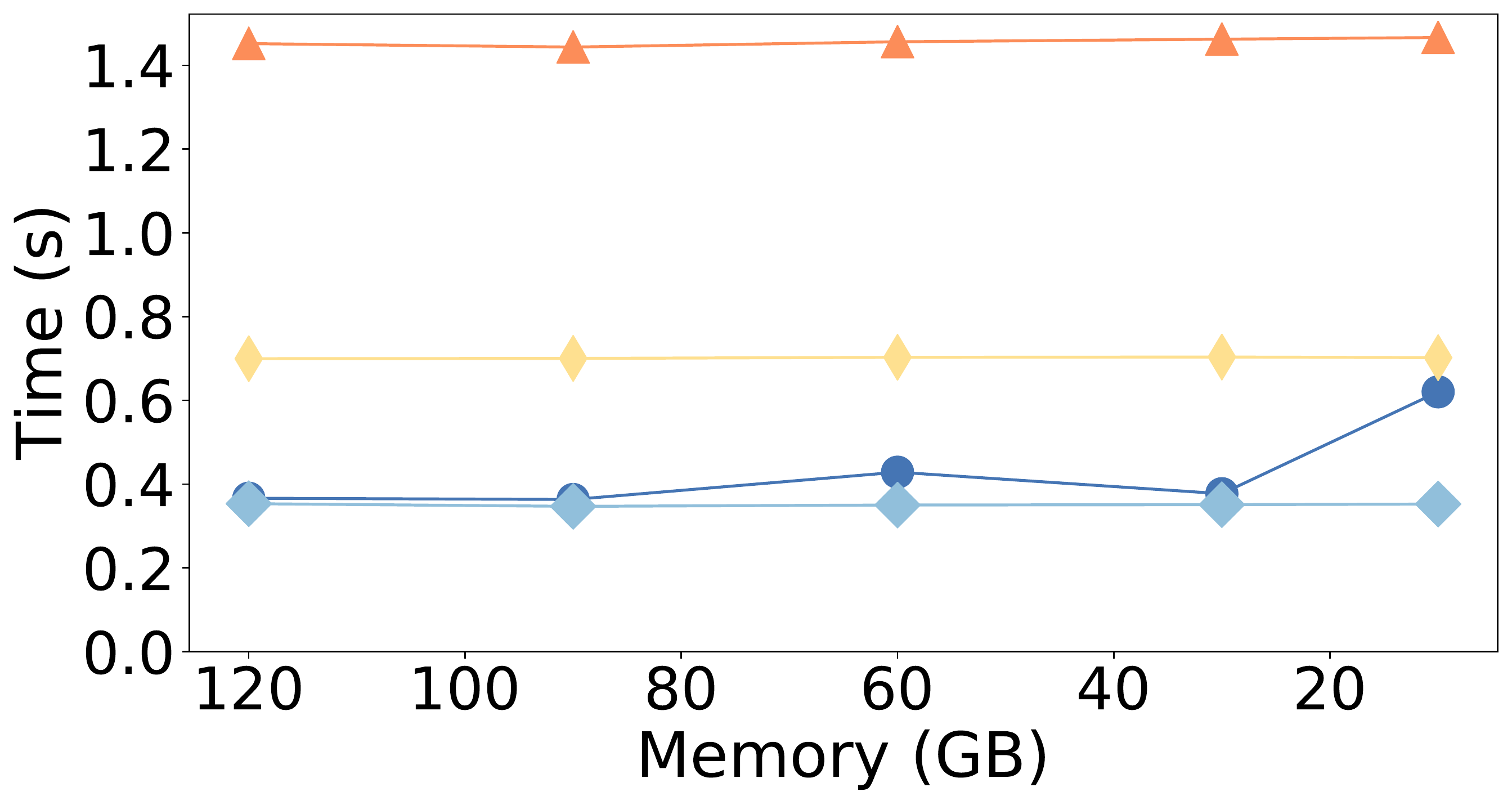} 
}
\subfloat[$Q_3$]{
     \includegraphics[width=0.19\textwidth]{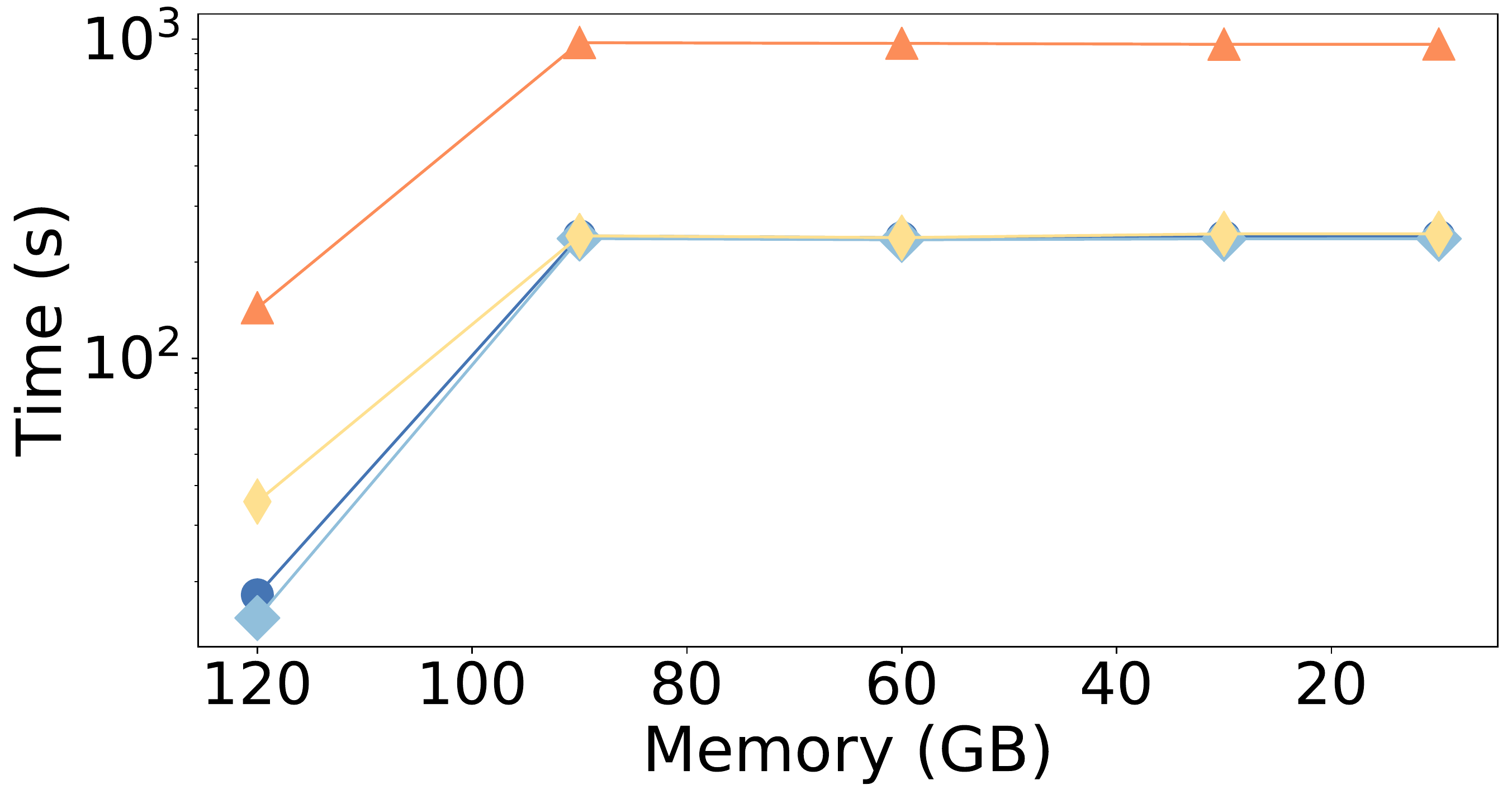} 
}
\subfloat[$Q_4$]{
     \includegraphics[width=0.19\textwidth]{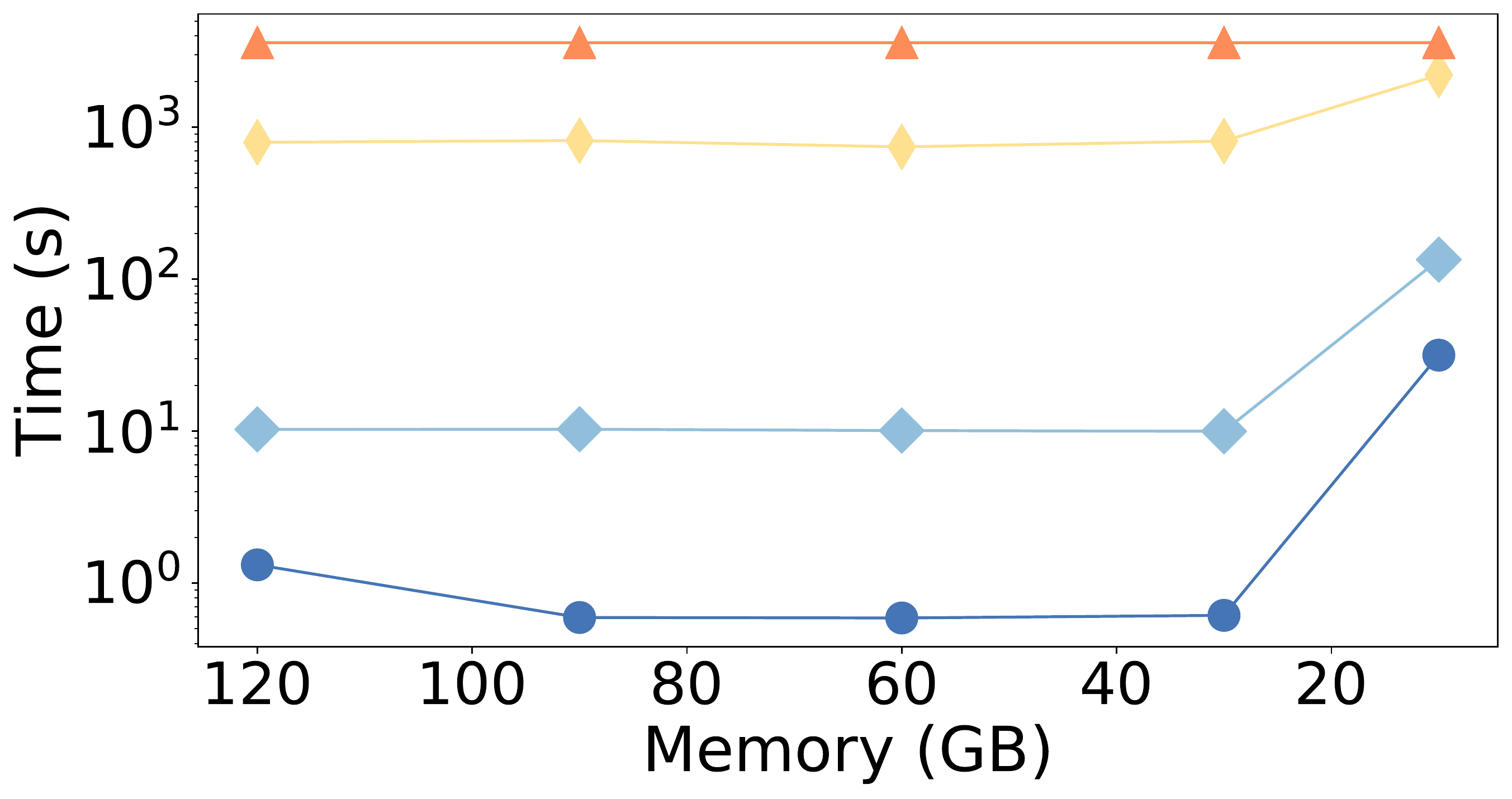} 
}
\subfloat[$Q_5$]{
     \includegraphics[width=0.19\textwidth]{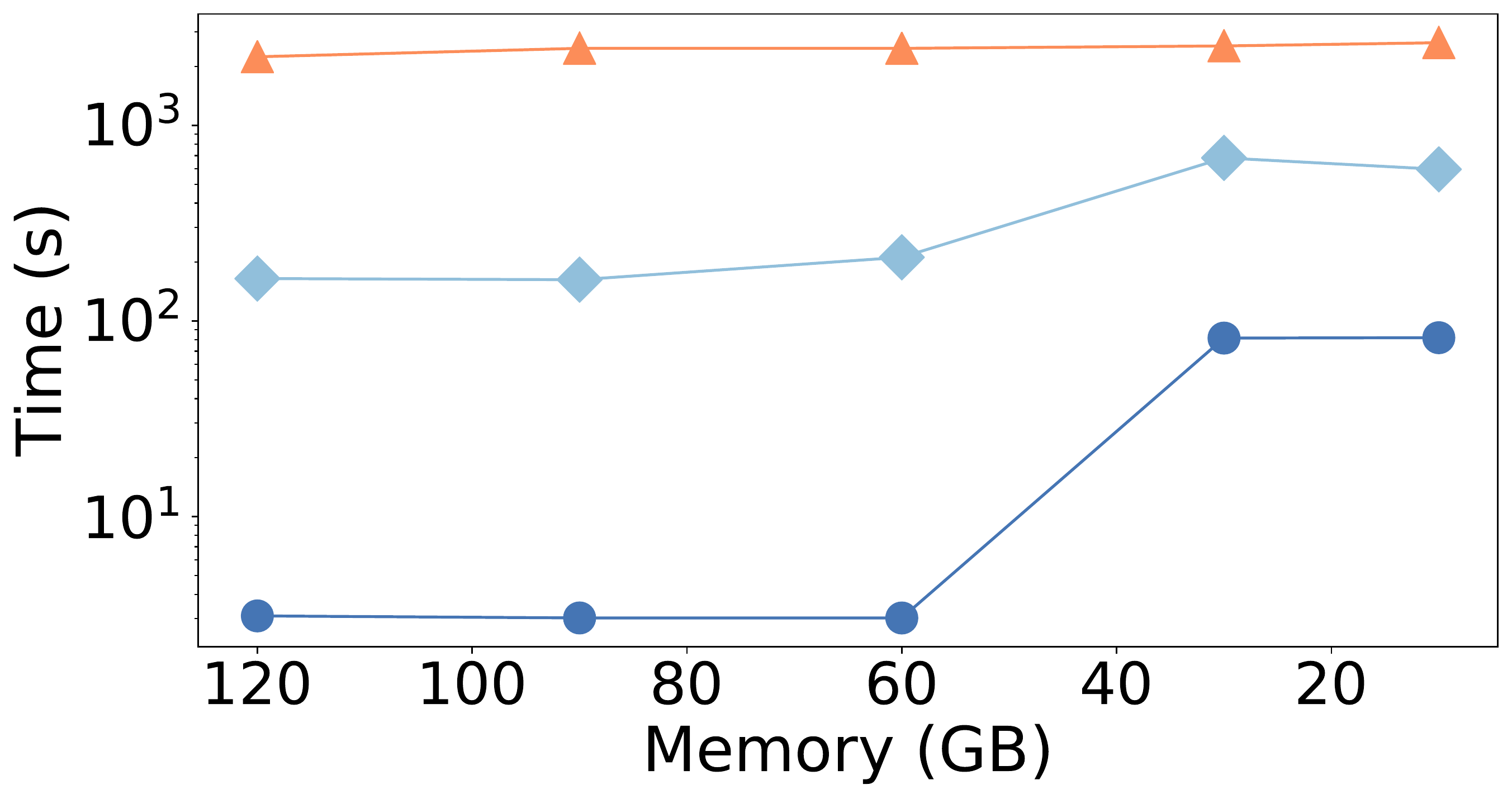} 
}
\caption{\revm{Performance of StackOverflow queries with varying amount of available memory}}
\label{fig:stackoverflow-memory}
\end{figure*}


\revm{
\introparagraph{Summary}
Our experiments show that both \system\ and ConQuer outperform FastFO and CAvSAT, systems that produce generic \FO-rewritings and reduce to SAT respectively. Despite \system\ and ConQuer showing a similar performance on most queries in our experiments, we observe that \system\ is (1) applicable to a wider class of acyclic queries than ConQuer and (2) more scalable than ConQuer when the database size increases significantly. 
}
\begin{figure*}[t]
\vspace{-5ex}	
\subfloat[$Q_{\mathsf{2-path}}$]{
	\includegraphics[width=0.24\linewidth]{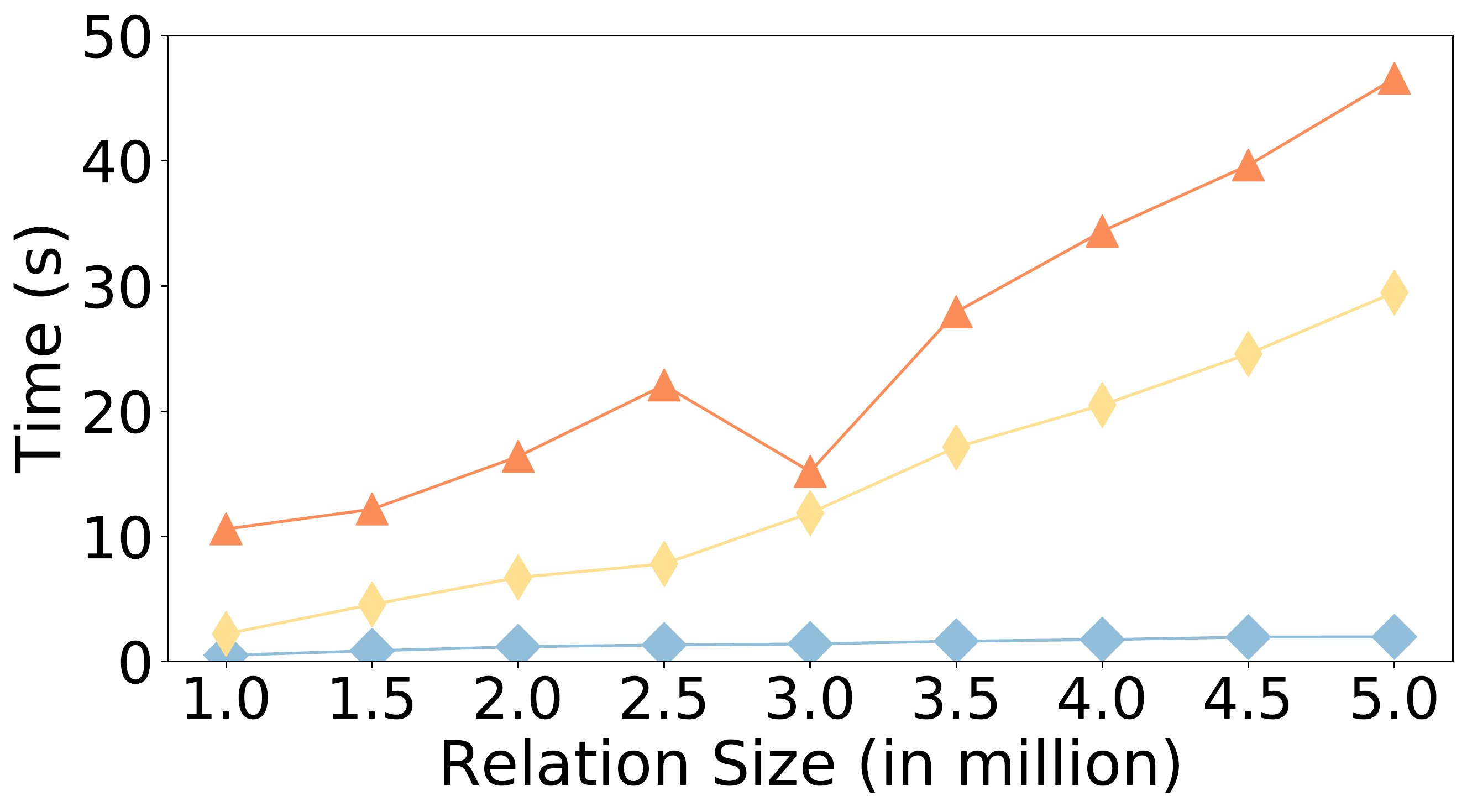} \label{fig:worst-case-size-2path}
}
\subfloat[$Q_{\mathsf{3-path}}$]{
	\includegraphics[width=0.24\linewidth]{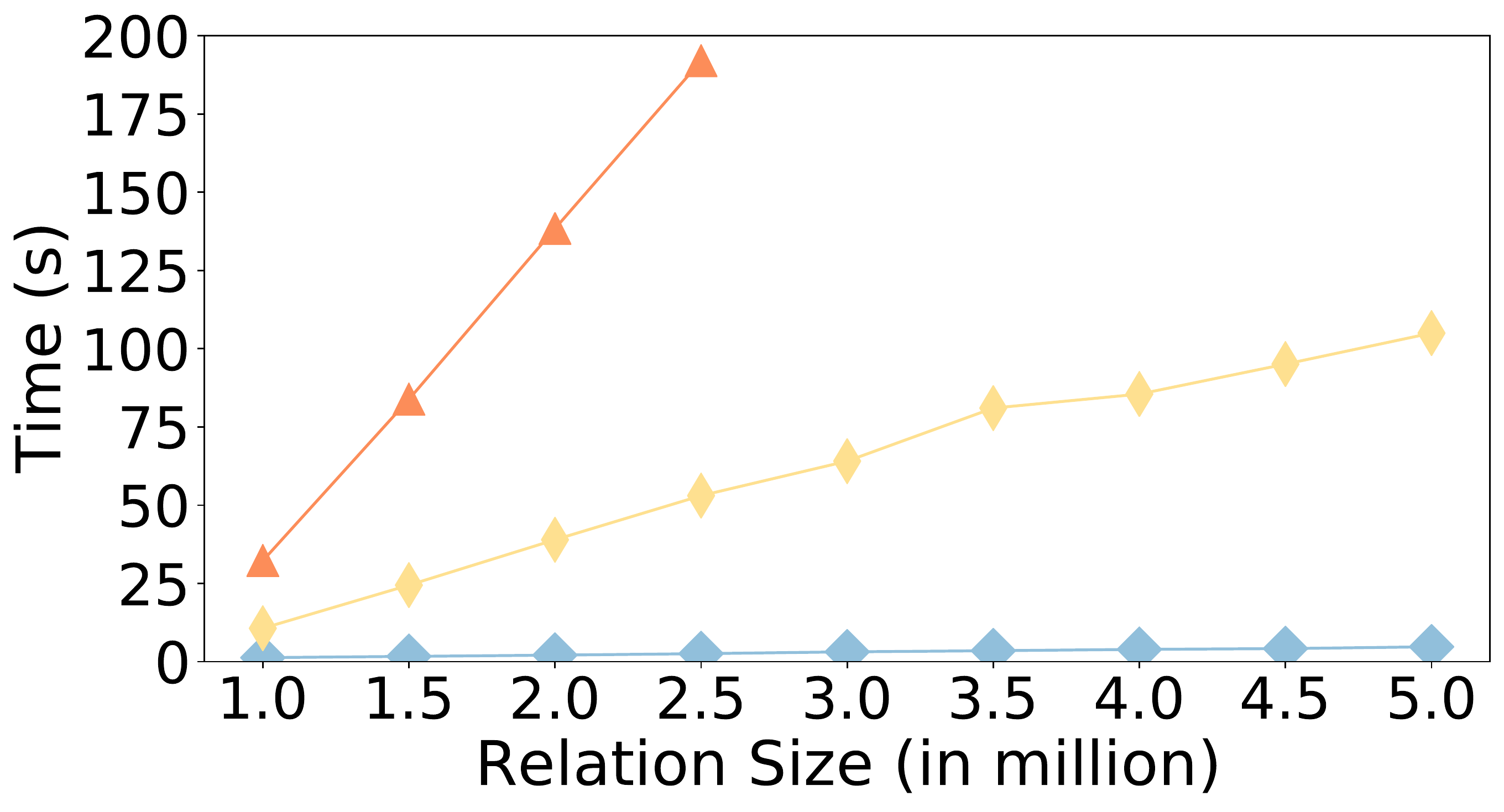}\label{fig:worst-case-size-3path}
}
\subfloat[$Q_{\mathsf{2-path}}$ ]{
	\includegraphics[width=0.24\linewidth]{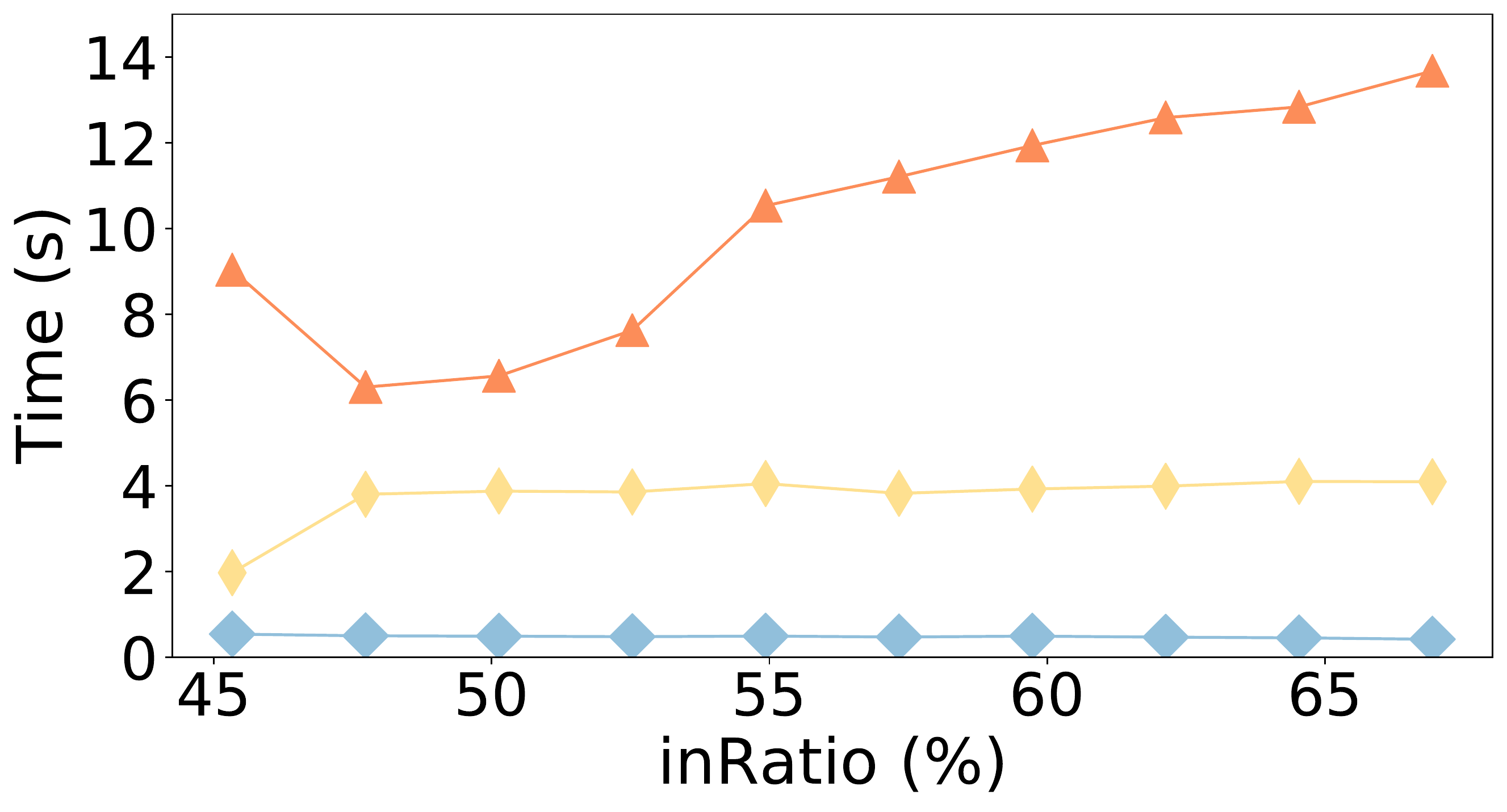} \label{fig:worst-case-ratio-2path}
}
\subfloat[$Q_{\mathsf{3-path}}$ ]{
	\includegraphics[width=0.24\linewidth]{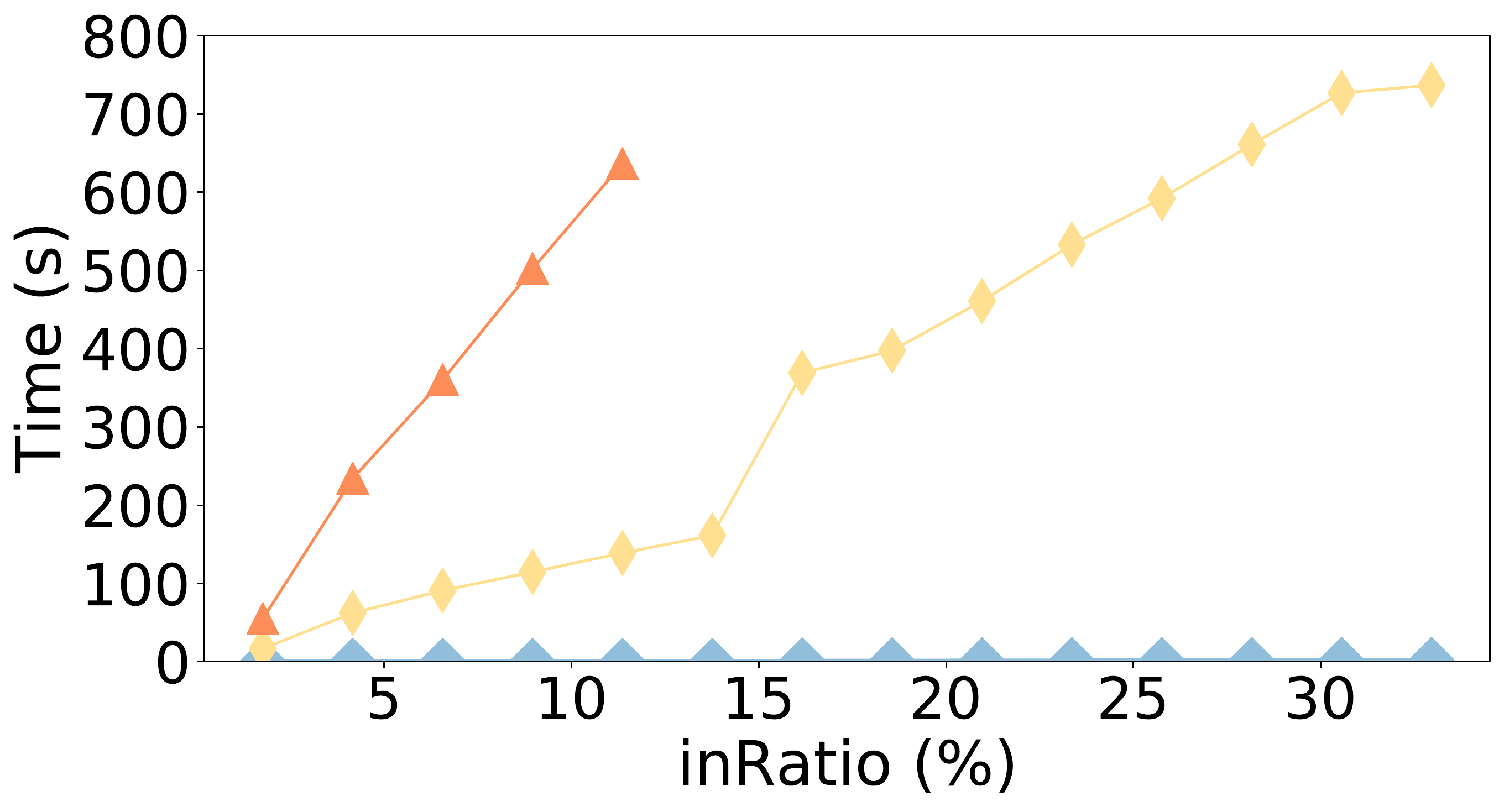} \label{fig:worst-case-ratio-3path}	
}
\caption{Performance comparison between different systems on varying relation sizes/inconsistency ratios}
\label{fig:path-size}
\vspace{-2.5ex}
\end{figure*}

\subsection{Worst-Case Study}
\label{sec:worst}
\revm{To demonstrate the robustness and efficiency of \system\ due to its theoretical guarantees}, we generate synthetic {\em worst-case\/} inconsistent database instances for the  2-path query $Q_{\mathsf{2-path}}$ and the 3-path query $Q_{\mathsf{3-path}}$:
\begin{align*}
Q_{\mathsf{2-path}}(x) & \obtainedfrom \relation{R}(\underline{x}, y), \relation{S}(\underline{y}, z). \\
Q_{\mathsf{3-path}}(x) & \obtainedfrom \relation{R}(\underline{x}, y), \relation{S}(\underline{y}, z), \relation{T}(\underline{z}, w).
\end{align*}  

We compare the performance of \system\ with ConQuer and FastFO on both queries. CAvSAT does not finish its execution on any instance within one hour, due to the long time it requires to solve the SAT formula. Thus, we do not report the time of CAvSAT.
 
%

We define a generic binary relation $\mathcal{D}(x, y, N)$ as
\begin{align*}
\mathcal{D}(x, y, N) &= ([x] \times [y]) \cup \{(u, u) \mid xy+1 \leq u \leq N, u \in \mathbb{Z}^+\},
\end{align*}
where $x, y, N \in \mathbb{Z}^+$, $[n] = \{1, 2, \dots, n\}$ and $[a] \times [b]$ denotes the cartesian product between $[a]$ and $[b]$.
To generate the input instances for $Q_{\mathsf{2-path}}$, we generate relations $\relation{R} = \mathcal{D}(a, b, N)$ and $\relation{S} = \mathcal{D}(b, c, N)$ with integer parameters $a$, $b$, $c$ and $N$. For $Q_{\mathsf{3-path}}$, we additionally generate the relation $\relation{T} = \mathcal{D}(c, d, N)$.
Intuitively, for $\relation{R}$, $[a] \times [b]$ is the set of inconsistent tuples and $\{(u, u) \mid ab+1 \leq u \leq N, u \in \mathbb{Z}^+\}$ is the set of consistent tuples. The values of $a$ and $b$ control both the number of inconsistent tuples (i.e. $ab$) and 
the size of inconsistent blocks (i.e. $b$). 
We note that $[a] \times [b]$ and $\{(u, u) \mid ab+1 \leq u \leq N, u \in \mathbb{Z}^+\}$  are disjoint. 


\introparagraph{Fixed database inconsistency with varying size.}
We perform experiments to see how robust different CQA systems are when running queries on an instance of increasing size. 
For $Q_{\mathsf{2-path}}$, we fix $b = c = 800$, and for each $k = 0, 1, \dots, 8$, we construct a database instance with $a = 120 + 460k$ and $N = (1 + k/2) \cdot 10^6$. By construction, each database instance has inconsistent block size $\mathsf{bSize} = b = c = 800$ in both relations $\relation{R}$ and $\relation{S}$, and $\mathsf{inRatio} = (ab + bc)/2N = 36.8\%$, with varying relation size $\mathsf{rSize} = N$ ranging from $1$M to $5$M. 
Similarly for $Q_{\mathsf{3-path}}$, we fix $b = c = d = 120$, and for each $k=0,1, \dots, 8$, we construct a database instance with $a = 120 + 180k$ and $N = (1 + k/2) \cdot 10^6$. 
Here the constructed database instances have $\mathsf{inRatio} = (ab + bc + cd)/3N = 1.44\%$. 
As shown in Figures~\ref{fig:worst-case-size-2path} and \ref{fig:worst-case-size-3path}, the performance of \system\ is much less sensitive to the changes of the relation sizes when compared to other CQA systems. 
We omit reporting the running time of FastFO for $Q_{\mathsf{3-path}}$ on relatively larger database instances in Figure~\ref{fig:worst-case-size-3path} for better contrast with ConQuer and \system.


\introparagraph{Fixed database sizes with varying inconsistency.}
Next, we experiment on instances of varying inconsistency ratio $\mathsf{inRatio}$ in which the joining mainly happens between inconsistent blocks of different relations. 
For $Q_{\mathsf{2-path}}$, we fix $b = c = 800$ and $N = 10^6$ and generate database instances for each $a = 100, 190, 280, \dots, 1000$. All generated database instances have inconsistent block size $\mathsf{bSize} = b = c = 800$ for both relations $\relation{R}$ and $\relation{S}$, and the size of each relation $\mathsf{rSize} = N = 10^6$ by construction. 
The inconsistency ratio $\mathsf{inRatio}$ varies from $36\%$ to $72\%$. 
For $Q_{\mathsf{3-path}}$, we fix $b = c = d = 120$ and $N = 10^6$ and generate database instances with $a = 200, 800, 1400, \dots, 8000$. The inconsistency ratio of the generated database instances varies from $1.76\%$ to $32.96\%$. 
Figures~\ref{fig:worst-case-ratio-2path} and~\ref{fig:worst-case-ratio-3path} show that \system\ is the only system whose performance is agnostic to the change of the inconsistency ratio. The running time of FastFO and Conquer increases when the input database inconsistency increases. Similar to the experiments varying relation sizes, the running times of FastFO for $Q_{\mathsf{3-path}}$ are omitted on relatively larger database instances in Figure~\ref{fig:worst-case-ratio-3path} for better contrast with ConQuer and \system.

\vspace{-2ex}
\section{Conclusion }

In this paper, we introduce the notion of a pair-pruning join tree (PPJT) and show that if a BCQ has a PPJT, then $\cqa{q}$ is in \FO\ and solvable in linear time in the size of the inconsistent database. We implement this idea in a system called \system\ that produces a SQL query to compute the consistent answers of $q$. Our experiments show that \system\ produces efficient rewritings, is scalable, and robust on worst case instances.

\reva{An interesting open question is whether CQA is in linear time for \emph{all} acyclic self-join-free SPJ queries with an acyclic attack graph, including those that do not admit a PPJT.
It would also be interesting to study the notion of PPJT for non-acyclic SPJ queries. }

\bibliographystyle{ACM-Reference-Format}
\bibliography{references}
\balance
\appendix

\section{Efficient construction of PPJT}
\label{sec:no-key-containment}

\begin{proposition}
\label{prop:no-key-containment}
Let $q$ be an acyclic self-join-free BCQ whose attack graph is acyclic. 
If for all two distinct atoms $F,G\in q$, neither of $\key{F}$ or $\key{G}$ is included in the other, then $q$ has a PPJT that can be constructed in quadratic time in the number of atoms in $q$.
\end{proposition}

\begin{proof}
Let $q$ be a self-join-free Boolean conjunctive query with an acyclic attack graph.
Let $\tau$ be a join tree for $q$ (thus $q$ is $\alpha$-acyclic).
Assume the following hypothesis:
\begin{quote}
\emph{Hypothesis of Disjoint Keys:}
for all atoms $G,H\in q$, $G\neq H$, we have that $\key{G}$ and $\key{H}$ are not comparable by set inclusion.
\end{quote}
We show, by induction on $\card{q}$, that $\cqa{q}$ is in linear time.
For the basis of the induction, $\card{q}=\emptyset$, it is trivial that $\cqa{q}$ is in linear time.
For the induction step, let $\card{q}\geq 1$.
Let $F$ be an unattacked atom of $q$.
Let $(\tau,F)$ be a join tree of $q$ with root $F$.
Let $F_{1},\dots,F_{n}$ be the children of~$F$ in $(\tau,F)$ with subtrees $\tau_1$, $\tau_2$, $\dots$, $\tau_n$. 

Let $i\in\{1,\dots,n\}$.
We claim that $q_{\tau_i}$ has an acyclic attack graph.
Assume for the sake of contradiction that the attack graph of~$q_{\tau_i}$ has a cycle, and therefore has a cycle of size~$2$.
Then there are $G,H\in q_{\tau_i}$ such that $G\attacks{q_{\tau_i}}H\attacks{q_{\tau_i}}G$.
From the \emph{Hypothesis of Disjoint Keys}, it follows $G\attacks{q}H\attacks{q}G$, contradicting the acyclicity of $q$'s attack graph.

We claim the following: 
\begin{equation}
\textnormal{for every $G\in q_{\tau_i}$, 
$\atomvars{G}\cap\atomvars{F}\subseteq\key{G}$.}
\end{equation}
This claim follows from the \emph{Hypothesis of Disjoint Keys} and the assumption that $F$ is unattacked in $q$'s attack graph.

It suffices to show that there is an atom $F_{i}'\in q_{\tau_i}$ (possibly $F_{i}'=F_{i}$) such that
\begin{enumerate}
\item
$F_{i}'$ is unattacked in the attack graph of $q_{\tau_i}$; and
\item
$\queryvars{q_{\tau_i}}\cap\atomvars{F}\subseteq\key{F_{i}'}$.
\end{enumerate}
We distinguish two cases:
\begin{description}
\item[Case that $F_{i}$ is unattacked in the attack graph of $q_{\tau_i}$.]
Then we can pick $F_{i}':= F_{i}$.
\item[Case that $F_{i}$ is attacked in the attack graph of $q_{\tau_i}$.]
We can assume an atom $G$ such that $G\attacks{q_{\tau_i}}F_{i}$.
Since $G\nattacks{q}F$, by the \emph{Hypothesis of Disjoint Keys}, it must be that $\atomvars{F_{i}}\cap\atomvars{F}\subseteq\key{G}$.
Then from $\queryvars{q_{\tau_i}}\cap\atomvars{F}\subseteq\atomvars{F_{i}}$, it follows $\queryvars{q_{\tau_i}}\cap\atomvars{F}\subseteq\key{G}$.
If $G$ is unattacked in the attack graph of $q_{\tau_ii}$, then we can pick $F_{i}':= G$.
Otherwise we repeat the same reasoning (with $G$ playing the role previously played by $F_{i}$).
This repetition cannot go on forever since the attack graph of $q_{\tau_i}$ is acyclic.
\end{description}
\end{proof}

We remark that $\cqa{q}$ remains solvable in linear time for certain acyclic self-join-free CQ that is \FO-rewritable but does not have a PPJT. It uses techniques from efficient query result enumeration algorithms \cite{DBLP:conf/icdt/DeepHK21,DBLP:conf/sigmod/DeepHK20}.

\begin{proposition}
\label{prop:non-ppjt-linear}
Let $q() \obtainedfrom {R}(\underline{c}, x), {S}(\underline{c}, y), {T}(\underline{x, y})$ where $c$ is a constant. Then there exists a linear-time algorithm for $\cqa{q}$.
\end{proposition}
\begin{proof}
Let $\db$ be an instance for $\cqa{q}$. We define $X = \{a \mid R(\underline{c}, a) \in \db\}$ and $Y = \{b \mid S(\underline{c}, b) \in \db\}$. It is easy to see that $\db$ is a ``yes''-instance for $\cqa{q}$ if and only if $X \times Y \subseteq T$, where $\times$ denotes the Cartesian product.

Next, consider the following algorithm that computes $X \times Y$ in linear time, exploiting that $T$ is also part of the input to $\cqa{q}$.
\begin{enumerate}
\item[1] Compute $X = \{a \mid R(\underline{c}, a) \in \db\}$ and $Y = \{b \mid S(\underline{c}, b) \in \db\}$
\item[2] \textbf{if } $|X| \cdot |Y| > |T|$ \textbf{ then }
\item[3] \hspace{0.5cm} \textbf{return }false
\item[4] \textbf{return whether } $X \times Y \subseteq T$
\end{enumerate}
Line 1 and 2 run in time $O(|R| + |S| + |T|)$. If the algorithm terminates at line 3, then the algorithm runs in linear time, or otherwise we must have $|R| \cdot |S| \leq |T|$, and the algorithm thus runs in time $O(|R| + |S| + |T| + |R|\cdot|S|) = O(|R| + |S| + |T|)$.

Note that $q$ does not have a PPJT: in $q_1() \obtainedfrom R(\underline{c}, x), T(\underline{x, y})$, $R$ attacks $T$, and in $q_2() \obtainedfrom S(\underline{c}, y), T(\underline{x, y})$, $S$ attacks $T$. 
\end{proof}

\section{Missing proofs}
\label{sec:missing-proofs}
In this section, we provide the missing proofs.

\subsection{Proof of Proposition~\ref{prop:ppjt}}
\label{proof:prop:ppjt}
\begin{proof}
Suppose, for the sake of contradiction, that the attack graph is not acyclic. Then there must be two atoms $R,S$ such that $R \attacks{q} S$ and $S \attacks{q} R$ by Lemma 3.6 of \cite{KoutrisW17}. Let $(\tau, T)$ be the PPJT for $q$, and let $(\tau', U)$ be the smallest subtree of $(\tau, T)$ that contains both $S$ and $R$ (it may be that $U=R$ or $U=S$). The first observation is that in the subquery $q_{\tau'}$ it also holds that $R$ attacks $S$ and vice versa. Moreover, since $(\tau', U)$ is the smallest possible subtree, the unique path that connects $R$ and $S$ must go through the root $U$. We now distinguish two cases:
\begin{itemize}
\item If $U = R$, then $S$ attacks the root of the subtree $q'$, a contradiction to the PPJT definition. 
\item If $U \neq R$, then the unique path from $R$ to $S$ goes through $U$. Since $R$ must attack every atom in that path by Lemma 4.9 of \cite{10.1145/2188349.2188351}, it must also attack $U$, a contradiction as well. 
\end{itemize}
The proof is now complete by the classification theorem of \cite{DBLP:conf/pods/KoutrisOW21}.
\end{proof}

\subsection{Proof of Proposition~\ref{prop:cforest}}
\label{proof:prop:cforest}
\begin{proof}
Let $q$ be a query in $\cforest$ and let $G$ be the join graph of $q$ as in Definition 6 of \cite{FUXMAN2007610}. In particular, $(i)$ the vertices of $G$ are the atoms of $q$, and $(ii)$ there is an arc from $R$ to $S$ if $R \neq S$ and there is some variable $w \in \var{S}$ such that $w \in \var{R} \setminus \key{R}$.
By the definition of $\cforest$, $G$ is a directed forest with connected components $\tau_1, \tau_2, \dots, \tau_n$, where the root atoms are $R_1, R_2, \dots, R_n$ respectively. 

\paragraph{Claim 1: each $\tau_i$ is a join tree.} Suppose for the sake of contradiction that $\tau_i$ is not a join tree. Then there exists a variable $w$ and two non-adjacent atoms $R$ and $S$ in $\tau_i$ such that $w \in \var{R}$, $w \in \var{S}$, and for any atom $T_i$ in the unique path $R - T_1 - \dots - T_k - S$, we have  $w \notin \var{T_i}$. We must have $w \in \key{R}$ and $w \in \key{S}$, or otherwise there would be an arc between $R$ and $S$, a contradiction. From the property $\cforest$, it also holds that no atom in the tree receives arcs from two different nodes. Hence, there is either an arc $(T_1, R)$ or $(T_k,S)$. Without loss of generality, assume there is an arc from $T_1$ to $R$. Then, since all nonkey-to-key joins are full, $w \in \var{T_1}$, a contradiction to our assumption. 

\medskip

\paragraph{Claim 2: the forest $\tau_1 \cup \dots \cup \tau_n$ can be extended to a join tree $\tau$ of $q$.}  To show this, we will show that $\tau_1 \cup \dots \cup \tau_n$ corresponds to a partial join tree as constructed by the GYO ear-removal algorithm. Indeed, suppose that atom $T$ is a child of atom $T'$ in $\tau_i$. Then, $T$ was an ear while constructing $\tau_i$ for $q_{\tau_i}$, with $T'$ as its witness. Recall that this means that if a variable $x$ is not exclusive in $T$, then $x \in T'$. We will show that this is a valid ear removal step for $q$ as well. Indeed, consider an exclusive variable $x$ in $T$ for $q_{\tau_i}$  that does not remain exclusive in $q$. Then, $x$ occurs in some other tree $\tau_j$. We will now use the fact that, by Lemma 2 of~\cite{FUXMAN2007610}, if $\tau_i$ and $\tau_j$ share a variable $x$, then $x$ can only appear in the root atoms $R_i$ and $R_j$. This implies that $x$ appears at the root of $\tau_i$, and hence at $T'$ as well, a contradiction. 

\medskip

We finally claim that $(\tau, R_1)$ is a PPJT for $q$. By construction, $\tau$ is a join tree. 
Next, consider any two adjacent atoms $R$ and $S$ in $\tau$ such that $R$ is a parent of $S$ in $(\tau, R_1)$. 
Let  $p$ be any connected subquery of $q$ containing $R$ and $S$. It suffices to show that $S$ does not attack $R$ in $p$. If $R$ and $S$ are both root nodes of some $\tau_i$ and $\tau_j$, we must have $\var{R} \cap \var{S} \subseteq \key{S} \subseteq S^{+,p}$, and thus 
$S$ does not attack $R$ in $p$. If $R$ and $S$ are in the same join tree $\tau_i$, since there is no arc from $S$ to $R$, all nonkeys of $S$ are not present in $R$, and thus $\var{R} \cap \var{S} = \var{R} \cap \key{S} \subseteq \key{S} \subseteq S^{+,p}$. Hence, there is no attack from $S$ to $R$ as well.
 \end{proof}
 
\subsection{Proof of Theorem~\ref{thm:main:boolean}}
\label{sec:correctness-boolean}

In this section, we prove Theorem~\ref{thm:main:boolean}.  

\begin{definition}
\label{defn:good-key}
Let $\db$ be a database instance for $\cqa{q}$ and $R(\underline{\vec{x}}, \vec{y})$ an atom in $q$. We define the good keys of $R$ with respect to query $q$ and $\db$, denoted by $\goodkey{R}{q}{\db}$, as follows: 
\begin{align*}
\goodkey{R}{q}{\db} &:= \{\vec{c} \mid \db\ \text{is a ``yes''-instance for } \cqa{q_{\vec{x} \rightarrow \vec{c}}}\}.
\end{align*}
\end{definition}

Let $q$ be a self-join-free acyclic BCQ with a PPJT $(\tau, R)$. Lemma~\ref{lemma:reification} implies that in order to solve $\cqa{q}$, it suffices to check whether $\goodkey{R}{q}{\db} \neq\emptyset$.

\begin{lemma}
\label{lemma:reification}
Let $q$ be a self-join-free acyclic BCQ with a PPJT $(\tau, R)$. Let $\db$ be a database instance for $\cqa{q}$. Then the following statements are equivalent:
\begin{enumerate}
\item $\db$ is a ``yes''-instance for $\cqa{q}$; and
\item $\goodkey{R}{q}{\db} \neq \emptyset$.
\end{enumerate}
\end{lemma}
\begin{proof}
By Lemma 4.4 in \cite{KoutrisW17}, we have that $\db$ is a ``yes''-instance for $\cqa{q}$ if and only if there exists a sequence of constant $\vec{c}$ such that $\db$ is a ``yes''-instance for $\cqa{q_{[\vec{x} \rightarrow \vec{c}]}}$, and the latter is equivelant to $\goodkey{R}{q}{\db} \neq \emptyset$ by Definition~\ref{defn:good-key}.
\end{proof}

\begin{example}
The atom $\relation{Employee}(\underline{x}, y, z)$ is unattacked in $\examplequery$. Observe that for $\text{employee\_id} = 0022$, no matter whether we choose the tuple \relation{Employee}(0022, \text{New York}, \text{New York}) or the tuple \relation{Employee}(0022, \text{Chicago}, \text{Chicago}) in a repair, the chosen tuple will join with some corresponding tuple in the \relation{Manager} and \relation{Contact} table. The query $\examplequery_{[x \rightarrow 0022]}$ will then return \texttt{True} for all repairs of database \database{Company}, and $0022 \in \goodkey{\relation{Employee}}{\examplequery}{\textbf{Company}} \neq\emptyset$. The \textbf{Company} database is then concluded to be a ``yes''-instance for $\cqa{\examplequery}$ by Lemma~\ref{lemma:reification}. 

We remark that converse also holds: if \textbf{Company} is known to be a ``yes''-instance for $\cqa{\examplequery}$, then by Lemma~\ref{lemma:reification}, the set $\goodkey{\relation{Employee}}{\examplequery}{\textbf{Company}}$ must also be nonempty.\qed 
\end{example}

\begin{lemma}
\label{lemma:pair-pruning}
     Let $R(\underline{\vec{x}}, \vec{y})$ be an atom in an acyclic self-join-free BCQ~$q$ with a PPJT $(\tau, R)$. Let $\db$ be an instance for $\cqa{q}$. For every sequence $\vec{c}$ of constants, of the same length as $\vec{x}$, the following are equivalent:
     \begin{enumerate}
          \item $\vec{c} \in \goodkey{R}{q}{\db}$; and \label{it:whole}
          \item 
          \label{it:pair-algo}
the block $R(\underline{\vec{c}},*)$ of $\db$ is non-empty and for every fact $R(\underline{\vec{c}},\vec{d})$ in $\db$, the following hold:
          \begin{enumerate}
          \item $\{R(\underline{\vec{c}},d)\}$ satisfies the BCQ $()\obtainedfrom R(\underline{\vec{x}}, \vec{y})$; and
          \label{it:self-pruning}
          \item 
          \label{it:pair-pruning}
for every child subtree $(\tau_S, S)$ of $(\tau, R)$, there exists  $\vec{s} \in \goodkey{S}{q_{\tau_S}}{\db}$ such that (i) all facts $S(\underline{\vec{s}}, \vec{t})$ agree on the joining variables in $\var{R} \cap \var{S}$ and (ii) for every fact $S(\underline{\vec{s}}, \vec{t})$ in $\db$, the pair $\{R(\underline{\vec{c}}, \vec{d})$, $S(\underline{\vec{s}}, \vec{t})\}$ satisfies the BCQ 
$()\obtainedfrom R(\underline{\vec{x}}, \vec{y}), S(\underline{\vec{u}}, \vec{v}),$ 
where $S(\underline{\vec{u}}, \vec{v})$ is the $S$-atom of $q$.
          \end{enumerate}
          
     \end{enumerate}   
\end{lemma}
\begin{proof}

We consider two directions.
\begin{description}
     \item[\ref{it:pair-algo}$\implies$\ref{it:whole}]
     Here we must have $\goodkey{S}{q_S}{\db} \neq \emptyset$ for all child node $S$ of $R$ in $\tau$. Let $r$ be any repair of $\db$ and let $R(\vec{c}, \vec{d}) \in r$. Since~\ref{it:pair-algo} holds, for every child node $S$ of $R$, there exists a fact $S(\vec{s}, \vec{d}) \in r$ with $\vec{s} \in \goodkey{S}{q_S}{\db}$ and a valuation $\mu_S$ such that 
     $R(\mu_S(\vec{x}), \mu_S(\vec{y})) = R(\vec{c}, \vec{d})$ and $S(\mu_S(\vec{u}), \mu_S(\vec{v})) = S(\vec{s}, \vec{t})$. Since $r$ is a repair of $\db$ and $\vec{s} \in \goodkey{S}{q_S}{\db}$, there exists a valuation $\xi_S$ such that 
     $\xi_S(q_S) \subseteq r$ with $\xi_S(\vec{u}) = \vec{s} = \mu_S(\vec{u})$. Note that all $\mu_S$ agree on the valuation of $\var{\vec{x}} \cup \var{\vec{y}}$, let $\mu$ be the valuation such that $R(\mu(\vec{x}), \mu(\vec{y})) = R(\mu(\vec{x_S}), \mu(\vec{y_S}))$ for all child node $S$ of $R$.

     Next we show that for all $q_S$ and any variable $z \in \var{R} \cap \var{q_S}$, $\mu(z) = \xi_S(z)$. Since $r$ is consistent, we must have 
     $S(\mu_S(\vec{u}), \mu_S(\vec{v})) = S(\xi_S(\vec{u_S}), \xi_S(\vec{v_S})) \in r.$ Since $T$ is a join tree, we must have $z \in \var{R} \cap \var{S}$, and it follows that $\xi_S(z) = \mu_S(z) = \mu(z)$, as desired.

     Then, the following valuation 
     $$\mu(z) = \begin{cases} \mu(z) & z \in \var{R} \\ \xi_i(z) & z \in \var{q_S} \setminus \var{R} \\ d & z = d \text{ is constant}\end{cases}$$ is well-defined and satisfies that 
     $\mu(q_{\vec{x} \rightarrow \vec{c}}) \subseteq r,$ as desired.

     \item[\ref{it:whole}$\implies$\ref{it:pair-algo}]
     By contraposition. Assume that~\ref{it:pair-algo} does not hold, and we show that there exists a repair $r$ of $\db$ that does not satisfy $q_{[\vec{x} \rightarrow \vec{c}]}$. 

     If~\ref{it:self-pruning} does not hold, then there exists some fact $f = R(\vec{c}, \vec{d})$ that does not satisfy $R(\underline{\vec{x}}, \vec{y})$, and any repair containing the fact $f$ does not satisfy $q_{[\vec{x} \rightarrow \vec{c}]}$. Next we assume that~\ref{it:self-pruning} holds but~\ref{it:pair-pruning} does not.

     If $\goodkey{S}{q_S}{\db} = \emptyset$ for some child node $S$ of $R$ in $\tau$, then by monotonicity of conjunctive queries and Lemma~\ref{lemma:reification}, $\db$ is a ``no''-instance for $\cqa{q_S}$, $\cqa{q}$ and thus $\cqa{q_{[\vec{x} \rightarrow \vec{c}]}}$. In what follows we assume that $\goodkey{S}{q_S}{\db} \neq \emptyset$ for all child node $S$ of $R$. 

     Since~\ref{it:pair-pruning} does not hold, there exist a fact $R(\vec{c}, \vec{d})$ and some child node $S$ of $R$ in $\tau$ and query $q_S$ such that for any block $S(\vec{s}, *)$ with $\vec{s} \in \goodkey{S}{q_S}{\db}$, there exists a fact $S(\vec{s}, \vec{t})$ that does not join with $R(\vec{c}, \vec{d})$.

     Let $\db' = \db \setminus R \setminus \{S(\vec{s}, *) \mid \vec{s} \in \goodkey{S}{q_S}{\db}\}$. We show that $\db'$ is a ``no''-instance for $\cqa{q_S}$. Indeed, suppose otherwise that $\db'$ is a ``yes''-instance for $\cqa{q_S}$, then there exists some $\vec{s}$ such that $\db'$ is a ``yes''-instance for $\cqa{q_{S, [\vec{u} \rightarrow \vec{s}]}}$. Note that by construction, $\vec{s} \notin \goodkey{S}{q_S}{\db}$. Since $\db' \subseteq \db$, we have $\db$ is a ``yes''-instance for $\cqa{q_{S, [\vec{u} \rightarrow \vec{s}]}}$, implying that $\vec{s} \in \goodkey{S}{q_S}{\db}$, a contradiction.

     Consider the following repair $r$ of $\db$ that contains

     \begin{itemize}
          \item $R(\vec{c}, \vec{d})$ and an arbitrary fact from all blocks $R(\vec{b}, *)$ with $\vec{b} \neq \vec{c}$;
          \item for each $\vec{s} \in \goodkey{S}{q_S}{\db}$, any fact $S(\vec{s}, \vec{t})$ that does not join with $R(\vec{c}, \vec{d})$; and 
          \item any falsifying repair $r'$ of $\db'$ for $\cqa{q_S}$.
     \end{itemize}

     We show that $r$ does not satisfy $q_{[\vec{x} \rightarrow \vec{c}]}$. Suppose for contradiction that there exists a valuation $\mu$ with $\mu(q_{[\vec{x} \rightarrow \vec{c}]}) \subseteq r$ and $R(\mu(\vec{x}), \mu(\vec{y})) = R(\vec{c}, \vec{d}) \in r$. Let $S(\vec{s^*}, \vec{t^*}) = S(\mu(\vec{u}), \mu(\vec{v}))$, then we must have 
     $\vec{s^*} \notin \goodkey{S}{q_S}{\db}$, since otherwise we would have $S(\vec{s^*}, \vec{t^*})$ joining with $R(\vec{c}, \vec{d})$ where we have $\vec{s^*} \in \goodkey{S}{q_S}{\db}$, a contradiction to the construction of $r$. Since $\vec{s^*} \notin \goodkey{S}{q_S}{\db}$, we would then have $\mu(q_S) \subseteq r'$, a contradiction to that $r'$ is a falsifying repair of $\db'$ for $\cqa{q_S}$. Finally, if~\ref{it:pair-pruning} holds, then all facts $S(\vec{s}, \vec{t})$ must agree on $\vec{w}$ since they all join with the same fact $R(\vec{c} ,\vec{d})$. 
\end{description}   

The proof is now complete. 
\end{proof}

\begin{proof}[Proof of Theorem~\ref{thm:main:boolean}]
It suffices to present rewriting rules to compute each $\goodkey{R}{q}{\db}$ for each atom $R$ in $q$ by Lemma~\ref{lemma:pair-pruning} , and show that these rewriting rules are equivalent to those presented in Section~\ref{sec:datalog}, which are shown to run in linear time. We denote $R_{\mathsf{gk}}$ as the Datalog predicate for $\goodkey{R}{q}{\db}$. 
It is easy to see that \textbf{Rule 1} computes all blocks of $R$ violating item~\ref{it:self-pruning} of Lemma~\ref{lemma:pair-pruning}. 

To compute the blocks of $R$ violating item~\ref{it:pair-pruning}, we denote $R_{\mathsf{gki}}$ as the predicate for the subset of $R_{\mathsf{gk}}$ that satisfies condition (i) of item~\ref{it:pair-pruning}.
For each child $S(\underline{\vec{u}}, \vec{v})$ of $R$ in a PPJT, let $\vec{w}$ be a sequence of all variables in $\var{R} \cap \var{S}$. The following rules find all blocks in $R$ that violate condition (ii) of item~\ref{it:pair-pruning}.
\begin{align}
\label{rule:join}
S_{\mathsf{join}}(\vec{w}) & \obtainedfrom S(\underline{\vec{u}}, \vec{v}), S_{\mathsf{gki}}(\vec{u}). \\ 
\label{rule:pp}
\fk{R}(\vec{x}) & \obtainedfrom R(\vec{x}, \vec{y}) \lnot S_{\mathsf{join}}(\vec{w}).
\end{align}
We then compute the predicate $R_{\mathsf{gk}}$ denoting $\goodkey{R}{q}{\db}$ with
\begin{align}
\label{rule:gk}
R_{\mathsf{gk}}(\vec{x}) &\obtainedfrom R(\underline{x}, \vec{y}), \lnot \fk{R}(\vec{x}).
\end{align}
If $R$ has a parent, then we may compute all blocks in $R$ violating condition (i) of item~\ref{it:pair-pruning} using the following rules for every variable at the $i$-th position of $\vec{y}$,
\begin{align}
\label{rule:gkinot}
R_{\mathsf{gk}^{\lnot}}(\vec{x}) &\obtainedfrom R(\vec{x}, \vec{y}), R(\vec{x}, \vec{y}'), y_i \neq y_i'. \\ 
\label{rule:gki}
R_{\mathsf{gki}}(\vec{x}) & \obtainedfrom R_{\mathsf{gk}}(\vec{x}), \lnot R_{\mathsf{gk}^{\lnot}}(\vec{x}). 
\end{align}
Now we explain why Rules~(\ref{rule:join}) through~(\ref{rule:gki}) are equivalent to \textbf{Rule 2, 3, 4}.
First, the head $R_{\mathsf{gk}^{\lnot}}$ of Rule~(\ref{rule:gkinot}) can be safely renamed to $\fk{R}$ and yield \textbf{Rule 2}. Rule~(\ref{rule:pp}) is equivalent to \textbf{Rule 3}. Finally, Rules~(\ref{rule:join}),~(\ref{rule:gk}) and~(\ref{rule:gki}) can be merged to \textbf{Rule 4} since to compute each $S_{\mathsf{join}}$, we only need $S(\vec{u}, \vec{v})$ and $\fk{S}$. Note that the Rule~(\ref{rule:join}) for the atom $R$ will have $0$ arity if $R$ is the root of the PPJT. 
\end{proof}

\subsection{Proof of Lemma~\ref{lemma:parameterization}}
\label{sec:lemma:parameterization}
\begin{proof}

Consider both directions. First we assume that $\vec{c}$ is a consistent answer of $q$ on $\db$. Let $r$ be any repair of $\db$. Then there exists a valuation $\mu$ with $\mu(q) \subseteq r$ with $\mu(\vec{u}) = \vec{c}$, and hence $\mu(q_{[\vec{u} \rightarrow \vec{c}]}) = \mu(q) \subseteq r$. That is, $q_{[\vec{u} \rightarrow \vec{c}]}(r)$ is true. Hence $\db$ is a ``yes''-instance for $\cqa{q_{[\vec{u} \rightarrow \vec{c}]}}$. For the other direction, we assume that $\db$ is ``yes''-instance for $\cqa{q_{[\vec{u} \rightarrow \vec{c}]}}$. Let $r$ be any repair of $\db$. Then there is a valuation $\mu$ with $\mu(q_{[\vec{u} \rightarrow \vec{c}]}) \subseteq r$. Let $\theta$ be the valuation with $\theta(\vec{u}) = \vec{c}$. Consider the valuation 
$$\mu^+(z)  = \begin{cases}
     \theta(z) & z \in \var{\vec{u}} \\
     \mu(z) & \text{otherwise,}
\end{cases}$$
and we have $\mu^+(q) = \mu(q_{[\vec{u} \rightarrow \vec{c}]}) \subseteq r$ with $\mu^+(\vec{u}) = \theta(\vec{u}) = \vec{c}$, as desired.
\end{proof}

\subsection{Proof of Theorem~\ref{thm:main:full}}
\label{proof:thm:main:full}
\begin{proof}

Our algorithm first evaluates $q$ on $\db$ to yield a set $S$ of size $|\mathsf{OUT}_{p}|$ in time $O(N \cdot |\mathsf{OUT}_p|)$.  
Here the set $S$ must contain all the consistent answers of $q$ on $\db$. By Lemma~\ref{lemma:parameterization}, we then return all answers $\vec{c} \in S$ such that $\db$ is a ``yes''-instance for $\cqa{q_{[\vec{u} \rightarrow \vec{c}]}}$, which runs in $O(N)$ by Theorem~\ref{thm:main:boolean}. This approach gives an algorithm with running time $O(N \cdot |\mathsf{OUT}_p|)$.

If $q$ is full, there is an algorithm that computes the set of consistent answers even faster. The algorithm proceeds by (i) removing all blocks with at least two tuples from $\db$ to yield $\db^c$ and (ii) evaluating $q$ on $\db^c$. 
It suffices to show that every consistent answer to $q$ on $\db$ is an answer to $q$ on $\db^c$. 
Assume that $\vec{c}$ is a consistent answer to $q$ on $\db$. Consider $q_{[\vec{u} \rightarrow \vec{c}]}$, a disconnected CQ where $\vec{u}$ is a sequence of all variables in $q$. Its \FO-rewriting effectively contains \textbf{Rule 1} for each atom in $q_{[\vec{u} \rightarrow \vec{c}]}$, which is equivalent to step (i), and then checks whether $\db^c$ satisfies $q_{[\vec{u} \rightarrow \vec{c}]}$ by Lemma~B.1 of~\cite{DBLP:conf/pods/KoutrisOW21}. By Lemma~\ref{lemma:parameterization}, $\db$ is a ``yes''-instance for $\cqa{q_{[\vec{u} \rightarrow \vec{c}]}}$, and thus the \FO-rewriting concludes that $\db^c$ satisfies $q_{[\vec{u} \rightarrow \vec{c}]}$, i.e.\ $\vec{c}$ is an answer to $q$ on $\db^c$.
In our algorithm, step (i) runs in $O(N)$ and since $q$ is full, step (ii) runs in time $O(N + |\mathsf{OUT}_c|)$.
\end{proof}

\section{Improvement upon existing systems}
\label{sec:sys-improvements}
\subsection{Conquer}
Fuxman and Miller~\cite{FUXMAN2007610} identified $\cforest$, a class of CQs whose consistent answers can be computed via an \FO-rewriting. However, their accompanying system can only handle queries in $\cforest$ whose join graph is a tree, unable to handle the query in $\cforest$ whose join graph is not connected \cite{fuxman2005conquer}. Since we were unable to find the original ConQuer implementation, we re-implemented ConQuer and added an efficient implementation of the method \texttt{RewriteConsistent} in Figure~2 of \cite{FUXMAN2007610}, enabling us to produce the consistent  SQL rewriting for every query in $\cforest$.

\subsection{Conquesto}

Conquesto \cite{Conquesto} produces a non-recursive Datalog program that implements the algorithm in \cite{KoutrisW17}, targetting all \FO-rewritable self-join-free CQs. However, it suffers from repeated computation and unnecessary cartesian products. 
For example, the Conquesto rewriting for the CQ $q(z) \obtainedfrom R_1(\underline{x}, y, z), R_2(\underline{y}, v, w)$ is shown as follows, where Rule~(\ref{rule:common-predicate1}) and~(\ref{rule:common-predicate2}) share the common predicate $\mathsf{R_2}(y, v, w)$ in their bodies, resulting in re-computation, and Rule~(\ref{conquesto:bad_block}) involves a Cartesian product.
\begin{align}
\mathsf{Sr_{R_2}}(y) \obtainedfrom&  \mathsf{R_2}(y, v, w). \label{rule:common-predicate1}\\ 
\mathsf{Yes_{R_2}}(y) \obtainedfrom & \mathsf{Sr_{R_2}}(y), \mathsf{R_2}(y, v, w). \\ 
\mathsf{\mathsf{Sr_{R_1}}}(z) \obtainedfrom & \mathsf{R_1}(x, y, z), \mathsf{R_2}(y, v, w). \label{rule:common-predicate2}\\ 
\mathsf{Gf_{R_1}}(v_2, x, y, z) \obtainedfrom & \mathsf{\mathsf{Sr_{R_1}}}(z), \mathsf{R_1}(x, y, v_2), \mathsf{Yes_{R_2}}(y), v_2 = z. \\
\mathsf{Bb_{R_1}}(x, z) \obtainedfrom & \mathsf{\mathsf{Sr_{R_1}}}(z), \mathsf{R_1}(x, y, v_2), \lnot \mathsf{Gf_{R_1}}(v_2, x, y, z). \label{conquesto:bad_block} \\
\mathsf{Yes_{R_1}}(z) \obtainedfrom & \mathsf{\mathsf{Sr_{R_1}}}(z), \mathsf{R_1}(x, y, z), \lnot \mathsf{Bb_{R_1}}(x, z). \label{conquesto:yes_rule}
\end{align}

We thus implement FastFO to address the aforementioned issues, incorporating our ideas in Subsection~\ref{sec:ground-atom}. Instead of re-computing the \emph{local safe ranges} such as $\mathsf{Sr_{R_1}}(y)$ and $\mathsf{\mathsf{Sr_{R_2}}}(z)$, we compute a \emph{global safe range} $\mathsf{Sr}(x, y, z)$, which includes all key variables from all atoms and the free variables. This removes all undesired Cartesian products and the recomputations of the local safe ranges at once. The FastFO rewriting for $q$ is presented below.
\begin{align}
\mathsf{Sr}(x, y, z) \obtainedfrom & \mathsf{R_1}(x, y, z), \mathsf{R_2}(y, v, w). \\ 
\mathsf{Yes_{R_2}}(y) \obtainedfrom & \mathsf{Sr}(x, y, z), \mathsf{R_2}(y, v, w). \\
\mathsf{Gf_{R_1}}(v_2, x, y, z) \obtainedfrom & \mathsf{Sr}(x, \_, z), \mathsf{R_1}(x, y, v_2), \mathsf{Yes_{R_2}}(y), v_2 = z. \\
\mathsf{Bb_{R_1}}(x, z) \obtainedfrom & \mathsf{Sr}(x, \_, z), \mathsf{R_1}(x, y, v_2), \lnot \mathsf{Gf_{R_1}}(v_2, x, y, z). \\
\mathsf{Yes_{R_1}}(z) \obtainedfrom & \mathsf{Sr}(x, y, z), \mathsf{R_1}(x, y, z), \lnot \mathsf{Bb_{R_1}}(x, z).  \label{fastfo:yes_rule}
\end{align}

For evaluation, the rules computing each intermediate relation (i.e.\ all rules except for the one computing $\mathsf{Yes_{R_1}}(z)$) are then translated to a SQL subquery via a \lstinline!WITH! clause.

\section{Additional Tables and Figures}
\begin{table*}
\footnotesize
\caption{A summary of the Stackoverflow Dataset.}
\label{tbl:stackoverflow-schema-full}
\begin{tabular}{l c c c p{11.8cm}}
Table & \# of rows & $\mathsf{inRatio}$ & max.\ $\mathsf{bSize}$ & Attributes \\
\hline
Users & 14,839,627 & 0\% & 1 & \underline{Id}, AboutMe, Age, CreationDate, DisplayName, DownVotes, EmailHash, LastAccessDate, Location, Reputation, UpVotes, Views, WebsiteUrl, AccountId \\
Posts & 53,086,328 & 0\% & 1 & \underline{Id}, AcceptedAnswerId, AnswerCount, Body, ClosedDate, CommentCount, CommunityOwnedDate, CreationDate, FavoriteCount, LastActivityDate, LastEditDate, LastEditorDisplayName, LastEditorUserId, OwnerUserId, ParentId, PostTypeId, Score, Tags, Title, ViewCount \\
PostLinks & 7,499,403 & 0\% & 1 & Id, CreationDate, \underline{PostId}, \underline{RelatedPostId}, \underline{LinkTypeId} \\
PostHistory & 141,277,451 & 0.001\% & 4 & Id, \underline{PostHistoryTypeId}, \underline{PostId}, RevisionGUID, \underline{CreationDate}, \underline{UserId}, UserDisplayName, Comment, Text \\
Comments & 80,673,644 & 0.0012\% & 7 & Id, \underline{CreationDate}, PostId, Score, Text, \underline{UserId} \\
Badges  & 40,338,942 & 0.58\% & 941 & Id, \underline{Name}, \underline{UserId}, \underline{Date} \\
Votes & 213,555,899 & 30.9\% & 1441 &  Id, \underline{PostId}, \underline{UserId}, BountyAmount, VoteTypeId, \underline{CreationDate}
\end{tabular}
\end{table*} 

\label{sec:additional-figures}
Figure~\ref{fig:syn-incon-rest} presents the performance of queries $q_2$, $q_4$ and $q_6$ on varying inconsistency ratios, supplementing Figure~\ref{fig:syn-incon}. Figure~\ref{fig:syn-block} summarizes the performance of all seven synthetic queries on varying block sizes.

\begin{figure*}[t]
\includegraphics[scale=0.25]{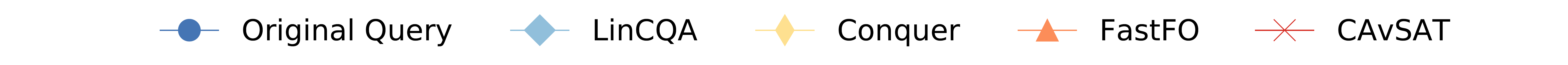}
\vspace{-3ex}  
\end{figure*}
\begin{figure*}[t]
\subfloat[$q_2$]{
     \includegraphics[width=0.25\textwidth]{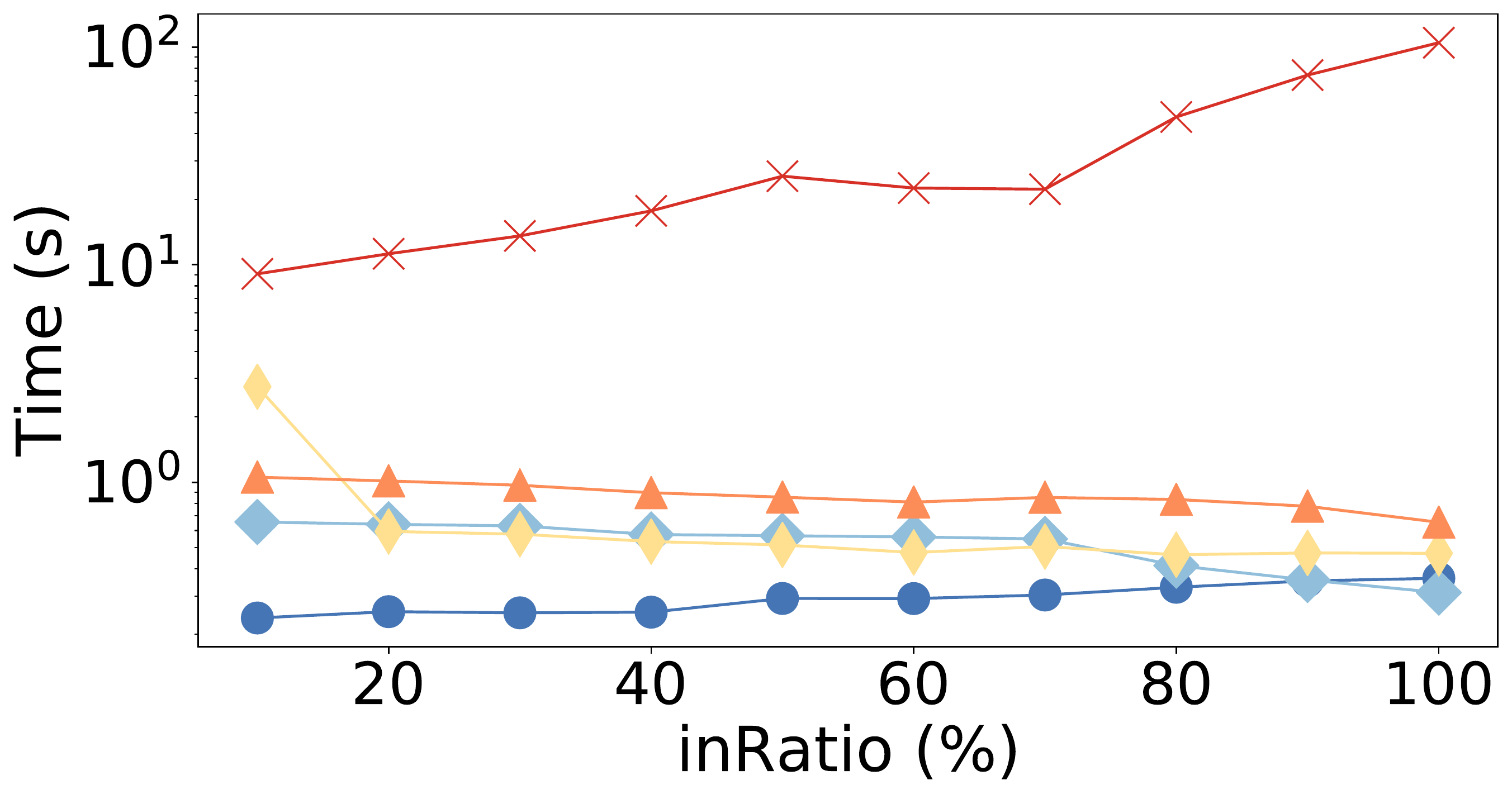} 
}
\subfloat[$q_4$]{
     \includegraphics[width=0.25\textwidth]{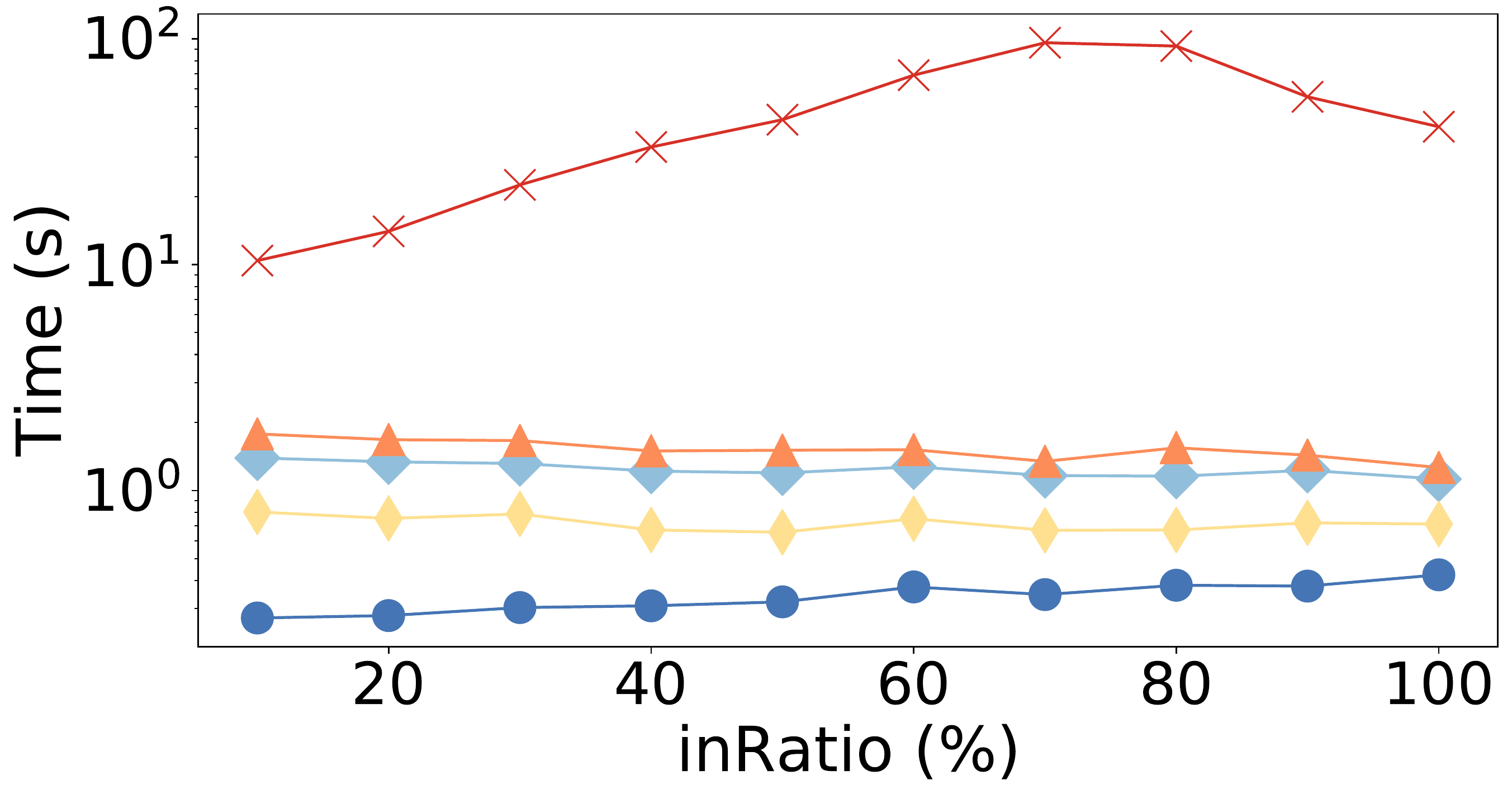} 
}
\subfloat[$q_6$]{
     \includegraphics[width=0.25\textwidth]{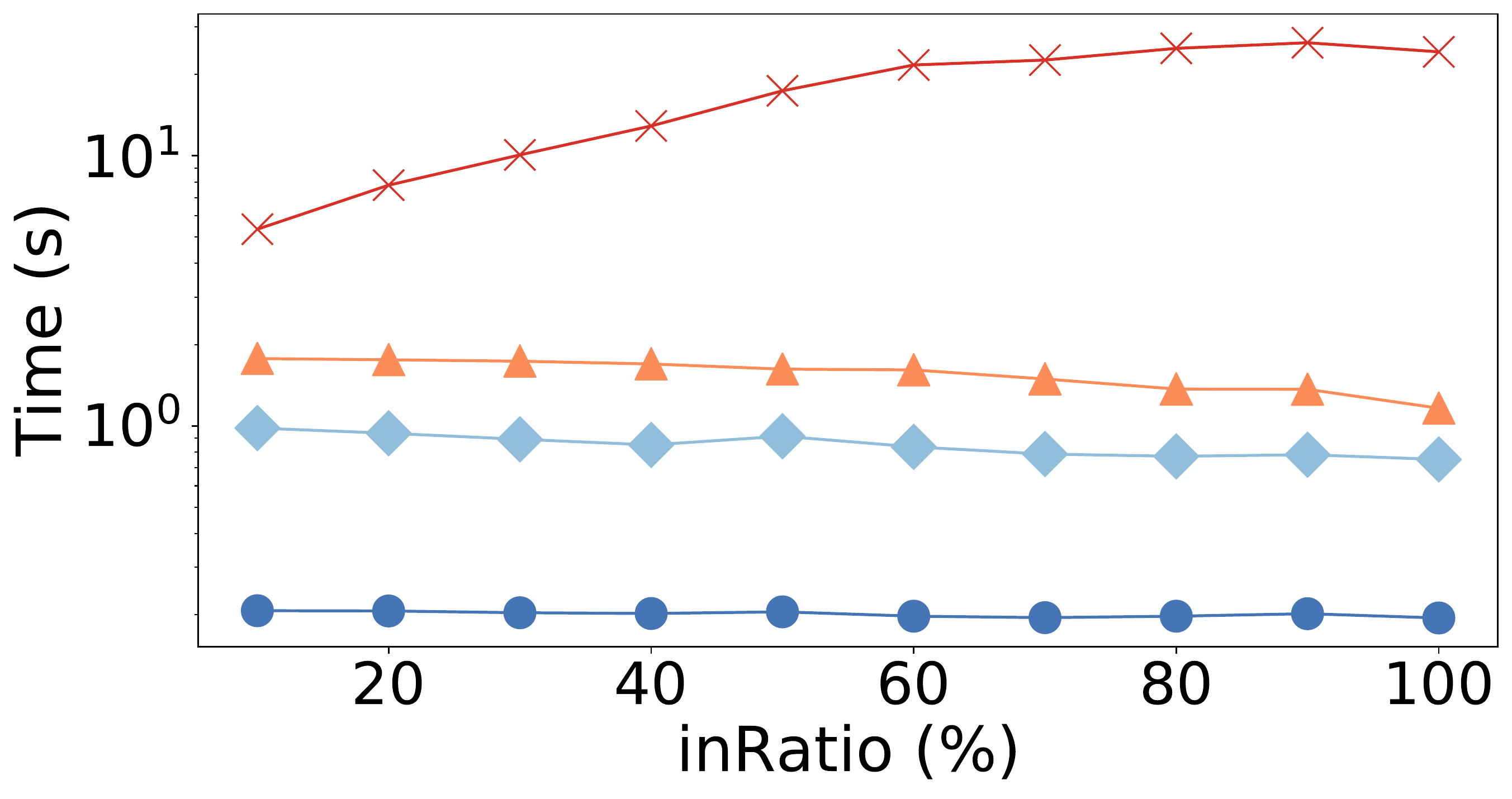} 
}
\caption{Performance of different systems on inconsistent databsae of varying inconsistency ratio}
\label{fig:syn-incon-rest}
\end{figure*}

\begin{figure*}[t]
\subfloat[$q_1$]{
     \includegraphics[width=0.25\textwidth]{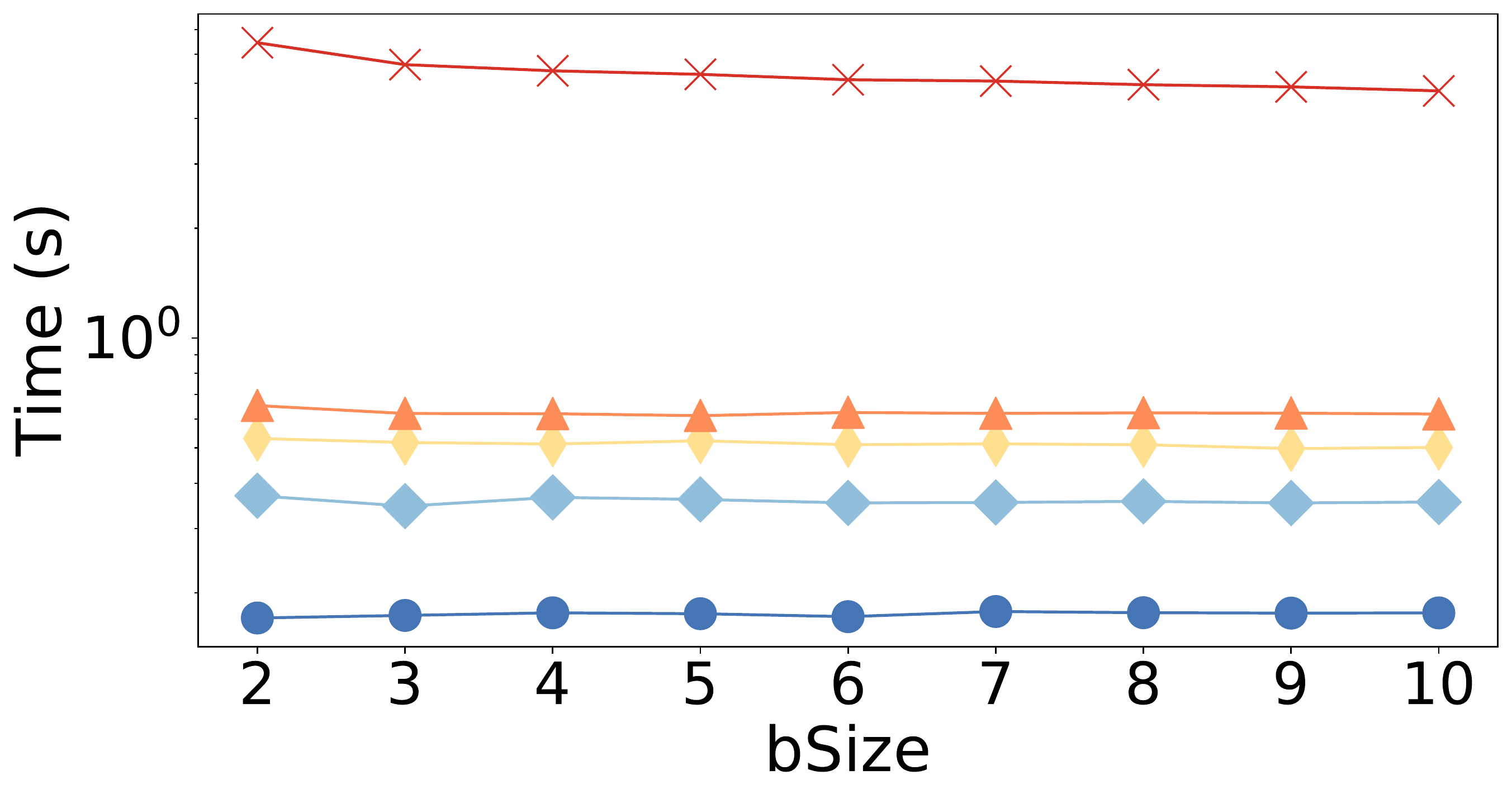} 
}
\subfloat[$q_2$]{
     \includegraphics[width=0.25\textwidth]{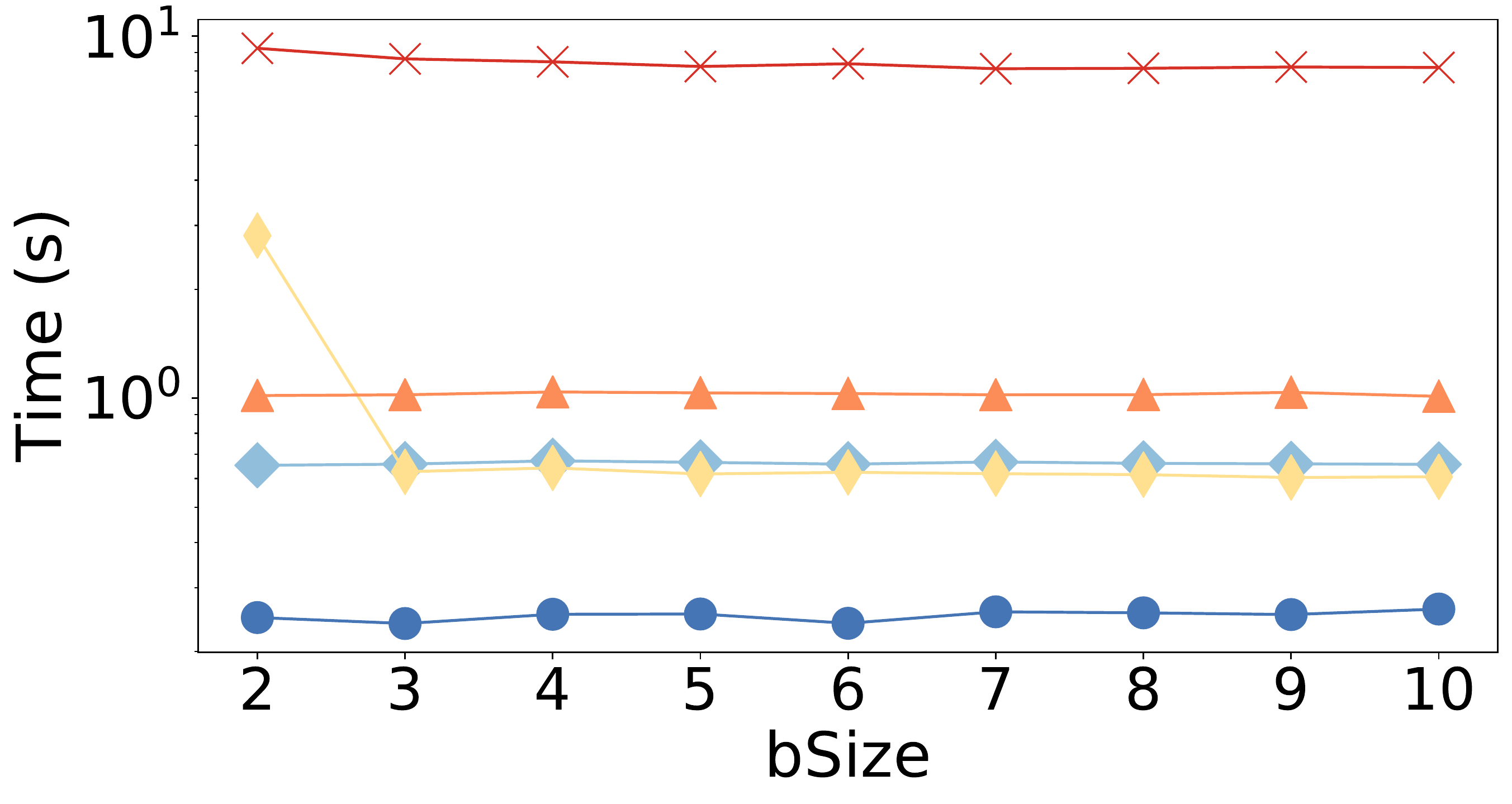} 
}
\subfloat[$q_3$]{
     \includegraphics[width=0.25\textwidth]{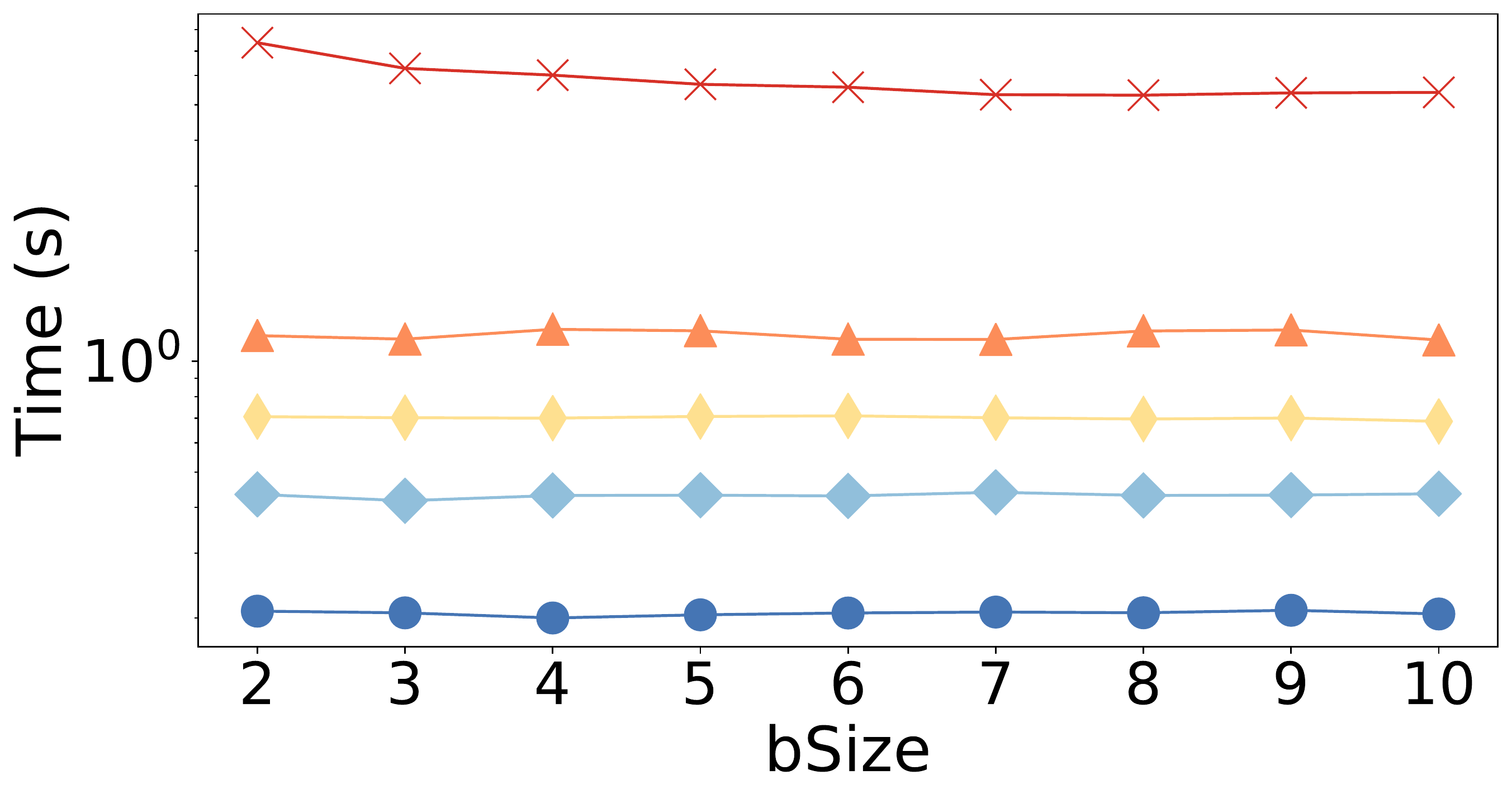} 
}
\subfloat[$q_4$]{
     \includegraphics[width=0.25\textwidth]{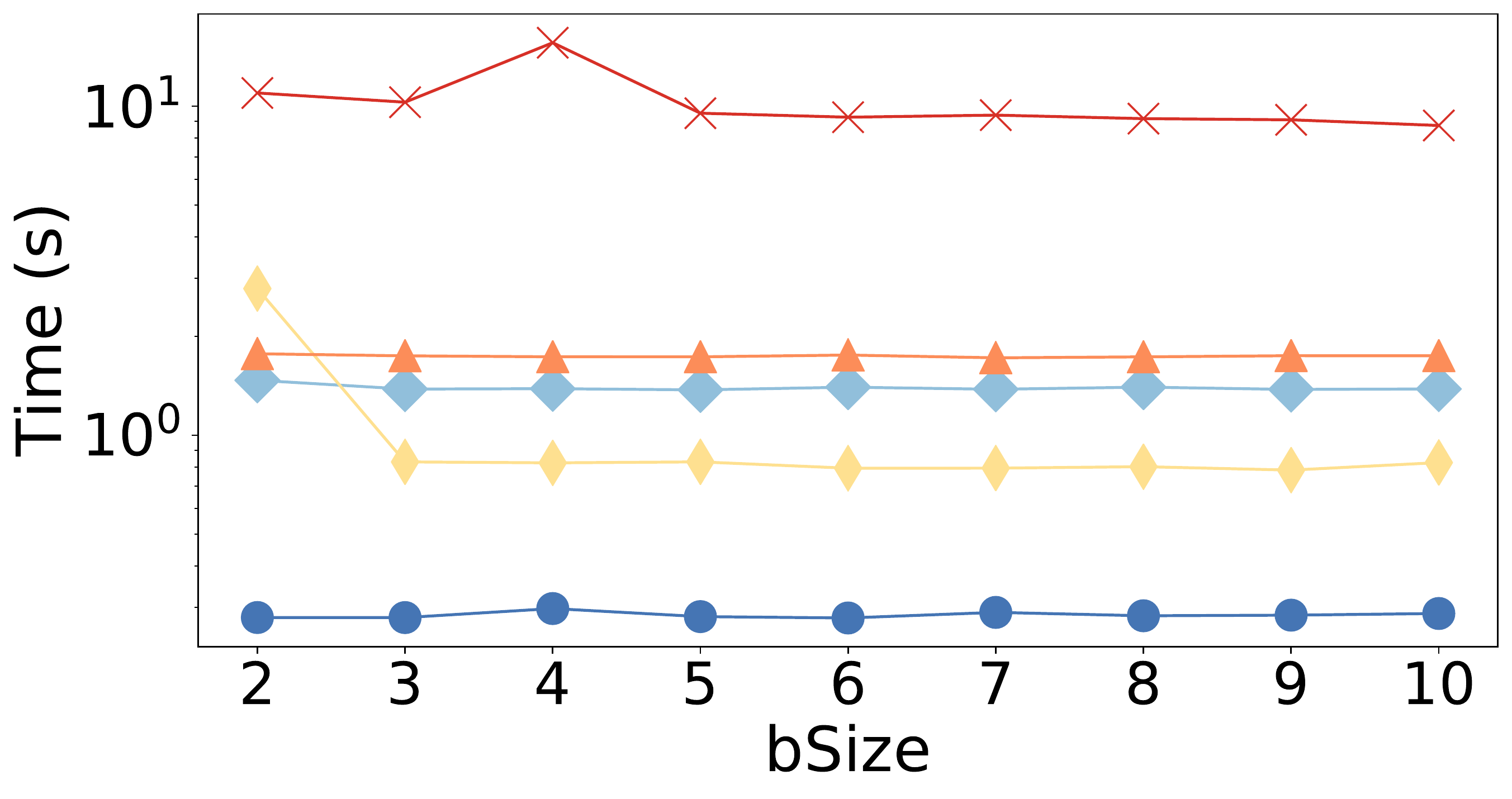} 
}
\quad
\subfloat[$q_5$]{
     \includegraphics[width=0.25\textwidth]{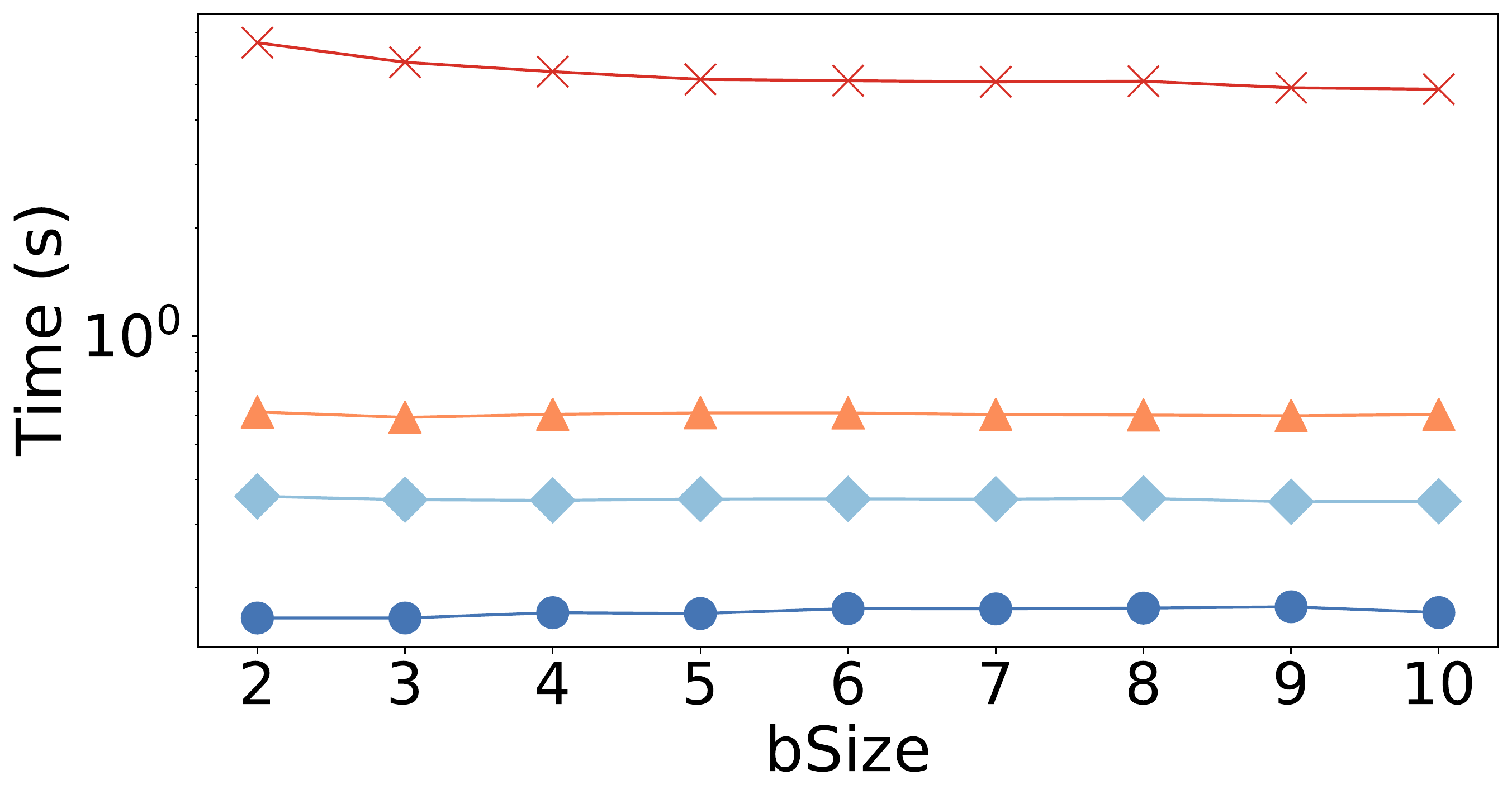} 
}
\subfloat[$q_6$]{
     \includegraphics[width=0.25\textwidth]{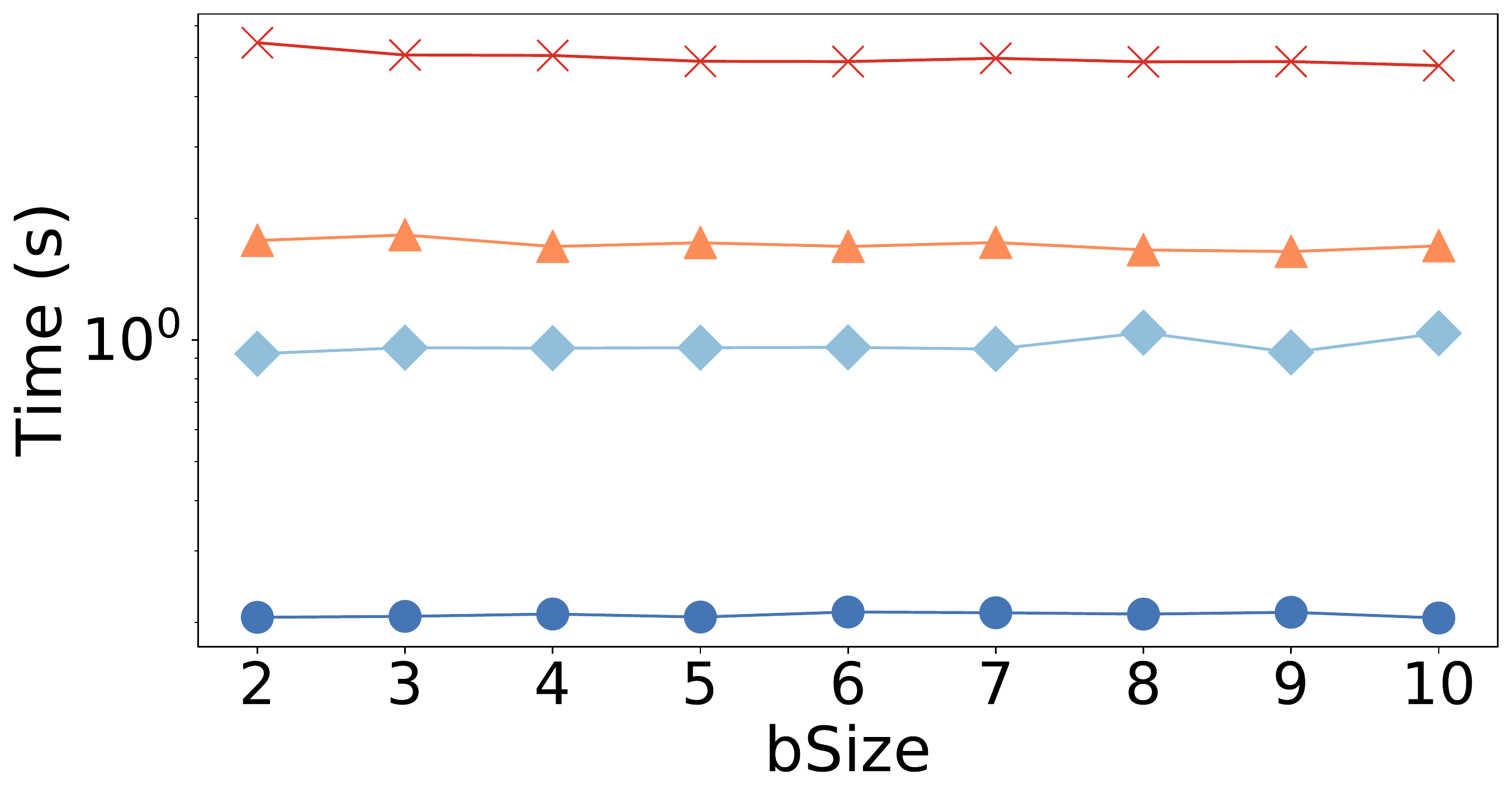} 
}
\subfloat[$q_7$]{
     \includegraphics[width=0.25\textwidth]{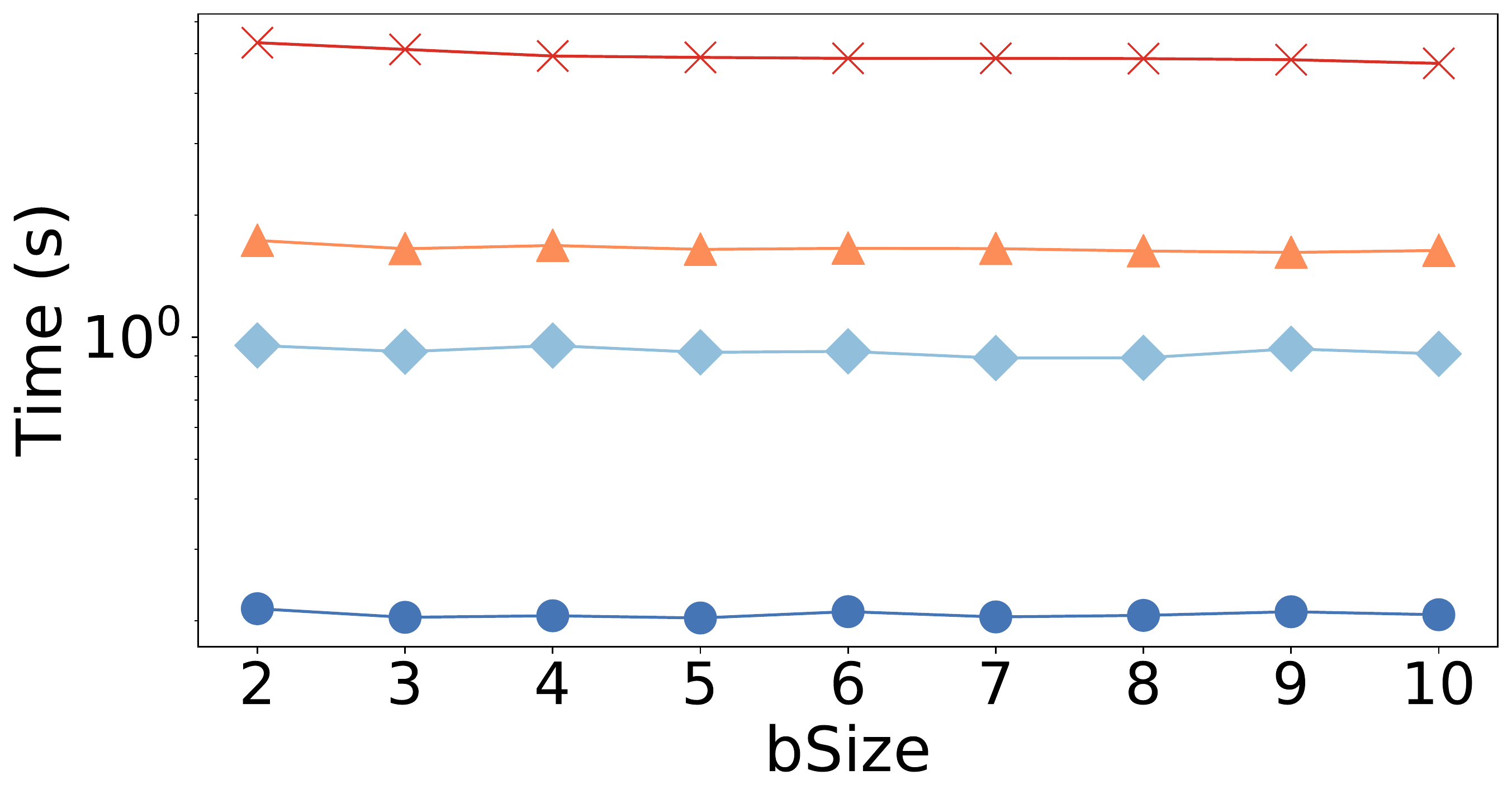} 
}
\caption{Performance of different systems on inconsistent database of varying block size}
\label{fig:syn-block}
\end{figure*}


\end{document}